\documentclass[11pt,fleqn,a4paper]{article} 

\usepackage{amsmath,amssymb,amsthm,amsfonts} 
\usepackage[mathscr]{eucal}

\usepackage{hyperref}
\hypersetup{colorlinks, linkcolor=blue, citecolor=blue, urlcolor=blue}

\allowdisplaybreaks

\newcommand{\p}{\partial}
\newcommand{\sgn}{\mathop{\rm sgn}\nolimits}
\newcommand{\const}{{\rm const}}
\newcommand{\ve}{\varepsilon}
\newcommand{\de}{\delta}

\newlength{\mylength}
\newcommand{\solution}{\hspace*{-\mylength}\bullet\quad}

\newtheorem{theorem}{Theorem}
\newtheorem{lemma}[theorem]{Lemma}
\newtheorem{corollary}[theorem]{Corollary}
\newtheorem{proposition}[theorem]{Proposition}
{\theoremstyle{definition}

\newtheorem{remark}[theorem]{Remark}
}

\newcommand{\todo}[1][\null]{\ensuremath{\clubsuit}}

\newcommand{\noprint}[1]{}

\begin{document}

\par\noindent {\LARGE\bf
Point-symmetry pseudogroup, Lie reductions\\ and exact solutions of Boiti--Leon--Pempinelli system
\par}

\vspace{5mm}\par\noindent{\large
Diana S.\ Maltseva$^\dag$ and Roman O.\ Popovych$^\ddag$
}

\vspace{5mm}\par\noindent{\it\small
$^\dag$\,Faculty of Mechanics and Mathematics, Taras Shevchenko National University of Kyiv,\\
$\phantom{^\dag}$\,4-e Hlushkova Ave., 03127, \,Kyiv, Ukraine
\par}

\vspace{2mm}\par\noindent{\it\small
$^\ddag$\,Mathematical Institute, Silesian University in Opava, Na Rybn\'\i{}\v{c}ku 1, 746 01 Opava, Czech Republic
\\
$\phantom{^\ddag}$Institute of Mathematics of NAS of Ukraine, 3 Tereshchenkivska Str., 01024 Kyiv, Ukraine
\\
$\phantom{^\ddag}$Fakult\"at f\"ur Mathematik, Universit\"at Wien, Oskar-Morgenstern-Platz 1, A-1090 Wien, Austria
\par}

\vspace{4mm}\par\noindent{\small
E-mail:
rop@imath.kiev.ua
\par}

\vspace{7mm}\par\noindent\hspace*{10mm}\parbox{140mm}{\small
We carry out extended symmetry analysis of the (1+2)-dimensional Boiti--Leon--Pempinelli system,
which corrects, enhances and generalizes many results existing in the literature.
The point-symmetry pseudogroup of this system is computed
using an original megaideal-based version of the algebraic method.
A number of meticulously selected differential constraints allow us to construct families of exact solutions of this system,
which are significantly larger than all known ones.
After classifying one- and two-dimensional subalgebras
of the entire (infinite-dimensional) maximal Lie invariance algebra of this system,
we study only its essential Lie reductions, which give solutions beyond the above solution families.
Among reductions of the Boiti--Leon--Pempinelli system
using differential constraints or Lie symmetries,
we identify a number of famous partial and ordinary differential equations.
We also show how all the constructed solution families can significantly be extended
by Laplace and Darboux transformations.
}\par\vspace{5mm}

\noprint{


}

\section{Introduction}

Inspired by the inverse scattering method applied in~\cite{gard1967a} for determining integrability of the Korteweg--de Vries equation,
the theory of integrable systems was rapidly developed.
It turned out that the cubic Schr\"odinger equation, the sine-Gordon equation, the Dym equation, the Camassa--Holm equation,
the hierarchies related to these equations, and many other (1+1)-dimensional models are integrable.
At the same time, even the question whether there exist integrable systems of partial differential equations
with more than two independent variables was open.
This is why finding new integrable systems at least in dimension 1+2,
studying them within the framework of integrability theory
and integrating them by various methods specific for this theory, like the inverse scattering transform,
became urgent.
In a few years, the integrability of several existing multidimensional systems was revealed
and a number of new integrable multidimensional systems were constructed,
which includes the Kadomtsev--Petviashvili equation, the Davey--Stewartson equation and the Nizhnik equation.

One of the first (1+2)-dimensional integrable systems, now commonly known as the Boiti--Leon--Pempinelli system,
\begin{equation}\label{eq:BLPsystem}
\begin{split}
&u_{ty}=(u^2-u_x)_{xy}+2v_{xxx},\\
&v_t=v_{xx}+2uv_x,
\end{split}
\end{equation}
was suggested in~\cite[Eq.~(2.16)]{boit1987a} as a generalization of the sinh-Gordon equation
to the case of more dimensions, see also~\cite[Eq.~(2.3)]{gara1994a}.
In fact, an equivalent system was originally presented in~\cite{boit1987a},
\begin{equation}\label{eq:BLPsystemOriginal}
\begin{split}
&u_{ty}=(u^2-u_x)_{xy}+2w_{xx},\\
&w_t=w_{xx}+2(uw)_x,
\end{split}
\end{equation}
where $w=v_x$,
although now the system~\eqref{eq:BLPsystem} is more widespread in the literature as a form of the Boiti--Leon--Pempinelli system.%
\footnote{%
A modification of~\eqref{eq:BLPsystemOriginal} by means of alternating the sign of~$t$ 
is often, but unjustifiedly, called the (2+1)-dimensional Broer--Kaup--Kupershmidt system/equation(s).
}
Another form of this system,
\begin{equation}\label{eq:BLPsystemPWform}
\begin{split}
&\breve u_{ty}=-\breve w_{xx}-\tfrac12(\breve u^2)_{xy},\\
&\breve w_t=-(\breve u\breve w+\breve u+\breve u_{xy})_x,
\end{split}
\end{equation}
was studied in~\cite{paqu1990a} and called the ``integrable dispersive long-wave equations in two space dimensions'' therein.
It is obviously reduced to the system~\eqref{eq:BLPsystemOriginal}
by the invertible differential substitution $\breve u=-2u$, $\breve w=4w-1-2u_y$.
As a regular integrable system, the system~\eqref{eq:BLPsystem} possesses the Lax pair
$L=\p_x\p_y+u\p_y+v_x$ and $P=-\p_x^2+p$, where $p_y=-2v_{xx}$~\cite{boit1987a}.
In other words, this system is the compatibility condition for the linear system of
partial differential equations
\begin{gather}\label{eq:BLPCovering}
\psi_{xy}+u\psi_y+v_x\psi=0,\quad \psi_t+\psi_{xx}-p\psi=0
\end{gather}
with respect to the auxiliary unknown function~$\psi=\psi(t,x,y)$.
The system~\eqref{eq:BLPsystem} admits a Hamiltonian structure,
and the Painlev\'e test is fulfilled for it~\cite{gara1994a},
see also Remark~\ref{rem:BLPsystemPainleveTest}.
The initial-value problem for the system~\eqref{eq:BLPsystem} was studied in~\cite{pogr1996a}
using the inverse scattering transform.
Two Darboux transformations and two Laplace transformations for an equivalent form of~\eqref{eq:BLPsystem}
as well as the iterated versions of these transformations were constructed in~\cite{yuro1999a}
(see Section~\ref{sec:LaplaceAndDarbouxTrans})
but, unfortunately, they are not of common usage.

First exact solutions of~\eqref{eq:BLPsystem} were constructed in the original paper~\cite{boit1987a}
after deriving a B\"acklund transforation for~\eqref{eq:BLPsystem}.
Lie reductions of the form~\eqref{eq:BLPsystemPWform} of the Boiti--Leon--Pempinelli system
were carried out in~\cite{paqu1990a}. 
Some exact solutions were also found in~\cite{gara1994a,yuro1999a}.
Over the past decade, dozens of papers exclusively devoted to finding exact solutions of~\eqref{eq:BLPsystem}
using various methods have been published but vast majority of the constructed solutions
are trivial or even incorrect, and their construction implicitly involved
two obvious reductions of~\eqref{eq:BLPsystem} to the Burgers equation or simple Lie reductions,
see also the discussion in~\cite{kudr2011b}.
In particular, a partial classification of Lie reductions of the system~\eqref{eq:BLPsystem} was attempted in~\cite{zhao2017a}.
More specifically, a seven-dimensional subalgebra~$\mathfrak c$
of the (infinite-dimensional) maximal Lie invariance algebra~$\mathfrak g$ of~\eqref{eq:BLPsystem}
was merely used for a partial Lie-reduction procedure.
Nevertheless, even the first step of this procedure,
which is the classification of one- and two-dimensional subalgebras of~$\mathfrak c$,
was carried out incorrectly. 
Moreover, all the correct exact solutions presented therein
satisfy the differential constraint $u_y=v_x$,
which reduces~\eqref{eq:BLPsystem} to the Burgers equation.

In the present paper, we carry out the enhanced and extended symmetry analysis of the Boiti--Leon--Pempinelli system~\eqref{eq:BLPsystem}.
This includes
the computation of the point-symmetry pseudogroup~$G$ of this system,
the construction of whole families of its exact solutions satisfying simple non-Lie differential constraints,
the exhaustive classification of Lie reductions of the system~\eqref{eq:BLPsystem} in the optimal way,
the complete study of Lie reductions of this system to ordinary differential equations and
the generation of wide families of its exact solutions using Darboux and Laplace transformations.

More specifically, the megaideal-based%
\footnote{%
A \emph{megaideal}~$\mathfrak m$ \cite{bihl2015a,popo2003a} 
(or a \emph{fully characteristic ideal} \cite[Exercise~14.1.1]{hilg2011A})
of a Lie algebra~$\mathfrak g$
is a linear subspace of $\mathfrak g$ such that $\mathfrak T\mathfrak m\subseteq\mathfrak m$
for any transformation~$\mathfrak T$ from the automorphism group ${\rm Aut}(\mathfrak g)$ of~$\mathfrak g$.
Each megaideal of~$\mathfrak g$ is  a characteristic ideal and thus an ideal of~$\mathfrak g$.
} 
version~\cite{bihl2015a,bihl2011c,card2013a} of the algebraic method~\cite{hydo1998a,hydo2000b}
is applied in Section~\ref{PointSymPseudogroup} for finding the pseudogroup~$G$.
This version is convenient in view of the fact
that the maximal Lie invariance algebra~$\mathfrak g$ of the system~\eqref{eq:BLPsystem} is infinite-dimensional.
We additionally modify this version of the algebraic method via pushing forward
a proper finite-dimensional subalgebra of~$\mathfrak g$ instead of the entire algebra~$\mathfrak g$.
We also analyze the properties of the action of the pseudogroup~$G$ on the algebra~$\mathfrak g$
that underlie the proposed modification.

In Section~\ref{sec:ClassificationSubalgebras}, 
we construct optimal lists of one- and two-dimensional subalgebras of the algebra~$\mathfrak g$.
There are exactly eight $G$-inequivalent one-dimensional subalgebras of~$\mathfrak g$.
(Note that the ``optimal'' list of one-dimensional subalgebras of a seven-dimensional subalgebra of~$\mathfrak g$ 
that was given in~\cite{zhao2017a}
consists of 29 elements each of which ranges between one and four subalgebras.)
The list of $G$-inequivalent two-dimensional subalgebras of~$\mathfrak g$ include
eighteen families of subalgebras, among which there are eight families of non-Abelian algebras
and ten families of Abelian algebras.
Most of these families are singletons or parameterized by discrete parameters
and thus include only a few subalgebras;
five families are parameterized by arbitrary smooth unary functions.
The presented lists of subalgebras is additionally refined
for convenience in carrying out Lie reductions of the system~\eqref{eq:BLPsystem}.
Since the algebra~$\mathfrak g$ is infinite-dimensional,
we use an alternative method for finding the adjoint action of the pseudogroup~$G$ on the algebra~$\mathfrak g$
via pushing forward the vector fields from~$\mathfrak g$ by elements of~$G$.
This also allows us to take into account discrete point symmetries of the system~\eqref{eq:BLPsystem}
from the very beginning when classifying $G$-inequivalent subalgebras of~$\mathfrak g$.
The computations are simplified by the fact that the algebra~$\mathfrak g$ can be represented as a direct sum of two ideals.

Before proceeding with Lie reductions of the system~\eqref{eq:BLPsystem},
in Section~\ref{sec:SolutionsByMethodOfDiffConstraints}
we apply the method of differential constraints to this system,
selecting various simple differential constraints.
Most of them are $G$-invariant,
which allows us to simplify the computations and the form of the constructed solutions
and to avoid analyzing a large number of special cases that arise when deriving solutions.
Each of the constraints $v_x=0$ and $u_y=v_x$ leads to reducing the system~\eqref{eq:BLPsystem}
to a family of inhomogeneous Burgers equations,
which is linearized by the Hopf--Cole transformation to a family of heat equations with potentials,
and the independent variable~$y$ plays the role of an implicit parameter.
The stationary solutions of~\eqref{eq:BLPsystem} satisfying the constraint $u_y=v_x$
are expressed in terms of the well-known general solution of the Liouville equation.
After separately attaching each of the constraints $u_y=0$, $v_{xxx}=0$ or $u_{xx}=v_{4x}=0$ to the system~\eqref{eq:BLPsystem},
the general solution of the obtained overdetermined system can be expressed in a simple explicit form
that is parameterized arbitrary sufficiently smooth functions depending on~$t$ or~$y$.
Under the constraint $u_{xx}=0$, the system~\eqref{eq:BLPsystem}
reduces to a quite general family of Abel equations of the first kind with the independent variable~$y$,
which is parameterized by two arbitrary functions of~$y$
and where $t$ plays the role of an implicit parameter.
This family contains a family of Bernoulli equations, which are easily integrable.
Note that the selection of $G$-invariant differential constraints leading to nontrivial solutions of~\eqref{eq:BLPsystem}
is by no means a simple problem.
For instance, we show that the constraint $u_{xt}u_y=v_xu_{yt}$, $u_{xy}u_y=v_xu_{yy}$ does not result in new solutions of~\eqref{eq:BLPsystem}
in comparison with the above simpler constraints.
We also consider the non-$G$-invariant constraint $v=u$.
Among solutions of the corresponding overdetermined system, there is a new solution family of the system~\eqref{eq:BLPsystem},
which is parameterized by an arbitrary smooth function of the variable~$y$.\looseness=-1

In Section~\ref{sec:LieReductionsOfCodim1}, we carry out the Lie reductions of the system~\eqref{eq:BLPsystem}
with respect to the seven appropriate $G$-inequivalent one-dimensional subalgebras listed in Section~\ref{sec:ClassificationSubalgebras}.
Only the first four reductions may lead to solutions of~\eqref{eq:BLPsystem}
that do not belong to solution families constructed in Section~\ref{sec:SolutionsByMethodOfDiffConstraints}.
The most interesting is the Lie reduction of~\eqref{eq:BLPsystem} to the corresponding stationary system,
which possesses wide families of explicit solutions
expressed in terms of the general solutions of the Liouville, $\sinh$-Gordon and $\cosh$-Gordon equations.
Another reduced system is equivalent to the dispersive long-wave equations,
which are further reduced to the Burgers equation using a simple differential constraint.
To look for hidden symmetries of the system~\eqref{eq:BLPsystem},
we compute the normalizers of the considered one-dimensional subalgebras and
the maximal Lie invariance algebras of the corresponding reduced systems.
The system~\eqref{eq:BLPsystem} admits a great number of hidden symmetries
associated with Lie reductions of codimension one but
these symmetries are not useful since they are Lie symmetries of trivially integrable reduced systems.

Although the complete list of $G$-inequivalent two-dimensional subalgebras of the algebra~$\mathfrak g$
that is presented in Section~\ref{sec:ClassificationSubalgebras} is large,
in fact it suffices to carry out Lie reductions only with respect to five families from the list,
the non-Abelian subalgebras
$\mathfrak s_{2.1}$, $\mathfrak s_{2.2}^\de$, $\mathfrak s_{2.3}^\de$, $\mathfrak s_{2.4}^1$
and the Abelian subalgebras $\mathfrak s^{\de1}_{2.9}$,
which is done in Section~\ref{sec:LieReductionsOfCodim2}.
The families $\{\mathfrak s_{2.1}\}$ and $\{\mathfrak s_{2.4}^1\}$ are singletons,
whereas each of the families $\{\mathfrak s_{2.2}^\de\}$, $\{\mathfrak s_{2.3}^\de\}$ and $\{\mathfrak s^{\de1}_{2.9}\}$
is parameterized by the single parameter $\delta\in\{0,1\}$.
The other listed subalgebras either do not satisfy the rank condition~\cite{boyk2016a}
and are thus not appropriate for Lie reductions at all
or the corresponding Lie reductions of the system~\eqref{eq:BLPsystem} can be represented as two-step reductions,
where the first steps are trivial codimension-one Lie reductions discussed in Section~\ref{sec:LieReductionsOfCodim1}.
The reduced system obtained with the subalgebra~$\mathfrak s_{2.3}^\de$ is inconsistent if $\delta=1$
and, if $\delta=0$, gives a trivial solution of~\eqref{eq:BLPsystem} satisfying the differential constraint $v_x=0$.
For each of the reduced systems
constructed with the subalgebras~$\mathfrak s_{2.1}$, $\mathfrak s_{2.2}^\de$, $\mathfrak s_{2.4}^1$ and~$\mathfrak s^{\de1}_{2.9}$,
we find three independent first integrals, thus lowering its order by three,
and derive a single second-order ordinary differential equation for~$\varphi$
and an expression for~$\psi$ as a differential function of~$\varphi$, which both are parameterized by integration constants.
Here $\varphi$ and~$\psi$ are the chosen invariant dependent variable that are associated with~$u$ and~$v$, respectively.
Modulo the $G$-equivalence, some integration constants are not significant and can be set to be equal convenient values.
For the subalgebra~$\mathfrak s_{2.1}$ (resp.\ $\mathfrak s_{2.2}^\de$) and for each value of the tuple of integration constants,
the above equation for~$\varphi$ is reduced by a differential substitution to a fifth (resp.\ third) Painlev\'e equation,
whereas in the case of the subalgebra~$\mathfrak s_{2.4}^1$ (resp.\ $\mathfrak s^{\de1}_{2.9}$)
it is similar with respect to a simple point transformation to a fourth (resp.\ second) Painlev\'e equation.
An exception is a singular solution subset of the reduced system corresponding to the subalgebras~$\mathfrak s^{\de1}_{2.9}$,
where the unknown function~$\varphi$ is expressed in terms of
either Weierstrass elliptic functions or, in degenerate cases, elementary functions;
the expression for~$\psi$ contains a quadrature and can respectively be represented in terms of
Weierstrass elliptic functions and Weierstrass $\zeta$-functions or elementary functions.
As a result, all the significant inequivalent reduced systems of ordinary differential equations
obtained from the system~\eqref{eq:BLPsystem} using Lie reductions
are exhaustively solved in terms of
Painlev\'e transcendents, Weierstrass $\wp$- and $\zeta$-functions and elementary functions.

Symmetry properties of several equivalent forms of the Boiti--Leon--Pempinelli system
are discussed in Section~\ref{sec:BLPsystemEquivForms}.

We show in Section~\ref{sec:LaplaceAndDarbouxTrans}
that the constructed families of exact solutions of the system~\eqref{eq:BLPsystem}
can be significantly extended using the aforementioned Laplace and Darboux transformations,
which were found in~\cite{yuro1999a}.

Section~\ref{sec:Conclusion} is devoted to the discussions of the results obtained in the paper.
In particular, we pose several problems concerning optimization of the algebraic method
for computing the complete point-symmetry (pseudo)group of a system of differential equations.
We also list all the well-known differential equations that arise as various reductions
of the Boiti--Leon--Pempinelli system.

In Section~\ref{sec:IntegrationOfODEsRelatedToEllipticFunctions}, we recall
how to integrate general autonomous first-order ordinary differential equations of the form~\eqref{eq:EllipticODEs}
and special equations of this form in terms of Weierstrass elliptic functions and elementary functions, respectively.

The reduced conservation-law characteristics up to order four of the system~\eqref{eq:BLPsystem}
are presented in Section~\ref{sec:BLPsystemCLs}.
In fact, they are exhausted by those up to order one,
and each of the corresponding conservation law contains a ``short'' conserved current.

For readers' convenience,
throughout the paper we marked constructed solutions of the Boiti--Leon--Pempinelli system by bullets.

\section{Point-symmetry pseudogroup}\label{PointSymPseudogroup}

The maximal Lie invariance algebra~$\mathfrak g$ of the system~\eqref{eq:BLPsystem}
is infinite-dimensional and is spanned by the vector fields
\[
\begin{split}
&D(f)=f\p_t+\tfrac12f_tx\p_x-\left(\tfrac12f_tu+\tfrac14f_{tt}x\right)\p_u,\quad
 S(\alpha)=\alpha\p_y-\alpha_y v\p_v,\\
&P(g)=g\p_x-\tfrac12g_t\p_u,\quad
 Z(\beta)=\beta\p_v,
\end{split}
\]
where $f=f(t)$, $g=g(t)$, $\alpha=\alpha(y)$, $\beta=\beta(y)$ are arbitrary smooth functions of their arguments,
\[
\mathfrak g=\langle D(f),S(\alpha),P(g),Z(\beta)\rangle.
\]
Up to antisymmetry of Lie brackets,
the nonzero commutations relations between vector fields spanning~$\mathfrak g$
are exhausted by
\begin{gather}\label{eq:CommRels}
\begin{split}
&[D(f^1),D(f^2)]=D(f^1f^2_t-f^1_tf^2),\quad
 [S(\alpha^1),S(\alpha^2)]=S(\alpha^1\alpha^2_y-\alpha^1_y\alpha^2),\\
&[P(g),D(f)]=P(\tfrac12f_tg-fg_t),\quad
 [S(\alpha),Z(\beta)]=Z\big((\alpha \beta)_y\big).
\end{split}
\end{gather}
Hence the algebra~$\mathfrak g$ can be represented as the direct sum of its ideals
$\mathfrak i_1:=\langle D(f), P(g)\rangle$ and $\mathfrak i_2:=\langle S(\alpha), Z(\beta)\rangle$,
$\mathfrak g=\mathfrak i_1\oplus\mathfrak i_2$.
The spans $\mathfrak i_3:=\langle P(g)\rangle$ and $\mathfrak i_4:=\langle Z(\beta)\rangle$ are ideals of~$\mathfrak g$ as well. The question whether the above ideals are megaideals of~$\mathfrak g$ remains open.
At the same time, we can prove that the sum $\mathfrak i_3\oplus\mathfrak i_4$ is a megaideal of~$\mathfrak g$.

\begin{lemma}\label{lem:RadicalOfg}
$\mathrm R_{\mathfrak g}=\langle P(g),Z(\beta)\rangle$,
where $\mathrm R_{\mathfrak g}$ denotes the radical of~$\mathfrak g$.
\end{lemma}

\begin{proof}
Let $\mathfrak s:=\langle P(g),Z(\beta)\rangle$.
It is obvious from the commutation relations~\eqref{eq:CommRels}
that $\mathfrak s$ is a solvable (more precisely, Abelian) ideal of~$\mathfrak g$.
Moreover, $\mathfrak s$ is the maximal solvable ideal of~$\mathfrak g$.
Indeed, let $\mathfrak s_1$ be an ideal of $\mathfrak g$ properly containing $\mathfrak s$.
This means that there exists a vector field $Q^1:=D(f^1)+S(\alpha^1)$ in $\mathfrak s_1$,
where at least one of the parameter functions $f^1=f^1(t)$ and $\alpha^1=\alpha^1(y)$
does not identically vanish.
Suppose that the function $f^1$ does not identically vanish.
Let $I$ be an interval contained in the domain of~$f^1$ such that $f^1(t)\neq0$ for any $t\in I$.
We restrict all the parameter functions of~$t$ in~$\mathfrak g$ on the interval~$I$.
Since $\mathfrak s_1$ is an ideal of $\mathfrak g$,
the commutator $[D(f),Q^1]=D(\tilde f)$, where $\tilde f:=ff^1_t-f_tf^1$,
belongs to~$\mathfrak s_1$ for all~$f\in\mathrm C^\infty(I)$.
In~view of the existence theorem for first-order linear ordinary differential equations,
the function~$\tilde f$ runs through~$\mathrm C^\infty(I)$ if
the function~$f$ runs through~$\mathrm C^\infty(I)$.
Therefore, $\mathfrak s_1\supset\langle D(f)\rangle$ and thus
$\mathfrak s_1^{(n)}\supseteq\langle P(g),D(f)\rangle\ne\{0\}$ for all $n\in\mathbb N$,
i.e., the ideal~$\mathfrak s_1$ is not solvable.
The proof in the case when the function $\alpha^1$ does not identically vanish is similar.
Therefore, $\mathfrak s=\mathrm R_{\mathfrak g}$.
\end{proof}

\begin{theorem}\label{thm:PointSymPseudogroupOfBLPSystem}
The point-symmetry pseudogroup~$G$ of the Boiti--Leon--Pempinelli system~\eqref{eq:BLPsystem}
consists of the transformations of the form
\begin{gather*}
\tilde t=T(t),\quad
\tilde x=\ve T_t^{1/2}x+X^0(t),\quad
\tilde y=Y(y),\\
\tilde u=\frac{\ve}{T_t^{1/2}}u-\frac{\ve T_{tt}}{4T_t^{3/2}}x-\frac{X_t^0}{2T_t},\quad
\tilde v=\frac1{Y_y}v+V^0(y),
\end{gather*}
where $T$, $X^0$, $Y$ and $V^0$ are arbitrary smooth functions of their arguments with $T_t>0$ and $Y_y\neq0$, and $\ve=\pm1$.
\end{theorem}

\noprint{
\begin{theorem}\label{thm:PointSymPseudogroupOfBLPSystem}
The point-symmetry pseudogroup~$G$ of the Boiti--Leon--Pempinelli system~\eqref{eq:BLPsystem}
consists of the transformations of the form
\begin{gather*}
\tilde t=T(t),\quad
\tilde x=X^1(t)x+X^0(t),\quad
\tilde y=Y(y),\quad
\tilde u=\frac u{X_1}-\frac{X_t^1x+X_t^0}{2(X_t^1)^2},\quad
\tilde v=\frac v{Y_y}+V^0(y),\quad
\end{gather*}
where $X^1$, $X^0$, $Y$ and $V^0$ are arbitrary function of their arguments
with $X^1\neq0$ and $Y_y\neq0$, and $T_t=(X^1)^2$.
\end{theorem}
}

\begin{proof}
Since the maximal Lie invariance algebra~$\mathfrak g$ of the system~\eqref{eq:BLPsystem} is infinite-dimensional,
we compute the pseudogroup~$G$ using the megaideal-based version~\cite{bihl2015a,bihl2011c,card2013a} of the algebraic method~\cite{hydo1998a,hydo2000b},
additionally modifying it.
The application of this method to the system~\eqref{eq:BLPsystem} is based on the following fact.
Let a point transformation~$\Phi$ in the space with the coordinates $(t,x,y,u,v)$,
$
\Phi\colon\ (\tilde t,\tilde x,\tilde y,\tilde u,\tilde v)=(T,X,Y,U,V),
$
where $(T,X,Y,U,V)$ is a tuple of smooth functions of $(t,x,y,u,v)$ with nonvanishing Jacobian,
be a point symmetry of the system~\eqref{eq:BLPsystem}.
Then $\Phi_*\mathrm R_{\mathfrak g}\subseteq\mathrm R_{\mathfrak g}$
and $\Phi_*\mathfrak g\subseteq\mathfrak g$,
where $\Phi_*$ is the pushforward of vector fields by~$\Phi$,
and $\mathrm R_{\mathfrak g}$ is the radical of~$\mathfrak g$, see Lemma~\ref{lem:RadicalOfg}.
We take five linearly independent vector fields~$Q^1$,~\dots, $Q^5$ from $\mathrm R_{\mathfrak g}$
and four linearly independent vector fields~$Q^6$,~\dots,~$Q^9$ from $\mathfrak g\setminus\mathrm R_{\mathfrak g}$,
\begin{gather*}
\begin{split}
&Q^1:=Z(1),\quad Q^2:=Z(y),\quad Q^3:=P(1),\quad Q^4:=P(t),\quad Q^5:=P(t^2),\\[1ex]
&Q^6:=S(1),\quad Q^7:=S(y),\quad Q^8:=D(1),\quad Q^9:=D(t),
\end{split}
\end{gather*}
and expand the conditions $\Phi_*Q^i\in\mathrm R_{\mathfrak g}$, $i=1,\dots,5$,
and $\Phi_*Q^i\in\mathfrak g$, $i=6,\dots,9$,
\begin{gather}\label{eq:MainPushforwards}
\begin{split}
&\Phi_*Q^i=\tilde P(\tilde g^i)+\tilde Z(\tilde\beta^i),\quad (\tilde g^i,\tilde\beta^i)\ne(0,0),\quad i=1,\dots,5,
\\[1ex]
&\Phi_*Q^i=\tilde D(\tilde f^i)+\tilde S(\tilde\alpha^i)+\tilde P(\tilde g^i)+\tilde Z(\tilde\beta^i),\quad
(\tilde f^i,\tilde\alpha^i)\ne(0,0),\quad i=6,\dots,9.
\end{split}
\end{gather}
For each~$i\in\{1,\dots,9\}$, we expand the corresponding equation in~\eqref{eq:MainPushforwards} and split it componentwise.
We simplify the obtained constraints, taking into account constraints derived in the same way for preceding values of~$i$
and omitting the constraints satisfied identically in view of other constraints.
Thus, for $i=1,2$, we get
\begin{gather}\label{eq:MainPushforwards1stSubsystem}
\begin{split}
&T_v=0,\quad  X_v=\tilde g^1(T),\quad Y_v=0,\quad U_v=-\frac12\tilde g^1_{\tilde t}(T),\quad  V_v=\tilde\beta^1(Y),\\
&y\tilde g^1(T)=\tilde g^2(T),\quad y\tilde g^1_{\tilde t}(T)=\tilde g^2_{\tilde t}(T),\quad  y\tilde\beta^1(Y)=\tilde\beta^2(Y).
\end{split}
\end{gather}
Since \smash{$\p_y\big(y\tilde g^1(T)-\tilde g^2(T)\big)-T_y(y\tilde g^1_{\tilde t}(T)-\tilde g^2_{\tilde t}(T))=\tilde g^1$},
a differential consequence of the above system is $\tilde g^1=0$, and thus $\tilde g^2=0$ as well, i.e., $X_v=U_v=0$.
Then $\tilde\beta^1\neq0$ and $y=\tilde\beta^2(Y)/\tilde\beta^1(Y)$.
The derivative $(\tilde\beta^2/\tilde\beta^1)_{\tilde y}$ does not vanish
since $y$ cannot be a constant.
Therefore, the inverse function theorem implies that $Y=Y(y)$.

The subsystem of equations from~\eqref{eq:MainPushforwards} with $i=3,4,5$ give
\begin{subequations}\label{eq:MainPushforwards2ndSubsystem}
\begin{gather}
T_x=0,\quad T_u=0,\quad X_x=\tilde g^3(T),\quad X_u=2t\tilde g^3(T)-2\tilde g^4(T),
\label{eq:MainPushforwards2ndSubsystemA}\\
U_x=-\smash{\frac12}\tilde g^3_{\tilde t}(T),\quad U_u=\tilde g^4_{\tilde t}(T)-t\tilde g^3_{\tilde t}(T),\quad V_x=\tilde\beta^3(Y),\quad V_u=2t\tilde\beta^3(Y)-2\tilde\beta^4(Y),
\label{eq:MainPushforwards2ndSubsystemB}\\
t^2\tilde g^3(T)-2t\tilde g^4(T)+\tilde g^5(T)=0,\quad t^2\tilde g^3_{\tilde t}(T)-2t\tilde g^4_{\tilde t}(T)+\tilde g^5_{\tilde t}(T)=0,
\label{eq:MainPushforwards2ndSubsystemC}\\[.5ex]
t^2\tilde\beta^3(Y)-2t\tilde\beta^4(Y)+\tilde\beta^5(Y)=0.
\label{eq:MainPushforwards2ndSubsystemD}
\end{gather}
\end{subequations}
Since the functions $\tilde\beta^i(Y)$ depend at most on $y$,
we can split the equation~\eqref{eq:MainPushforwards2ndSubsystemD} with respect to~$t$,
obtaining $\tilde\beta^i(Y)=0$, i.e., $V_x=V_u=0$ and $\tilde g^3(T)\neq0$.
A differential consequence of the equations~\eqref{eq:MainPushforwards2ndSubsystemC} is $t\tilde g^3(T)-\tilde g^4(T)=0$,
and hence $X_u=0$ in view of the last equation in~\eqref{eq:MainPushforwards2ndSubsystemA}.
The equation $t=\tilde g^4(T)/\tilde g^3(T)$ implies that the derivative $(\tilde g^4/\tilde g^3)_{\tilde t}$ does not vanish
since otherwise we would have the inconsistent condition that $t$ is a constant.
In view of the inverse function theorem, we derive that $T=T(t)$.

Considering the equations from~\eqref{eq:MainPushforwards} with $i=6,7$ and with $i=8,9$, we get
\begin{gather}
\begin{split}\label{eq:MainPushforwards3rdSubsystem}
&X_y=\tilde g^6(T),\quad Y_y=\tilde\alpha^6(Y),\quad U_y=-\frac12\tilde g^6_{\tilde t}(T),\\
&V_y=-\tilde\alpha^1_{\tilde y}(Y)V+\tilde\beta^6(Y),\quad vV_v=(\tilde\alpha^7_{\tilde y}(Y)-y\tilde\alpha^6_{\tilde y}(Y))V+y\tilde\beta^6(Y)-\tilde\beta^7(Y),\\[.5ex]
&y\tilde\alpha^6(Y)=\tilde\alpha^7(Y),\quad \tilde f^6(T)=\tilde f^7(T)=0,\quad y\tilde g^6(T)=\tilde g^7(T),\quad y\tilde g^6_{\tilde t}(T)=\tilde g^7_{\tilde t}(T),
\end{split}
\\[1.5ex]
\begin{split}\label{eq:MainPushforwards4thSubsystem}
&T_t=\tilde f^8(T),\quad X_t=\frac12\tilde f^8_{\tilde t}(T)X+\tilde g^8(T),\\[.5ex]
&xX_x=-\big(t\tilde f^8_{\tilde t}(T)-\tilde f^9_{\tilde T}(T)\big)X-2\big(t\tilde g^8(T)-\tilde g^9(T)\big),\\[.5ex]
&U_t=-\frac12\tilde f^8_{\tilde t}(T)U-\frac14\tilde f^8_{\tilde t\tilde t}(T)-\frac12\tilde g^8_{\tilde t}(T),\\
&xU_x-U_u=(t\tilde f^8_{\tilde t}(T)-\tilde f^9_{\tilde t}(T))U-\frac12(t\tilde f^8_{\tilde t\tilde t}(T)-\tilde f^9_{\tilde t\tilde t}(T))X+t\tilde g^8_{\tilde t}(T)-\tilde g^9_{\tilde t}(T),\\[.5ex] 
&V_t=\tilde\beta^8(Y),\quad \tilde\alpha^8(Y)=\tilde\alpha^9(Y)=0,\quad t\tilde f^8(T)=\tilde f^9(T),\quad t\tilde\beta^8(Y)=\tilde\beta^9(Y).
\end{split}
\end{gather}
Splitting the equations $y\tilde g^6(T)=\tilde g^7(T)$ and $t\tilde\beta^8(Y)=\tilde\beta^9(Y)$
with respect to~$y$ and~$t$, respectively,
leads to the equations $\tilde g^6(T)=\tilde g^7(T)=0$ and $\tilde\beta^8(Y)=\tilde\beta^9(Y)=0$, i.e., $X_y=U_y=V_t=0$.

Jointly to the conditions $T=T(t)$ and~$Y=Y(y)$, we collect and then integrate
the simplest derived equations for the $(x,u,v)$-components,
\begin{gather*}
X_y=X_u=X_v=0,\quad X_x=\tilde g^3(T),\\
U_y=U_v=0,\quad U_x=-\frac12\tilde g^3_{\tilde t}(T),\quad U_u=\tilde g^4_{\tilde t}-t\tilde g^3_{\tilde t}(T),\\
V_t=V_x=V_u=0,\quad V_v=\tilde\beta^1(Y).
\end{gather*}
As a result, we obtain the following representation for the components of point symmetry transformations of the system~\eqref{eq:BLPsystem}:
\begin{gather}\label{eq:MainPushforwardsResult}
\begin{split}
&T=T(t),\quad X=X^1(t)x+X^0(t),\quad Y=Y(y),\\
&U=U^2(t)u+U^1(t)x+U^0(t),\quad V=V^1(y)v+V^0(y),
\end{split}
\end{gather}
where $X^0$, $X^1$, $U^0$, $U^1$, $U^2$, $V^0$ and $V^1$ are smooth functions of their arguments.

Continuing the computation within the framework of the algebraic method, we can derive more constraints on elements of the pseudogroup~$G$.
At the same time, as will be discussed in Remark~\ref{rem:FurtherUseOfAlgMethod} below,
this computation is quite cumbersome and, moreover, does not give all the constraints on the transformations constituting the pseudogroup~$G$.
Since the form~\eqref{eq:MainPushforwardsResult} is sufficiently restrictive to simplify the expressions for transformed derivatives,
this is an appropriate point to start using the direct method.
It is based on the fact that, by the definition of point symmetries,
the transformation~$\Phi$ maps the system~\eqref{eq:BLPsystem} to the system of the same form in the variables with tildes,
\begin{gather}\label{eq:BLP-systemWithTilde}
\begin{split}
&\tilde u_{\tilde t\tilde y}=\left(\tilde u^2-\tilde u_{\tilde x}\right)_{\tilde x \tilde y}+2\tilde v_{\tilde x\tilde x\tilde x},\\
&\tilde v_{\tilde t}=\tilde v_{\tilde x\tilde x}+2\tilde u\tilde v_{\tilde x}.
\end{split}
\end{gather}
Therefore, expressing derivatives with tildes in~\eqref{eq:BLP-systemWithTilde} in terms of the original variables
in view of the chain rule,
\[
\p_{\tilde t}=\frac1{T_t}\p_t-\frac{X^1_tx+X^0}{X^1T_t}\p_x,\quad
\p_{\tilde x}=\frac1{X^1}\p_x,\quad
\p_{\tilde y}=\frac1{Y_y}\p_y,
\]
\noprint{
\begin{gather*}
\tilde u_{\tilde y}=\frac{U^2}{Y_y}u_y,\quad \tilde u_{\tilde x}=\frac{U^2u_x+U^1}{X_x},\\
\tilde u_{\tilde t\tilde y}=\frac{(U^2u_y)_t}{Y_yT_t}-\frac{(X^1_tx+X^0_t)U^2u_{xy}}{X^1Y_yT_t},\quad\tilde u_{\tilde x\tilde y}=\frac{U^2u_{xy}}{X^1Y_y},\quad\tilde u_{\tilde x\tilde x\tilde y}=\frac{U^2u_{xxy}}{(X^1)^2Y_y},\\
\tilde v_{\tilde t}=V^1(y)(\frac1{T_t}v_t-\frac{X^1_tx+X^0_t}{X^1T_t}v_x),\quad \tilde v_{\tilde x}=\frac{V^1}{X_x}v_x,\quad \tilde v_{\tilde x\tilde x}=\frac{V^1}{(X_x)^2}v_{xx}.
\end{gather*}
}
we should obtain a system identically satisfied on the solution set of the system~\eqref{eq:BLPsystem}.
Thus, substituting the expression for the leading derivatives~$u_{ty}$ and~$v_t$ in view of the system~\eqref{eq:BLPsystem}
into the second equation of~\eqref{eq:BLP-systemWithTilde} results in the identity
\[
V^1\left(\frac1{T_t}(v_xx+2uv_x)-\frac{X^1_tx+X^0_t}{X^1T_t}v_x\right)=\frac{V^1}{(X_x)^2}v_{xx}+2(U^2u+U^1x+U^0)\frac{V^1}{X^1}v_x.
\]
After collecting the coefficients of $v_{xx}$, $v_x$, $xv_x$, and $uv_x$, we derive
\begin{gather}\label{eq:FinalEqsForPointSyms}
T_t=(X^1)^2,\quad U^0=-\frac{X^0_t}{2T_t},\quad U^1=-\frac{X^1_t}{2T_t},\quad U^2=\frac{X^1}{T_t}.
\end{gather}
It is obvious from the first equation that $T_t>0$ and $X^1=\ve\sqrt T_t$, where $\ve=\pm1$.
In a similar way, the first equation of~\eqref{eq:BLP-systemWithTilde} implies only one more constraint
for point symmetry transformations of~\eqref{eq:BLPsystem}, $V^1=1/Y_y$,
obtained by collecting the coefficients of~$v_{xxx}$ in view of~\eqref{eq:FinalEqsForPointSyms}.

Thus, the transformation~$\Phi$ has the form declared in the statement of the theorem,
and any point transformation of this form is a point symmetry transformation of the system~\eqref{eq:BLPsystem}.
\end{proof}

Each transformation $\Phi$ from the point-symmetry pseudogroup~$G$ of the system~\eqref{eq:BLPsystem}
can be represented as a composition
\[
\Phi=\mathscr S(Y)\circ\mathscr Z(Y_yV^0)\circ\mathscr I(\ve)\circ\mathscr D(T)\circ\mathscr P(\ve T_t^{-1/2}X^0)
\]
of transformations from (pseudo)subgroups each of which is parameterized by a single functional or discrete parameter,
\begin{equation}\label{eq:ElementaryTransformations}\arraycolsep=0ex
\begin{array}{llllll}
\mathscr D(T)  \colon\quad &\tilde t=T,\quad & \tilde x=T_t^{1/2}x, \quad & \tilde y=y,\quad & \tilde u=T_t^{-1/2}u-\frac14T_{tt}T_t^{-3/2}x,\quad & \tilde v=v,\\[1ex]
\mathscr S(Y)  \colon\quad &\tilde t=t,\quad & \tilde x=x,          \quad & \tilde y=Y,\quad & \tilde u=u,                                   \quad & \tilde v=Y_y^{-1}v,\\[1ex]
\mathscr P(X^0)\colon\quad &\tilde t=t,\quad & \tilde x=x+X^0,      \quad & \tilde y=y,\quad & \tilde u=u-\frac12X_t^0,                      \quad & \tilde v=v,\\[1ex]
\mathscr Z(V^0)\colon\quad &\tilde t=t,\quad & \tilde x=x,          \quad & \tilde y=y,\quad & \tilde u=u,                                   \quad & \tilde v=v+V^0,\\[1ex]
\mathscr I(\ve)\colon\quad &\tilde t=t,\quad & \tilde x=\ve x,      \quad & \tilde y=y,\quad & \tilde u=\ve u,                               \quad & \tilde v=v,\\[1ex]
\end{array}
\end{equation}
where $T=T(t)$, $X^0=X^0(t)$, $Y=Y(y)$ and $V^0=V^0(y)$ are arbitrary smooth function of their arguments with $T_t>0$ and $Y_y\neq0$, and $\ve\in\{-1,1\}$.
We will call transformations from the families~\eqref{eq:ElementaryTransformations}
elementary point symmetry transformations of the system~\eqref{eq:BLPsystem}.
Note that the (pseudo)subgroups~$\{\mathscr D(T)\}$, $\{\mathscr S(Y)\}$, $\{\mathscr P(X^0)\}$ and $\{\mathscr Z(V^0)\}$ of~$G$
are associated with the subalgebras~$\{D(f)\}$, $\{S(\alpha)\}$, $\{P(g)\}$ and $\{Z(\beta)\}$ of~$\mathfrak g$, respectively.
Here all the parameter functions run through the specified sets of their values.

\begin{corollary}\label{cor:BLPSystemDiscretePointSymTrans}
A complete list of discrete point symmetry transformations of the Boiti--Leon--Pempinelli system~\eqref{eq:BLPsystem}
that are independent up to combining with each other and with continuous point symmetry transformations of this system
is exhausted by two involutions alternating signs of variables,
\ $\mathscr I(-1)\colon (t,x,y,u,v)\mapsto(t,-x,y,-u,v)$,
\ $\mathscr S(-y)\colon (t,x,y,u,v)\mapsto(t,x,-y,u,-v)$.
\noprint{
\begin{align*}
\mathscr I(-1)\colon\ &(t,x,y,u,v)\mapsto(t,-x,y,-u,v),\\
\mathscr S(-y)\colon\ &(t,x,y,u,v)\mapsto(t,x,-y,u,-v).
\end{align*}
}
\end{corollary}

\begin{corollary}\label{cor:BLPSystemFactorGroup}
The factor group of the point-symmetry pseudogroup $G$ of the system~\eqref{eq:BLPsystem}
with respect to its identity component is isomorphic to the group $\mathbb Z_2\times\mathbb Z_2$.
\end{corollary}

\begin{remark}\label{rem:dNinvolutions}
Corollaries~\ref{cor:BLPSystemDiscretePointSymTrans} and~\ref{cor:BLPSystemFactorGroup}
become completely rigorous after assuming each of the listed discrete transformations
to be defined on the entire corresponding underlying space
and to absorb its restrictions.
\end{remark}

\begin{remark}\label{rem:FurtherUseOfAlgMethod}
Substituting the form~\eqref{eq:MainPushforwardsResult}
into the system~\eqref{eq:MainPushforwards1stSubsystem}--\eqref{eq:MainPushforwards4thSubsystem}
and integrating the remaining equations, we obtain that
\begin{gather}\label{eq:FurtherUseOfAlgMethod}
C_1T_t=(X^1)^2,\quad U^0=-\frac{X^0_t}{2T_t},\quad U^1=-\frac{X^1_t}{2T_t},\quad U^2=\frac{X^1}{T_t},\quad V^1=\frac{C_2}{Y_y},
\end{gather}
where $C_1$ and~$C_2$ are nonzero constants.
The vector fields~$Q^1$, \dots, $Q^9$ span a subalgebra~$\mathfrak s$ of~$\mathfrak g$, $\mathfrak s=\langle Q^1,\dots,Q^9\rangle$.
The system~\eqref{eq:MainPushforwards1stSubsystem}--\eqref{eq:MainPushforwards4thSubsystem}
is the expansion of the condition~\eqref{eq:MainPushforwards},
which can be briefly written as
\begin{gather}\label{eq:AlgCriterionEvaluatedOnSubalgebraA}
\Phi_*(\mathfrak s\cap\mathrm R_{\mathfrak g})\subseteq\mathrm R_{\mathfrak g}
\quad\mbox{and}\quad
\Phi_*\mathfrak s\subseteq\mathfrak g.
\end{gather}
One can check that the complete condition $\Phi_*\mathfrak g\subset\mathfrak g$
is satisfied by any point transformation~$\Phi$ of the form~\eqref{eq:MainPushforwardsResult}
with the additional constraints~\eqref{eq:FurtherUseOfAlgMethod}.
In other words, no more restrictions on elements of~$G$ can be obtained using the algebraic method.
The condition $\Phi_*\mathfrak g\subset\mathfrak g$ is equivalent to the condition~\eqref{eq:AlgCriterionEvaluatedOnSubalgebraA},
i.e., it suffices to check the former condition
only for a basis of the finite-dimensional subalgebra~$\mathfrak s$ of the infinite-dimensional algebra~$\mathfrak g$,
but taking into account the megaideal structure of~$\mathfrak g$.
It is also clear that using merely the algebraic method does not allow one to find
all the restrictions on the form of point symmetry transformations of the system~\eqref{eq:BLPsystem},
although the direct method in fact adds only simple constraints for the constants~$C_1$ and~$C_2$, $C_1=C_2=1$.
\end{remark}

In fact, each point transformation~$\Phi$ with $\Phi_*\mathfrak s\subset\mathfrak g$
respects the megaideal $\mathrm R_{\mathfrak g}$.

\begin{proposition}\label{pro:RespectingMegaideal}
The condition $\Phi_*\mathfrak s\subset\mathfrak g$ for a point transformation~$\Phi$ implies
$\Phi_*(\mathfrak s\cap\mathrm R_{\mathfrak g})\subseteq\mathrm R_{\mathfrak g}$,
and hence $\Phi_*\mathrm R_{\mathfrak g}\subseteq\mathrm R_{\mathfrak g}$ and $\Phi_*\mathfrak g\subset\mathfrak g$ as well.
\end{proposition}

\begin{proof}
In the notation of the proof of Theorem~\ref{thm:PointSymPseudogroupOfBLPSystem},
the condition $\Phi_*\mathfrak s\subset\mathfrak g$ expands to $\Phi_*Q^i\in\mathfrak g$, $i=1,\dots,9$, i.e.,
\begin{gather}\label{eq:ModifiedMainPushforwards}
\Phi_*Q^i=\tilde D(\tilde f^i)+\tilde S(\tilde\alpha^i)+\tilde P(\tilde g^i)+\tilde Z(\tilde\beta^i),\quad
(\tilde f^i,\tilde\alpha^i,\tilde g^i,\tilde\beta^i)\ne(0,0),\quad i=1,\dots,9.
\end{gather}
We have to show that $f^i=0$ and $\alpha^i=0$, $i=1,\dots,5$.

We split the equality $\Phi_*Q^2=y\Phi_*Q^1$ componentwise and take into account the first two equations
in~\eqref{eq:ModifiedMainPushforwards}.
As a result, we derive the system
\begin{gather}\label{eq:ModifiedMainPushforwardsSplit1}
\begin{split}
&\tilde f^2(T)=y\tilde f^1(T), \quad
 \tilde f^2_{\tilde t}(T)X+2\tilde g^2(T)=y\big(\tilde f^1_{\tilde t}(T)X+2\tilde g^1(T)\big),\\[1ex]
&2\tilde f^2_{\tilde t}(T)U+\tilde f^2_{\tilde t\tilde t}(T)X+2\tilde g^2_{\tilde t}(T)
=y\big(2\tilde f^1_{\tilde t}(T)U+\tilde f^1_{\tilde t\tilde t}(T)X+2\tilde g^1_{\tilde t}(T)\big),\\[1ex]
&\tilde\alpha^2(Y)=y\tilde\alpha^1(Y), \quad
 \tilde\alpha^2_{\tilde y}(Y)V-\tilde\beta^2(Y)=y\big(\tilde\alpha^1_{\tilde y}(Y)V-\tilde\beta^1(Y)\big).
\end{split}
\end{gather}

Suppose that $\tilde f^1\ne0$. Then the first equation of~\eqref{eq:ModifiedMainPushforwardsSplit1}
implies that $T$ is a function only of~$y$, $T=T(y)$ with $T_y\ne0$.
After differentiating this equation with respect to~$y$,
we also obtain \smash{$\tilde f^2_{\tilde t}(T)-y\tilde f^1_{\tilde t}(T)=\tilde f^1(T)/T_y\ne0$}.
Hence the second equation of~\eqref{eq:ModifiedMainPushforwardsSplit1}
implies that $X$ is a function only of~$y$ as well, which contradicts the nondegeneracy of~$\Phi$.
Therefore, $\tilde f^1=\tilde f^2=0$.

Similarly, supposing $\tilde\alpha^1\ne0$, we derive from the last two equations of~\eqref{eq:ModifiedMainPushforwardsSplit1}
that both $Y$ and~$V$ are functions only of~$y$, which again contradicts the nondegeneracy of~$\Phi$.
Hence $\tilde \alpha^1=\tilde \alpha^2=0$.

The second and the third equations of~\eqref{eq:ModifiedMainPushforwardsSplit1},
which respectively reduce to $\tilde g^2(T)=y\tilde g^1(T)$ and $\tilde g^2_{\tilde t}(T)=y\tilde g^1_{\tilde t}(T)$,
imply the differential consequence $\tilde g^1(T)=\big(\tilde g^2_{\tilde t}(T)-y\tilde g^1_{\tilde t}(T)\big)T_y=0$,
and thus $\tilde g^2(T)=0$ as well.
Then the two first inequalities from~\eqref{eq:ModifiedMainPushforwards} mean that $\tilde\beta^1\tilde\beta^2\ne0$,
and we obtain from the last equation of~\eqref{eq:ModifiedMainPushforwardsSplit1}
that $Y$ is a function only of~$y$, $Y=Y(y)$ with $Y_y\ne0$.

In view of the equations in~\eqref{eq:ModifiedMainPushforwards} with $i=3,4,5$,
the componentwise splitting of the equality $\Phi_*Q^5-2t\Phi_*Q^4+t^2\Phi_*Q^3=0$
leads, in particular, to the equations
\begin{gather}\label{eq:ModifiedMainPushforwardsSplit2}
\begin{split}
&\tilde f^5(T)-2t\tilde f^4(T)+t^2\tilde f^3(T)=0, \\[1ex]
&\big(\tilde f^5_{\tilde t}(T)-2t\tilde f^4_{\tilde t}(T)+t^2\tilde f^3_{\tilde t}(T)\big)X
 +2\big(\tilde g^5(T)-2t\tilde f^g(T)+t^2\tilde f^g(T)\big)=0,\\[1ex]
&2\big(\tilde f^5_{\tilde t}(T)-2t\tilde f^4_{\tilde t}(T)+t^2\tilde f^3_{\tilde t}(T)\big)U
 +\big(\tilde f^5_{\tilde t\tilde t}(T)-2t\tilde f^4_{\tilde t\tilde t}(T)+t^2\tilde f^3_{\tilde t\tilde t}(T)\big)X\\
&\qquad+2\big(\tilde g^5(T)-2t\tilde f^g(T)+t^2\tilde f^g(T)\big)=0,\\[1ex]
&\tilde\alpha^5(Y)-2t\tilde\alpha^4(Y)+t^2\tilde\alpha^3(Y)=0.
\end{split}
\end{gather}
Since $Y=Y(y)$, the last equation directly implies
that $\tilde\alpha^5=\tilde\alpha^4=\tilde\alpha^3=0$.

Let us suppose that $(\tilde f^3,\tilde f^4)\ne(0,0)$.
Then we get from the first equation of~\eqref{eq:ModifiedMainPushforwardsSplit2}
that $T$ is a function only of~$t$, $T=T(t)$ with $T_t\ne0$.
In view of this condition, it follows from the next two equations of~\eqref{eq:ModifiedMainPushforwardsSplit2}
that
\[
\tilde f^5_{\tilde t}(T)-2t\tilde f^4_{\tilde t}(T)+t^2\tilde f^3_{\tilde t}(T)=0,\quad
\tilde f^5_{\tilde t\tilde t}(T)-2t\tilde f^4_{\tilde t\tilde t}(T)+t^2\tilde f^3_{\tilde t\tilde t}(T)=0.
\]
(Otherwise, $X$ is also a function only of~$t$, which contradicts the nondegeneracy of~$\Phi$.)
As differential consequences of the obtained equations for $(\tilde f^3,\tilde f^4,\tilde f^5)$,
we successively derive $\tilde f^4(T)=t\tilde f^3(T)$, $\tilde f^4_{\tilde t}(T)=t\tilde f^3_{\tilde t}(T)$,
and then $\tilde f^3=0$, $\tilde f^4=0$ and $\tilde f^5=0$.
\end{proof}

\section{Classification of subalgebras}\label{sec:ClassificationSubalgebras}

To carry out Lie reductions of codimension one and two for the system~\eqref{eq:BLPsystem} in the optimal way,
we should classify one- and two-dimensional subalgebras of the algebra~$\mathfrak g$ up to $G_*$-equivalence.
Instead of the classical approach for finding inner automorphisms~\cite[Section~3.3]{olve1993A},
we act on~$\mathfrak g$ by~$G$ via pushing forward of vector fields by elements of~$G$.
Recall that this way is more convenient for computing in the infinite-dimensional case~\cite{bihl2012b,card2011a}.
Moreover, it also allows us to properly use the entire point-symmetry pseudogroup~$G$
and not be limited to its connected component of the identity transformation.
Thus the non-identity adjoint actions of elementary transformations from~$G$
on vector fields spanning~$\mathfrak g$ are merely
\begin{equation*}\arraycolsep=0ex
\begin{array}{ll}
\mathscr D_*(T)D(f)=D\big(\hat T^{-1}_tf(\hat T)\big),             &\mathscr S_*(Y)S(\alpha)=S\big(\hat Y^{-1}_y\alpha(\hat Y)\big),\\[1ex]
\mathscr D_*(T)P(g)=P\big(\hat T^{-1/2}_tg(\hat T)\big),           &\mathscr S_*(Y)Z(\beta)=Z\big(\hat Y_y\beta(\hat Y)\big),\\[1ex]
\mathscr P_*(X^0)D(f)=D(f)+P\big(fX^0_t-\tfrac12f_tX^0\big),\qquad &\mathscr Z_*(V^0)S(\alpha)=S(\alpha)+Z(\alpha V^0_y-\alpha_yV^0),\\[1ex]
\mathscr I_*(\ve)P(g)=P(\ve g),
\end{array}
\end{equation*}
where~$\hat T$ and~$\hat Y$ are the inverses of the functions~$T$ and~$Y$, respectively.

Classifying of subalgebras of the algebra~$\mathfrak g$,
we implicitly use its representation as the direct sum $\mathfrak g=\mathfrak i_1\oplus\mathfrak i_2$ of its ideals
$\mathfrak i_1:=\langle D(f), P(g)\rangle$ and $\mathfrak i_2:=\langle S(\alpha), Z(\beta)\rangle$.
In fact, we weave complete lists of $G$-inequivalent subalgebras of dimension not greater than two
for the ideals~$\mathfrak i_1$ and~$\mathfrak i_2$
into an analogous list for the entire algebra~$\mathfrak g$.
The Goursat method for classifying discrete subgroups of direct products of continuous and discrete groups
was adapted by Patera, Winternitz and Zassenhaus in~\cite{pate1975a}
to classify subalgebras of Lie algebras that are direct sums of their ideals;
see also~\cite{wint2004a}.
Since we need to classify only one- and two-dimensional subalgebras of~$\mathfrak g$,
we do not follow this method precisely.

\begin{lemma}\label{lem:1DInequivSubalgs}
A complete list of $G$-inequivalent one-dimensional subalgebras of the algebra~$\mathfrak g$
is exhausted by the following algebras:
\begin{gather*}
\mathfrak s_{1.1}=\langle D(1)-S(1)\rangle,\quad
\mathfrak s_{1.2}=\langle D(1)+Z(1)\rangle,\quad
\mathfrak s_{1.3}=\langle D(1)     \rangle,\quad
\mathfrak s_{1.4}=\langle P(1)-S(1)\rangle,\\
\mathfrak s_{1.5}=\langle S(1)     \rangle,\quad
\mathfrak s_{1.6}=\langle P(1)+Z(1)\rangle,\quad
\mathfrak s_{1.7}=\langle P(1)     \rangle,\quad
\mathfrak s_{1.8}=\langle Z(1)     \rangle.
\end{gather*}
\end{lemma}

\begin{proof}
Let~$\mathfrak s_1$ be a one-dimensional subalgebra of~$\mathfrak g$,
which is spanned by a nonzero vector field
$Q^1=D(f^1)+S(\alpha^1)+P(g^1)+Z(\beta^1)$ from~$\mathfrak g$,
$\mathfrak s_1=\langle Q^1\rangle$,
where $f^1=f^1(t)$, $g^1=g^1(t)$, $\alpha^1=\alpha^1(y)$ and $\beta^1=\beta^1(y)$
are arbitrary smooth functions of their arguments that are not simultaneously zero.
If the function~$f^1=f^1(t)$ does not vanish, we use~$\mathscr D_*(T)$ with $T_t=1/f^1$ to set $f^1=1$
and then act on the (new) vector field~$Q^1$ by $\mathscr P_*(X^0)$ with $X^0_t=-g^1$ to set $g^1=0$.
Similarly, if the function~$\alpha^1=\alpha^1(y)$ does not vanish,
we can set $\alpha^1=1$ (or, equivalently, $\alpha^1=-1$) and $\beta^1=0$, successively acting on~$Q^1$ by~$\mathscr S_*(Y)$ with $Y_y=1/\alpha^1$
and on new~$Q^1$ by~$\mathscr Z_*(V^0)$ with $V^0_y=-\beta^1$.
If both the functions~$f^1$ and~$\alpha^1$ vanish, each  nonzero function from~$\{g^1,\beta^1\}$
can be set to~$1$ by~$\mathscr D_*(T)$ with $T_t=(g^1)^{-2}$ or~$\mathscr S_*(Y)$ with $Y_y=\beta^1$, respectively.
This is why the proof splits into the cases
\begin{gather*}
1.\ f^1\!\neq0,\ \alpha^1\!\ne0                        ;\quad
2.\ f^1\!\neq0,\ \alpha^1\!=0  ,\         \ \beta^1\!\ne0;\quad
3.\ f^1\!\neq0,\ \alpha^1\!=0  ,\         \ \beta^1\!=0  ;\\
4.\ f^1\!=0   ,\ \alpha^1\!\ne0,\ g^1\!\ne0            ;\quad
5.\ f^1\!=0   ,\ \alpha^1\!\ne0,\ g^1\!=0              ;\quad
6.\ f^1\!=0   ,\ \alpha^1\!=0  ,\ g^1\!\ne0,\ \beta^1\!\ne0;\\
7.\ f^1\!=0   ,\ \alpha^1\!=0  ,\ g^1\!\ne0,\ \beta^1\!=0  ;\quad
8.\ f^1\!=0   ,\ \alpha^1\!=0  ,\ g^1\!=0  ,\ \beta^1\!\ne0.
\end{gather*}
They lead to the subalgebras~$\mathfrak s_{1.1}$, \dots, $\mathfrak s_{1.8}$, respectively.
In each of these cases, we reduce the basis element~$Q^1$ to a maximum simple form,
pushing forward~$Q^1$ by elements of~$G$.
In fact, we do not need to change the basis of~$\mathfrak s_1$ here.
\end{proof}

\begin{lemma}\label{lem:2DInequivSubalgs}
A complete list of $G$-inequivalent two-dimensional subalgebras of the algebra~$\mathfrak g$
is exhausted by
the non-Abelian algebras
\begin{gather*}
\mathfrak s_{2.1}=\langle D(1)+S(1),D(t)+S(y)         \rangle,\\
\mathfrak s_{2.2}^\de=\langle D(1)+Z(\de),D(t)-S(y)         \rangle,\quad
\mathfrak s_{2.3}^\de=\langle D(1),D(t)+Z(\de)            \rangle,\\
\mathfrak s_{2.4}^\de=\langle S(1)-P(\de),D(2t)+S(y)        \rangle,\quad
\mathfrak s_{2.5}^\de=\langle S(1),S(y)+P(\de)            \rangle,\\
\mathfrak s_{2.6}^{\de\de'}=\langle P(\de)+Z(\de'),D(2t)-S(y)\rangle_{(\de,\de')\neq(0,0)},\\
\mathfrak s_{2.7}^\de=\langle P(1),D(2t)+Z(\de)       \rangle,\quad
\mathfrak s_{2.8}^\de=\langle Z(1),-S(y)+P(\de)       \rangle
\end{gather*}
and the Abelian algebras
\begin{gather*}
\mathfrak s_{2.9}^{\de\de'}=\langle D(1)+Z(\de),S(1)-P(\de')\rangle,\quad
\mathfrak s_{2.10}^{\de\de'}=\langle D(1)+S(1),P(\de)+Z(\de')\rangle_{(\de,\de')\neq(0,0)},\\
\mathfrak s_{2.11}^{\de\beta}=\langle D(1)+Z(\beta),P(\de)+Z(1)\rangle,\quad
\mathfrak s_{2.12}^\de=\langle D(1)+Z(\de),P(1)\rangle,\\
\mathfrak s_{2.13}^{g\de}=\langle S(1)+P(g),P(1)+Z(\de)\rangle,\quad
\mathfrak s_{2.14}^\de=\langle S(1)-P(\de),Z(1)\rangle,\\
\mathfrak s_{2.15}=\langle P(1),Z(1)\rangle,\quad
\mathfrak s_{2.16}^{g\de}=\langle P(1)+Z(\de),P(g)\rangle_{g\ne\const},\\
\mathfrak s_{2.17}^{\de\beta}=\langle P(\de)+Z(1),Z(\beta)\rangle_{\beta\ne\const},\quad
\mathfrak s_{2.18}^{g\beta}=\langle P(1)+Z(1),P(g)+Z(\beta)\rangle_{g,\beta\ne\const},
\end{gather*}
where $\beta=\beta(y)$ and $g=g(t)$ are arbitrary functions of their arguments that satisfy the indicated constraints,
and $\de,\de'\in\{0,1\}$.
\end{lemma}

\begin{proof}
An arbitrary subalgebra~$\mathfrak s$ of~$\mathfrak g$ is spanned by two linearly independent vector fields from~$\mathfrak g$,
\begin{gather*}
Q^1=D(f^1)+S(\alpha^1)+P(g^1)+Z(\beta^1),\\
Q^2=D(f^2)+S(\alpha^2)+P(g^2)+Z(\beta^2).
\end{gather*}
The classification procedure reduces to the consideration of different cases
of possible simplifications of the form of the basis elements~$Q^1$ and~$Q^2$
using basis changes and pushforwards of~$(Q^1,Q^2)$ by transformations from~$G$.
We also need to take into account the closedness of~$\mathfrak s$
with respect to the Lie bracket of vector fields, i.e.,
the condition $[Q^1,Q^2]=\langle Q^1,Q^2\rangle$.
Therefore, the above procedure essentially depends on
whether or not the subalgebra~$\mathfrak s$ is Abelian.

\medskip\par\noindent{\bf I.}\
Suppose that the subalgebra~$\mathfrak s$ is not Abelian.
Up to changing the basis $(Q^1,Q^2)$, we can assume that $[Q^1,Q^2]=Q^1$.
This commutation relation implies that for $f^2=0$ we have $f^1=g^1=0$ as well,
and similarly $\alpha^1=\beta^1=0$ if $\alpha^2=0$.
Hence  $(f^2,\alpha^2)\neq(0,0)$.

Analogously to the proof of Lemma~\ref{lem:1DInequivSubalgs},
in the case $f^1\ne0$, modulo the $G$-equivalence we can set $f^1=1$ and $g^1=0$.
Then we get from the commutation relation $[Q^1,Q^2]=Q^1$ that
$f^2=t+C_1$ and $g^2=C_2$ with constants~$C_1$ and~$C_2$.
The pushforward of vector fields by the transformation $\mathscr P(2C_2)\circ\mathscr D(t+C_1)$
preserves~$Q^1$ and sets $C_1=C_2=0$ in~$Q^2$.
If $f^1=0$ and $f^2\ne0$, then we choose $f^2=2t$ and $g^2=0$ up to the $G$-equivalence,
and thus $g^1$ is a constant in view of the commutation relation $[Q^1,Q^2]=Q^1$,
which allows us to set $g^1=\de\in\{0,1\}\,(\!{}\bmod G)$.
The condition $f^1=f^2=0$ gives that $g^1=0$,
and we set $g^2=\de\in\{0,1\}$ modulo the $G$-equivalence.

If $\alpha^1\ne0$, then we can set $\alpha^1=1$ and $\beta^1=0$ up to the $G$-equivalence,
and hence $\alpha^2=y+C_1$ and $\beta^2=C_2$ with constants~$C_1$ and~$C_2$
in view of the commutation relation $[Q^1,Q^2]=Q^1$.
Acting by $\mathscr Z_*(C_2)$ and $\mathscr S_*(y+C_1)$,
we preserve the form of~$Q^1$ and set $C_1=C_2=0$ in~$Q^2$.
If $\alpha^1=0$ and $\alpha^2\ne0$, then we choose $\alpha^2=-y$ and $\beta^2=0$ up to the $G$-equivalence,
and thus it follows from the commutation relation $[Q^1,Q^2]=Q^1$ that $\beta^1$ is a constant,
i.e., $\beta^1=\de\in\{0,1\}\,(\!{}\bmod G)$.
A~consequence of the condition $\alpha^1=\alpha^2=0$ is $\beta^1=0$,
and the $G$-equivalence allow us to make $\beta^2=\de\in\{0,1\}$.

As a result, the proof is partitioned into the following cases,
each of which results in a family of non-Abelian subalgebras from the lemma's statement:
\begin{gather*}
1.\ f^1\neq0,\ \alpha^1\neq0                     ;\\
2.\ f^1\neq0,\ \alpha^1=0,\ \alpha^2\neq0        ;\quad
3.\ f^1\neq0,\ \alpha^1=\alpha^2=0               ;\\
4.\ f^1=0,\ f^2\neq0,\ \alpha^1\neq0             ;\quad
5.\ f^1=f^2=0,\ \alpha^1\neq0                    ;\\
6.\ f^1=0,\ f^2\neq0,\ \alpha^1=0,\ \alpha^2\neq0;\\
7.\ f^1=0,\ f^2\neq0,\ \alpha^2=0                ;\quad
8.\ f^2=0,\ \alpha^1=0,\ \alpha^2\neq0.
\end{gather*}
Here and in what follows, the numeration of cases
corresponds to the numeration of subalgebra families in the lemma's statement.

\medskip\par\noindent{\bf II.}\
Let the subalgebra~$\mathfrak s$ be Abelian.
It follows from the commutation relations $[Q^1,Q^2]=0$ that $f^1f^2_t-f^1_tf^2=0$ and $\alpha^1\alpha^2_t-\alpha^1_t\alpha^2=0$,
i.e., the functions in each of the pairs $(f^1,f^2)$ and $(\alpha^1,\alpha^2)$ are linearly dependent.
Therefore, we can assume up to the $G$-equivalence that all these functions are constants.
We need to consider different cases depending, in particular,
on whether or not the tuples $(f^1,\alpha^1)$ and $(f^2,\alpha^2)$ are linearly dependent.

\medskip\par\noindent 9.\ $f^1\alpha^2-f^2\alpha^1\neq0$.
Linearly recombining the basic elements~$Q^1$ and~$Q^2$, we can reduce $(f^1,\alpha^1,f^2,\alpha^2)$ to $(1,0,0,1)$,
and then $g^1=0$ and $\beta^2=0$ (\!${}\bmod G$).
The commutation relation $[Q^1,Q^2]=0$ implies that the parameters~$\beta^1$ and~$g^2$ are arbitrary constants.
To set them to $(\delta,\delta')$, we act by the transformation
$\mathscr I\big(\sgn(\beta^1g^2)\big)\circ\mathscr D\big((\beta^1/g^2)^{2/3}t\big)\circ\mathscr S\big((\beta^1)^{1/3}(g^2)^{2/3}y\big)$
(resp.\ either $\mathscr I_*(\sgn g^2)\circ\mathscr D_*\big((g^2)^{-2}t\big)$ or $\mathscr S_*(\beta^1y)$)
if $\beta^1g^2\ne0$ (resp.\ either $\beta^1=0$ or $g^2=0$),
and we additionally scale the basis elements if necessary.

\medskip\par\noindent 10.\ $f^1\alpha^2-f^2\alpha^1=0,\ (f^1,f^2)\ne(0,0),\ (\alpha^1,\alpha^2)\ne(0,0)$.
Changing the subalgebra basis, we make $f^2=\alpha^2=0$, and thus $f^1\alpha^1\ne0$.
Then $g^1=0$ and $\beta^1=0$ (\!${}\bmod G$), $g^2_t=0$ and~$\beta^2_y=0$ in view of the commutativity of~$\mathfrak s$, and $(g^2,\beta^2)\ne(0,0)$.
In the case $g^2\beta^2=0$, scalings of~$Q^2$ suffices for reducing $\mathfrak s$ to the corresponding canonical form.
If $g^2\beta^2\ne0$, then we push forward the basic elements by
$\mathscr I_*(\ve)\circ\mathscr D_*\big((\beta^2/g^2)^{2/3}t\big)\circ\mathscr S_*\big((\beta^2/g^2)^{2/3}y\big)$
with $\ve=\sgn(\beta^2g^2)$ and scale them to set $(g^2,\beta^2)=(\delta,\delta')$.

\medskip\par\noindent 11--12.\ $(f^1,f^2)\ne(0,0),\ \alpha^1=\alpha^2=0$.
The pair $(f^1,f^2)$ can be reduced to $(1,0)$, and $g^1=0$ (\!${}\bmod G$).
The commutation relation  $[Q^1,Q^2]=0$ implies that $g^2$ is a constant,
and thus we can assume up to scalings of~$Q^2$ that $g^2=\delta$.
The further consideration splits into two cases depending on whether or not $\beta^2$ is nonzero,
which corresponds to the conditions $f^1\beta^2-f^2\beta^1\ne0$ and $f^1\beta^2-f^2\beta^1=0$
on the initial values of parameters.
Acting by the transformation $\mathscr Z(Y)$ with $Y_y=\beta^2$ in the first case,
we map the subalgebra $\mathfrak s$ to the subalgebra~$\mathfrak s_{2.11}^{\de\beta}$.
Otherwise, $g^2\ne0$, i.e., $g^2=1$, and we use the transformation~$\mathscr Z(Y)$ with $Y_y=\beta^1$
and obtain the subalgebra~$\mathfrak s_{2.12}^\de$.

\medskip\par\noindent 13--14.\ $f^1=f^2=0,\ (\alpha^1,\alpha^2)\ne(0,0)$.
Similarly to the previous case, separately considering the subcases
$g^1\alpha^2-g^2\alpha^1\ne0$ and $g^1\alpha^2-g^2\alpha^1=0$,
we obtain the subalgebras~$\mathfrak s_{2.13}^{g\de}$ and~$\mathfrak s_{2.14}^\de$, respectively.

\medskip\par
For all the other cases,
$f^1=f^2=0$ and $\alpha^1=\alpha^2=0$,
and thus the commutativity of~$\mathfrak s$ implies no more constraints on subalgebra's parameters.
Therefore, the consideration depends on whether or not the Wronskians ${\rm W}(g^1,g^2)$ and ${\rm W}(\beta^1,\beta^2)$ vanish.

\medskip\par\noindent 15.\ $g^1g^2_t-g^1_tg^2=0,\ \beta^1\beta^2_y-\beta^1_y\beta^2=0$.
We linearly recombine the basis elements to $Q^1=P(g^1)$ and $Q^2=Z(\beta^2)$,
where the pair $(g^1,\beta^2)$ can be reduced to $(1,1)$ via pushing forward the subalgebra~$\mathfrak s$
by $\mathscr D(T)\circ\mathscr S(Y)$ with $T_t=(g^1)^{-2}$ and $Y_y=\beta^2$.

\medskip\par\noindent 16.\ $g^1g^2_t-g^1_tg^2\ne0,\ \beta^1\beta^2_y-\beta^1_y\beta^2=0$.
Recombining the basis elements, we can obtain $\beta^2=0$.
The transformation $\mathscr D(T)\circ\mathscr S(Y)$ with $T_t=(g^1)^{-2}$ and $Y_y=\beta^1$
sets $g^1=1$ and $\beta^1=\de$.

\medskip\par\noindent 17.\ $g^1g^2_t-g^1_tg^2=0,\ \beta^1\beta^2_y-\beta^1_y\beta^2\ne0$.
This case is considered similarly to the previous case.
Only the roles of $P$ and~$Z$ should be swapped.

\medskip\par\noindent 18.\ $g^1g^2_t-g^1_tg^2\ne0,\ \beta^1\beta^2_y-\beta^1_y\beta^2\ne0$.
We get the canonical form of the subalgebra $\mathfrak s$ corresponding to this condition
by using the same transformation as in Case~16.
\end{proof}

\section{Solutions by the method of differential constraints}\label{sec:SolutionsByMethodOfDiffConstraints}

The method of differential constraints \cite{sido1984A,yane1964a} (or ``side conditions''~\cite{olve1986b})
is to attach additional differential equations, which are called differential constraints,
to the original system of differential equations
and then to integrate the obtained overdetermined system.
Instances of the general method of differential constraints are
the method of Lie reductions and its generalizations,
including the methods of nonclassical reductions~\cite{blum1969a}
and reductions with respect to generalized, conditional~\cite[Chapter~5]{fush1993A},
generalized conditional~\cite{zhda1995e} and weak~\cite{pucc1992b} symmetries;
see related discussions in \cite{boyk2016a,kunz2008b,kunz2009a}.
A disadvantage of the general method of differential constraints is
that, given a system of differential equations to be solved,
there is no efficient way for a priori selecting differential constraints that are compatible
with this system, except symmetry-related differential constraints.
Another problem is to integrate the obtained overdetermined system if it is compatible.

We consider several simple differential constraints for the system~\eqref{eq:BLPsystem}.
Each of them is invariant with respect to the pseudogroup~$G$,
and hence this pseudogroup is a point-symmetry pseudogroup for the corresponding overdetermined system.
A distinguished feature of these differential constraints is that
they single out subsets of solutions of the system~\eqref{eq:BLPsystem}
that are parameterized by solutions of linear differential equations.
We call such solution subsets \emph{linearizable}.

\medskip\par\noindent$\boldsymbol{v_x=0.}$
This is the condition for the solutions to the system~\eqref{eq:BLPsystem}
that are partially invariant%
\footnote{%
See~\cite{ovsi1982A} for the definition of partially invariant solutions.
}
of rank one and defect one with respect to the subalgebra~$\mathfrak s_{2.16}^{g0}$
with any nonconstant value of the parameter function $g=g(t)$, e.g., $g(t)=t$.
The joint system of~\eqref{eq:BLPsystem} and the constraint $v_x=0$
reduces to the system $(u_t-2uu_x+u_{xx})_y=0$, $v_t=v_x=0$, which integrates to $v=W(y)$ and
\begin{gather}\label{eq:InhomBackwardBurgersEq}
u_t-2uu_x+u_{xx}=F(t,x),
\end{gather}
where $F=F(t,x)$ and $W=W(y)$ are arbitrary (sufficiently smooth) functions of their argu\-ments.
Modulo the $G$-equivalence, we can set $W=0$, preserving~$F$, i.e., here $v=0$ $(\!{}\bmod G)$.
Solutions, for which additionally $u_y=0$
and thus $u$ is an arbitrary (sufficiently smooth) function of~$(t,x)$, are trivial.
The class constituted by the equations of the form~\eqref{eq:InhomBackwardBurgersEq},
where the tuple of arbitrary elements consists of the single parameter function~$F$, is normalized,
and its equivalence pseudogroup is induced by the point-symmetry pseudogroup~$G$ of the system~\eqref{eq:BLPsystem},
cf.\ \cite[Section~3]{poch2013d}.
For each fixed value of the parameter function~$F$,
the corresponding equation of the form~\eqref{eq:InhomBackwardBurgersEq} is an \emph{inhomogeneous (backward) Burgers equation},
where the independent variable~$y$ should be treated as a parameter.
The Hopf--Cole transformation $u=-\Phi_x/\Phi$ reduces this equation to
the \emph{linear backward heat equation with potential}
\begin{gather}\label{eq:BackwardHeatEqWithPot}
\Phi_t+\Phi_{xx}-H\Phi=0
\end{gather}
for the new unknown function $\Phi=\Phi(t,x,y)$, where the potential~$H$ is an antiderivative of~$F$ with respect to~$x$, $H_x=-F$.
Due to the indeterminacy of $\Phi$ up to an arbitrary nonvanishing multiplier depending on $(t,y)$,
we can gauge the potential~$H$ and assume that it is a function merely of~$(t,x)$.
Under this gauge, the function~$\Phi$ is still defined up to an arbitrary nonvanishing multiplier~$\theta=\theta(t)$,
which leads to equivalence transformations $(\tilde t,\tilde x,\tilde y)=(t,x,y)$, $\tilde\Phi=\theta\Phi$, $\tilde H=H+\theta_t/\theta$
between equations of the form~\eqref{eq:BackwardHeatEqWithPot}.
As a result, we obtain the class~\eqref{eq:BackwardHeatEqWithPot} of the (1+1)-dimensional linear backward heat equations with potentials,
where the potential~$H$ depends on~$(t,x)$, and thus only $t$ and~$x$ are the true independent variables
whereas $y$ is rather an implicit parameter.
This class is normalized as well, cf. \cite[Section~2]{popo2008a},
and the quotient of its equivalence pseudogroup~$G^\sim_{\rm bh}$ by the normal pseudosubgroup of the above transformations parameterized by~$\theta$
is induced by the pseudogroup~$G$.
Therefore, to construct $G$-inequivalent solutions of the system~\eqref{eq:BLPsystem} with the constraint $v_x=0$,
we should consider $G^\sim_{\rm bh}$-inequivalent equations of the form~\eqref{eq:BackwardHeatEqWithPot}
and for each of them, take solutions of this equation that are inequivalent with respect to its point-symmetry pseudogroup.
Alternating the sign of~$t$, one can transform backward heat equations to forward ones.
Wide families of closed-form solutions in terms of elementary and special functions
are known for (1+1)-dimensional linear heat equations with potentials,
especially for the classical heat equation ($H=0$).
For the latter equation, such solutions are comprehensively listed
in~\cite[Example 3.17]{olve1993A} and in~\cite{ivan2008b,widd1975A};
see also a review in \cite[Section~A]{vane2021a}.
For Kolmogorov equations that are equivalent to equations of the form~\eqref{eq:BackwardHeatEqWithPot} with $H=h(t)x^{-2}$,
infinite families of such solutions were constructed in~\cite{fush1994a,popop1995a}
using Lie reductions for constant~$h$ and recursive nonclassical reductions for general values of the parameter function~$h$.
For equations of the form~\eqref{eq:BackwardHeatEqWithPot} with other potentials~$H$,
exact solutions can be found using the Darboux transformation~\cite{matv1991A,popo2008a}.
In view of linearity of the equations discussed,
elements of an arbitrary finite set of solutions of the same equation
can be linearly combined with arbitrary constant coefficients,
and then all parameters of solutions can be considered to depend on~$y$.
Mapping the described solutions by the Hopf--Cole transformation,
one can construct huge families of exact solutions of the system~\eqref{eq:BLPsystem},
in particular, those parameterized by an arbitrary finite number of functions of~$(t,y)$.
These solutions are, modulo the $G$-equivalence, of the form\looseness=-1
\[
\solution u=-\frac{\Phi_x}\Phi,\quad v=0,
\]
where $\Phi=\Phi(t,x,y)$ is an arbitrary solution of any linear backward heat equations with potentials
from the class~\eqref{eq:BackwardHeatEqWithPot}.

\medskip\par\noindent$\boldsymbol{u_y=v_x.}$
This constraint means that (at least locally) the tuple $(u,v)$ is the gradient of a potential $\Psi=\Psi(t,x,y)$,
$u=\Psi_x$ and $v=\Psi_y$.
The substitution of this representation for $(u,v)$ into~\eqref{eq:BLPsystem}
leads to the equations $(\Psi_t-\Psi_x^{\,2}-\Psi_{xx})_{xy}=0$ and $(\Psi_t-\Psi_x^{\,2}-\Psi_{xx})_y=0$.
It is obvious that the former equation is a differential consequence of the latter one
and can hence be neglected.
The latter equation integrates to the class of \emph{inhomogeneous potential Burgers equations} of the form
\begin{gather}\label{eq:InhomPotentialBurgersEq}
\Psi_t-\Psi_x^{\,2}-\Psi_{xx}=-H(t,x),
\end{gather}
where $H=H(t,x)$ is an arbitrary (sufficiently smooth) function of~$(t,x)$.
Here again only $t$ and~$x$ are the true independent variables
whereas $y$ is an implicit parameter.
The substitution $\Psi=\ln\Phi$ establishes the similarity of this class
to the class of all \emph{linear heat equations with potentials}
\begin{gather}\label{eq:ForwardHeatEqWithPot}
\Phi_t-\Phi_{xx}+H\Phi=0
\end{gather}
with the same roles of~$(t,x)$ and $y$.
Analogously to the consideration of the constraint $v_x=0$,
this results in huge families of exact solutions of the system~\eqref{eq:BLPsystem}
in terms of the ``two-dimensional'' Hopf--Cole transformation,
\begin{gather}\label{eq:BLPSolutionFamilyViaHeatEq2}
\solution u=\frac{\Phi_x}\Phi,\quad v=\frac{\Phi_y}\Phi,
\end{gather}
where $\Phi=\Phi(t,x,y)$ is an arbitrary solution of any linear heat equation with potential
from the class~\eqref{eq:ForwardHeatEqWithPot}.
All the solutions found in Sections~V and~VI of~\cite{zhao2017a} are isolated instances
of such families.

The family of stationary solutions of~\eqref{eq:BLPsystem} that satisfy the differential constraint $u_y=v_x$
are expressed in terms of the general solution of the \emph{Liouville equation}.
Under the stationarity condition $u_t=v_t=0$, we can assume the potential~$\Psi$ is stationary as well, $\Psi_t=0$,
and the system~\eqref{eq:BLPsystem} reduces to the single equation $\Psi_{xxy}+2\Psi_x\Psi_{xy}=0$.
We integrate this equation as a linear first-order ordinary differential equation with respect to~$\Psi_{xy}$,
where $x$ is treated as the only independent variable, and $y$ is again assumed as a parameter.
This gives the equation $\Psi_{xy}=\beta(y){\rm e}^{-2\Psi}$ for~$\Psi$.
The parameter function $\beta=\beta(y)$ can be assumed nonvanishing
since otherwise the associated solution of~\eqref{eq:BLPsystem} is trivial.
Thus, we set $\beta=-\frac12$ by a transformation $\mathcal S(Y)$ trivially prolonged to~$\Psi$ and, denoting $\tilde\Psi=-2\Psi$,
obtain the classical Liouville equation \smash{$\tilde\Psi_{xy}={\rm e}^{\tilde\Psi}$} for~$\tilde\Psi$.
The well known formula for the general solution of the Liouville equation leads
via the representation $u=-\frac12\tilde\Psi_x$, $v=-\frac12\tilde\Psi_y$
to a family of solutions of the system~\eqref{eq:BLPsystem}
that is parameterized by an arbitrary (sufficiently smooth) function of~$x$,
where the second parameter function, which depends on~$y$, can be gauged by point symmetry transformations of~\eqref{eq:BLPsystem}.
We derive a representation for this family directly from the equation $(\Psi_{xx}+\Psi_x^{\,2})_y=0$.
The substitution~$\Psi=\ln\hat\Psi$ reduces it to $(\hat\Psi_{xx}/\hat\Psi)_y=0$,
which integrates to $\hat\Psi_{xx}=h(x)\hat\Psi$, where $h$ is an arbitrary (sufficiently smooth) function of~$x$.
The general solution of the last equation is represented in the form
$\hat\Psi=\alpha^1(y)\zeta^1(x)+\alpha^2(y)\zeta^2(x)$,
where $\zeta^1\zeta^2_x-\zeta^1_x\zeta^2=1$, and we assume that
$\alpha^1\alpha^2_y-\alpha^1_y\alpha^2\ne0$ (otherwise we have a trivial solution with $u_y=v=0$).
Without loss of generality, modulo the $G$-equivalence we can set $\alpha^1=1$ and~$\alpha^2=y$.
After denoting $\zeta:=\zeta^1/\zeta^2$, we obtain $\zeta_x=(\zeta^1)^{-2}$, and thus \smash{$\zeta^1=\zeta_x^{-1/2}$}.
In view of the representation $u=\hat\Psi_x/\hat\Psi$ and $v=\hat\Psi_y/\hat\Psi$,
the above implies that, up to the $G$-equivalence,
stationary solutions of the system~\eqref{eq:BLPsystem} with $u_y=v_x$ and $v\ne0$ are exhausted by
\[
\solution
u=-\frac{\zeta_{xx}(x)}{2\zeta_x(x)}+\frac{\zeta_x(x)}{y+\zeta(x)},\quad
v=\frac1{y+\zeta(x)},
\]
where $\zeta$ is an arbitrary (sufficiently smooth) function of~$x$.

\begin{remark}\label{rem:BLPsystemPainleveTest}
It is no coincidence that attaching each of the differential constraints $v_x=0$ and $v_x=u_y$
to the system~\eqref{eq:BLPsystem} leads to an integrable differential equation.
The Painlev\'e test of the system~\eqref{eq:BLPsystem} was made by Garagash~\cite{gara1994a}.
In its standard formulation,
it is assumed that solutions of the system of partial differential equations under consideration
can be represented in the form of a truncated formal Laurent series in powers of a smooth function.
For solutions the system~\eqref{eq:BLPsystem}, the following representation was initially used in~\cite{gara1994a}:
\begin{equation}\label{eq:LaurentSeriesForBLP}
u=\sum_{k=N_u}^\infty u^k(t,x,y)\big(\Phi(t,x,y)\big)^k,\quad v=\sum_{k=N_v}^\infty v^k(t,x,y)\big(\Phi(t,x,y)\big)^k
\end{equation}
in a neighborhood of a smooth two-dimensional surface~$S$,
where $\Phi$ is an arbitrary smooth function of~$(t,x,y)$ that has first-order zeros along~$S$,
the expansion coefficients $u^k$ and $v^k$ are of smooth function of~$(t,x,y)$,
and $N_u$ and~$N_v$ are integers.
For two possible expansion branches, $(N_u,N_v)=(-1,-1)$ and $(N_u,N_v)=(-1,0)$,
the series~\eqref{eq:LaurentSeriesForBLP} can simultaneously be truncated only for the solutions of~\eqref{eq:BLPsystem}
respectively satisfying the differential constraints $v_x=u_y$ and $v_x=0$, which have been studied above.
To prevent additional differential constraints in the course of series truncation,
the Painlev\'e test was modified in~\cite{gara1994a}.
Namely, the function~$u$ was further re-expanded with respect to a function~$\Psi$ instead of~$\Phi$,
where $\Psi=\Phi\sum_{k=0}^\infty \Psi^k(t,x,y)(\Phi)^k$, and $\Psi^k$, $k=0,1,\dots$, are smooth functions of~$(t,x,y)$ with $\Psi^0\ne0$,
i.e., $\Psi$ has first-order zeros along the same surface~$S$ as~$\Phi$.
This allowed Garagash to construct two Lax pairs and two B\"acklund transformations for the system~\eqref{eq:BLPsystem}.
\end{remark}

\medskip\par\noindent$\boldsymbol{u_y=0.}$
This constraint singles out the solutions to~\eqref{eq:BLPsystem}
that are partially invariant with respect to the subalgebra~$\mathfrak s_{2.5}^0$
with rank one and defect one.
The differential consequence $v_{xxx}=0$ of the first equation of~\eqref{eq:BLPsystem}
under the constraint $u_y=0$ implies $v=v^2(t,y)x^2+v^1(t,y)x+v^0(t,y)$.
Since the constraint $v_x=0$ has been considered above, here we assume that $v_x\ne0$, i.e., $(v^1,v^2)\ne(0,0)$.
Then from the second equation of the system~\eqref{eq:BLPsystem} we derive an expression for~$u$,
\begin{gather}\label{eq:ExprForUIfV_xxx=0}
u=\frac{v^2_tx^2+v^1_tx+v^0_t-2v^2}{2(2v^2x+v^1)}.
\end{gather}
We substitute this expression into the constraint $u_y=0$ and split the resulting equation with respect to~$x$.
This gives an overdetermined system of four differential equations for $(v^0,v^1,v^2)$.
In particular, we have the equation $v^2v^2_{ty}=v^2_tv^2_y$, which integrates to $v^2=v^{20}(t)v^{21}(y)$.

Any solution for which additionally $v^2=0$ is $G$-equivalent to the trivial solution
\[\solution u=0,\quad v=x.\]

Suppose that $v^2\ne0$, and thus both the functions~$v^{20}$ and~$v^{21}$ do not vanish.
We can assume without loss of generality that $v^{20}>0$.
Using equivalence transformations from the pseudosubgroups $\{\mathscr D(T)\}$ and $\{\mathscr S(Y)\}$,
we respectively set $v^{20}=1$ and $v^{21}=1$, i.e., $v^2=1$ $(\!{}\bmod G)$.
Then, another equation is simplified to $v^1_{ty}=0$.
Therefore, $v^1_t=0$ $(\!{}\bmod\{\mathscr P(X^0)\})$,
which allows us to make more simplifications and derive the equation $v^0_{ty}=0$,
giving $v^0_y=0$ $(\!{}\bmod\{\mathscr Z(V^0)\})$.
The last simplified equation is $v^1_y(v^0_t-2)=0$.
As a result, we derive two families of $G$-inequivalent solutions,
\[
\solution u=0,\ v=x^2+\zeta(y)x+2t
\qquad\mbox{and}\qquad\hspace{3ex}
\solution u=\frac{\theta_t(t)-1}{2x},\ v=x^2+2\theta(t),
\]
where $\zeta$ and~$\theta$ are arbitrary (sufficiently smooth) functions of their arguments.

\medskip\par
To exclude the already considered cases, we further assume that $v_xu_y\ne0$.

\medskip\par\noindent$\boldsymbol{v_{xxx}=0.}$
That is, $v=v^2(t,y)x^2+v^1(t,y)x+v^0(t,y)$.
Since the constraint $v_x=0$ has been considered above, here we assume that $v_x\ne0$, i.e., $(v^1,v^2)\ne(0,0)$.
Then the second equation of the system~\eqref{eq:BLPsystem} implies the expression~\eqref{eq:ExprForUIfV_xxx=0} for~$u$,
and hence we can represent~$u$ in the form
\[
u=u^1(t,y)x+u^0(t,y)+\frac{w^1(t,y)}{x+w^0(t,y)}.
\]
Substituting this representation into the first equation of the system~\eqref{eq:BLPsystem}
and splitting with respect to~$x$, we derive a system for the coefficients~$u^0$, $u^1$, $w^0$ and~$w^1$,
\begin{gather}\label{eq:V_xxx=0System}
\begin{split}
&\big(u^1_t-2(u^1)^2\big)_y=0,\quad
 \big(u^0_t-2u^1u^0\big)_y=0,\quad
 \big(2w^1(u^0-u^1w^0)-w^1w^0_t\big)_y=w^1_tw^0_y,\\
&w^1_{ty}=0,\quad
 \big(2w^1(u^0-u^1w^0)-w^1w^0_t\big)w^0_y=\big(w^1(w^1+1)\big)_y,\quad
 w^1(w^1+1)w^0_y=0.
\end{split}
\end{gather}
The last two equations imply that $w^1_y=0$.
We exhaustively solve the system~\eqref{eq:V_xxx=0System} under the assumption $u_y\ne0$,
analyzing different cases.
Then we substitute each of the obtained solutions jointly with the representation for~$v$
into the second equation of the system~\eqref{eq:BLPsystem},
split the resulting equation with respect to~$x$
and integrate the derived system to find $(v^0,v^1,v^2)$ when \mbox{$(v^1,v^2)\ne(0,0)$}.
Whenever it is possible, we use transformations from the pseudogroup~$G$
for gauging parameter functions arising in the course of integration,
which leads to simplifying the form of constructed exact solutions.

Let $u^1_y\ne0$.
The first equation of the system~\eqref{eq:V_xxx=0System} integrates once with respect to~$y$
to the Riccati equation $u^1_t-2(u^1)^2=\zeta(t)$,
where $\zeta$ is an arbitrary (sufficiently smooth) function of~$t$.
Here and in what follows, the independent variable~$y$ plays the role of an implicit parameter.
The substitution $u^1=-\varphi_t/(2\varphi)$ with $\varphi=\varphi(t,y)$ reduces this Riccati equation
to the second-order linear ordinary differential equation $\varphi_{tt}+2\zeta\varphi=0$
with respect to~$\varphi$.
The general solution of the latter equation takes the form
$\varphi=\beta^1(y)\varphi^1(t)+\beta^2(y)\varphi^2(t)$,
where $\beta^1$ and~$\beta^2$ are arbitrary (sufficiently smooth) functions of~$y$,
and $\varphi^1$ and~$\varphi^2$ are functions of~$t$ satisfying the conditions
$\varphi^1\varphi^2_t-\varphi^1_t\varphi^2=1$
and $\zeta=\frac12(\varphi^1_t\varphi^2_{tt}-\varphi^1_{tt}\varphi^2_t)$.
Without loss of generality, we can assume that $\beta^2\ne0$,
and $\varphi^2$ is a parameter function instead of~$\zeta$.
Then we set $\varphi^2=t$ using a transformation~$\mathscr D(T)$,
and thus $\varphi^1=1$ (up to combining $\varphi^1$ with~$\varphi^2$), $\zeta=0$,
and $u^1=-\frac12(t+\beta(y))^{-1}$ with $\beta:=\beta^1/\beta^2$.
In the case $w^1=0$, we need to additionally solve
only the second equation among those from the system~\eqref{eq:V_xxx=0System},
which gives, up to transformations from~$\{\mathscr P(X^0)\}$, that
$u^0=-\frac12\alpha(y)(t+\beta(y))^{-1}$
with an arbitrary (sufficiently smooth) function~$\alpha$ of~$y$.
If $w^1\ne0$ and $w^0_y=0$, then we set $w^0=0$ using an equivalence transformation~$\mathscr P(X^0)$,
and the system~\eqref{eq:V_xxx=0System} degenerates to the equations $u^0_y=u^0u^1_y=0$,
i.e., $u^0=0$ in view of the inequality $u^1_y\ne0$.
For $w^1\ne0$ and $w^0_y\ne0$, the system~\eqref{eq:V_xxx=0System} implies
that $w^1=-1$, $u^0=-\frac12\alpha(y)(t+\beta(y))^{-1}$ $(\!{}\bmod \{\mathscr P(X^0)\})$,
and then $w^0=\alpha(y)+\gamma(y)(t+\beta(y))$,
where $\alpha$ and~$\gamma$ are arbitrary (sufficiently smooth) functions of~$y$.

Now suppose that $u^1_y=0$. We set $u^1=0$ using a transformation~$\mathscr D(T)$ from~$G$.
Then, the second equation of the system~\eqref{eq:V_xxx=0System} takes the form $u^0_{ty}=0$.
It is obvious that $u=\alpha(y)$ $(\!{}\bmod \{\mathscr P(X^0)\})$ if $w^1=0$,
where $\alpha$ is an arbitrary (sufficiently smooth) function of~$y$.
All solutions with $w^1\ne0$ and $w^0_y=0$ additionally satisfy the constraint $u_y=0$,
and hence they are neglected here.
If $w^1\ne0$ and $w^0_y\ne0$, then the system~\eqref{eq:V_xxx=0System} reduces to
the system $w^1=-1$, $u^0=\frac12w^0_t$, $w^0_{tty}=0$,
and thus $w^0=2\alpha(y)t+\beta(y)$ $(\!{}\bmod \{\mathscr P(X^0)\})$.

As a result, we construct the following $G$-inequivalent families of solutions of the system~\eqref{eq:BLPsystem}:
\begin{gather}
\solution\label{eq:BLPsystemVxxxSol1}
u=-\frac12\frac{x+\alpha(y)}{t+\beta(y)},\quad
v=\delta\left(\frac{x+\alpha(y)}{t+\beta(y)}\right)^2+\gamma(y)\frac{x+\alpha(y)}{t+\beta(y)}-\frac{2\delta}{t+\beta(y)},
\\[1ex]\solution\label{eq:BLPsystemVxxxSol2}
u=-\frac12\frac x{t+\beta(y)}+\frac{\theta(t)}x,\quad
v=\frac{x^2}{(t+\beta(y))^2}+2\int_{t_0}^t\frac{2\theta(t')+1}{(t'+\beta(y))^2}\,{\rm d}t',
\\[1ex]\label{eq:BLPsystemVxxxSol3}
\solution
u=-\frac12\frac{x\!+\!\alpha(y)}{t\!+\!\beta(y)}-\frac1{x\!+\!\alpha(y)\!+\!\gamma(y)(t\!+\!\beta(y))},\quad
v=\left(\frac{x\!+\!\alpha(y)}{t\!+\!\beta(y)}+\gamma(y)\right)^2\!\!+\frac2{t+\beta(y)},
\\[1.5ex]\solution\label{eq:BLPsystemVxxxSol4}
u=\alpha(y),\quad
v=\alpha(y)(x+2\alpha(y)t)^2+\gamma(y)(x+2\alpha(y)t)-x,
\\[1.5ex]\solution\label{eq:BLPsystemVxxxSol5}
u=\alpha(y)-\frac1{x+2\alpha(y)t+\beta(y)},\quad
v=(x+2\alpha(y)t+\beta(y))^2-2t.
\end{gather}
Here $\alpha$, $\beta$, $\gamma$ and~$\theta$ are arbitrary functions of their arguments.
In the first solution, $\delta\in\{0,1\}$ 
and, if $\delta=0$, $\gamma=1$ $(\!{}\bmod \{\mathscr S(Y)\})$.

\medskip\par\noindent$\boldsymbol{u_{xx}=0,\ v_{4x}=0.}$
This two-component differential constraint integrates to the generalized ansatz
\[
u=u^1(t,y)x+u^0(t,y),\quad v=\sum_{k=0}^3v^k(t,y)x^k
\]
for $(u,v)$. Substituting the ansatz into the system~\eqref{eq:BLPsystem}
and splitting the resulting equations with respect to~$x$,
we derive a system for the coefficients~$u^0$, $u^1$ and~$v^k$, $k=0,\dots,3$,
\begin{gather}\label{eq:U_xx=0V4x=0System}
\begin{split}
&\big(u^1_t-2(u^1)^2\big)_y=0,\quad
 \big(u^0_t-2u^1u^0\big)_y=12v^3,\\
&v^k_t=2ku^1v^k+2(k+1)u^0v^{k+1}+(k+2)(k+1)v^{k+2},\quad k=0,\dots,3,\quad v^4,v^5:=0.
\end{split}
\end{gather}
In other words, we reduce the system~\eqref{eq:BLPsystem} to the well-defined system~\eqref{eq:U_xx=0V4x=0System}
with fewer independent variables but more unknown functions.
Such a property is common for reductions with respect to generalized conditional symmetries
(here the corresponding generalized conditional symmetry could be
the generalized vector field $u_{xx}\p_u+v_{4x}\p_v$)
but, unfortunately, a proper theory of these symmetries has been developed
only for single (1+1)-dimensional evolution equations \cite{kunz2011a,zhda1995e}.
Since the solutions of the system~\eqref{eq:BLPsystem} with the differential constraint $v_{xxx}=0$
have exhaustively been described above, we can assume that $v_{xxx}\ne0$ here, i.e., $v^3\ne(0,0)$.

If additionally $u_{xy}=u^1_y=0$, then $u^1=0$ $(\!{}\bmod \{\mathscr D(T)\})$,
the system~\eqref{eq:U_xx=0V4x=0System} implies $v^3_t=u^0_{ty}=0$, and
we can set $v^3=1$ and $u^0=12ty+\alpha(y)$ using transformations~$\mathscr S(Y)$ and~$\mathscr P(X^0)$.
Otherwise, we successively integrate some equations of~\eqref{eq:U_xx=0V4x=0System}
and simultaneously simplify the obtained expressions using equivalence transformations if possible;
in this way we get
\begin{itemize}
\item[] $u^1=-\frac12(t+\beta(y))^{-1}$ $(\!{}\bmod \{\mathscr D(T)\})$ from the first equation,
\item[] $v^3=\beta_y(y)(t+\beta(y))^{-3}$ $(\!{}\bmod \{\mathscr S(Y)\})$ from the equation with $k=4$, and
\item[] \par$u^0=-\frac12\big(\alpha(y)+12\ln|t+\beta(y)|-12\big)(t+\beta(y))^{-1}$ $(\!{}\bmod \{\mathscr P(X^0)\})$ from the second equation;
\end{itemize}
see the above consideration of the differential constraint $v_{xxx}=0$ with $u^1_y\ne0$
for deriving the expression for~$u^1$.
Integrating the other equations from~\eqref{eq:U_xx=0V4x=0System} in each of the above cases
and substituting the found values for $u^0$, $u^1$ and $v^k$ into the ansatz,
we respectively construct the following solutions of the system~\eqref{eq:BLPsystem}:
\begin{gather*}\solution
u=12ty+\alpha(y),\\
v=\big(x+12t^2y+2t\alpha(y)+\beta(y)\big)^3+\big(6t+\gamma(y)\big)\big(x+12t^2y+2t\alpha(y)+\beta(y)\big),
\\[1.5ex]\solution
u=-\frac\omega2+\frac6{t+\beta(y)},\quad
v=\beta_y(y)\omega^3+\gamma(y)\omega^2+\lambda(y)\omega-6\frac{\beta_y(y)\omega}{t+\beta(y)}-\frac{2\gamma(y)}{t+\beta(y)},
\\[.5ex]\phantom{\solution}\mbox{where}\quad
\omega:=\frac{x+\alpha(y)}{t+\beta(y)}+12\frac{\ln|t+\beta(y)|}{t+\beta(y)}.
\end{gather*}
Here $\alpha$, $\beta$, $\gamma$ and~$\lambda$ are arbitrary functions of~$y$.

\medskip\par\noindent$\boldsymbol{u_{xx}=0.}$
We weaken the previous differential constraint $u_{xx}=0$, $v_{4x}=0$, neglecting its second equation.
The first equation of the system~\eqref{eq:V_xxx=0System} implies $v_{5x}=0$, i.e.,
in this case we have an analogous generalized ansatz with one more unknown function~$v^4$,
\begin{gather}\label{eq:U_xx=0Ansatz}
u=u^1(t,y)x+u^0(t,y),\quad v=\sum_{k=0}^4v^k(t,y)x^k.
\end{gather}
The corresponding reduced system in~$u^0$, $u^1$ and~$v^k$, $k=0,\dots,4$,
is an extension of the system~\eqref{eq:U_xx=0V4x=0System},
\begin{gather}\label{eq:U_xx=0System}
\begin{split}
&\big(u^1_t-2(u^1)^2\big)_y=48v^4,\quad
 \big(u^0_t-2u^1u^0\big)_y=12v^3,\\
&v^k_t=2ku^1v^k+2(k+1)u^0v^{k+1}+(k+2)(k+1)v^{k+2},\quad k=0,\dots,4,\quad v^5,v^6:=0.
\end{split}
\end{gather}
Similarly to the system~\eqref{eq:U_xx=0V4x=0System},
the system~\eqref{eq:U_xx=0System} is well defined,
and hence the generalized vector field $u_{xx}\p_u+v_{5x}\p_v$ also can be interpreted
as a generalized conditional symmetry of the original system~\eqref{eq:BLPsystem}.
In view of the above consideration of the constraint $u_{xx}=v_{4x}=0$,
here we can assume that $v_{4x}\ne0$, i.e., $v^4\ne0$ and thus $u^1_y\ne0$.
In contrast to all the previous differential constraints,
we could not find the general solution of the system~\eqref{eq:U_xx=0System}
in an explicit form but we reduce it to a family of \emph{Abel equations of the first kind},
where $y$ is the only independent variable, and $t$ plays the role of an implicit parameter.

We make the differential substitution $u^1=-\varphi_t/(2\varphi)$, $u^2=\psi_t\varphi$
with $\varphi=\varphi(t,y)\ne0$ and $\psi=\psi(t,y)$.
Then the equation of the system~\eqref{eq:U_xx=0System} with $k=4$, $v^4_t=8u^1v^4$,
can be integrated to $v^4=\alpha^4(y)\varphi^{-4}$,
where $\alpha^4=\alpha^4(y)\ne0$ and thus $\alpha^4=1$ $(\!{}\bmod \{\mathscr S(Y)\})$,
i.e., we can set $v^4=\varphi^{-4}$.
Solving the equation of the system~\eqref{eq:U_xx=0System} with $k=3$, $v^3_t=6u^1v^3+8u^0v^4$,
with respect to~$v^3$, we derive the representation $v^3=8\psi\varphi^{-3}$,
where the ``constant of integration'' $\alpha^3=\alpha^3(y)$ is incorporated into~$\psi$
due to the indeterminacy of~$\psi$ up to a summand depending on~$y$.
The first two equations of the system~\eqref{eq:U_xx=0System} take the form
\begin{gather}\label{eq:U_xx=0SystemModifiedFirstTwoEqs}
\left(\frac{\varphi_{tt}}\varphi\right)_y=-\frac{96}{\varphi^4},\quad
(\psi_{tt}\varphi+2\psi_t\varphi_t)_y=96\frac\psi{\varphi^3}.
\end{gather}
Rewriting the first equation of~\eqref{eq:U_xx=0SystemModifiedFirstTwoEqs} as
$(\varphi\varphi_{ty}-\varphi_t\varphi_y)_t=-96\varphi^{-2}$,
we integrate it to
$\varphi\varphi_{ty}-\varphi_t\varphi_y=-96\chi$,
where $\chi$ is an antiderivative of~$\varphi^{-2}$ with respect to~$t$, $\chi_t=\varphi^{-2}$.
We divide the integration result by~$\varphi^2$
to represent it as $(\varphi_y/\varphi)_t=-96\chi\chi_t$ via replacing $\varphi^{-2}$ with~$\chi_t$ on the right hand side
and to integrate once more with respect to~$t$, which gives $\varphi_y/\varphi=-48\chi^2-\mu^1/2$ with $\mu^1=\mu^1(y)$.
We repeat the trick with dividing by~$\varphi^2$ and replacing $\varphi^{-2}$ with~$\chi_t$ for the last equation,
deriving the equation $\chi_{ty}=(96\chi^2+\mu^1)\chi_t$,
and integrating it again with respect to~$t$.
Thus, we obtain $\chi_y=32\chi^3+\mu^1(y)\chi+\mu^0(y)$.
In other words, the function~$\chi=\chi(t,y)$ runs through the union of the solution sets
of the Abel equations of the first kind of the above form
parameterized by two arbitrary sufficiently smooth functions~$\mu^0$ and~$\mu^1$ of~$y$,
where $t$ plays the role of an implicit parameter, and additionally $\chi_t>0$.
The substitution $\psi=\tilde\psi/(2\varphi)$ reduces
the second equation of~\eqref{eq:U_xx=0SystemModifiedFirstTwoEqs},
in view of the first equation of~\eqref{eq:U_xx=0SystemModifiedFirstTwoEqs},
to $(\varphi\tilde\psi_{ty}-\varphi_t\tilde\psi_y)_t=0$.
We integrate the last equation once with respect to~$t$, divide the result by~$\varphi^2$
and integrate once more with respect to~$t$, taking into account that $\varphi^{-2}=\chi_t$.
Then we solve the obtained expression with respect to~$\tilde\psi_y$ and integrate with respect to~$y$.
This leads to the expression for~$\tilde\psi$.
Finally, we successively integrate the equations from the system~\eqref{eq:U_xx=0System}
with $k=2$, $k=1$ and~$k=0$ with respect to~$v^2$, $v^1$ and~$v^0$, respectively,
\begin{gather*}
v^2=\frac1{\varphi^2}\big(24\psi^2+12\chi+\alpha^2\big),\quad
v^1=\frac1\varphi\big(32\psi^3+48\chi\psi+4\alpha^2\psi+\alpha^1\big),\\
v^0=16\psi^4+48\chi\psi^2+4\alpha^2\psi^2+2\alpha^1\psi+12\chi^2+2\alpha^2\chi+\alpha^0.
\end{gather*}
Here $\alpha^0$, $\alpha^1$ and $\alpha^2$ are arbitrary sufficiently smooth functions of~$y$,
and $\alpha^0=0$ $(\!{}\bmod\{\mathscr Z(V^0)\})$.
The substitution of the found expressions for~$u^0$, $u^1$ and~$v^k$, $k=0,\dots,4$
into the ansatz~\eqref{eq:U_xx=0Ansatz} gives the following solutions of the system~\eqref{eq:BLPsystem}:
\begin{gather*}\solution
u=-\frac{\omega_t}{2\sqrt{\chi_t(t,y)}},
\\[.5ex]\phantom{\solution}
v=\omega^4+\big(12\chi(t,y)+\alpha^2(y)\big)\omega^2+\alpha^1(y)\omega+2\chi(t,y)\big(6\chi(t,y)+\alpha^2(y)\big),
\\[.5ex]\phantom{\solution}
\mbox{where}\quad
\omega:=\frac{x+\tilde\psi(t,y)}{\sqrt{\chi_t(t,y)}},\quad
\tilde\psi(t,y):=\int_{y_0}^y\frac{\lambda^1(y')\chi(t,y')+\lambda^0(y')}{\sqrt{\chi_t(t,y')}}\,{\rm d}y'.
\end{gather*}
Here $\chi=\chi(t,y)$ is an arbitrary solution of the equation
\begin{gather}\label{eq:AbelOf1stKind}
\chi_y=32\chi^3+\mu^1(y)\chi+\mu^0(y)
\end{gather}
with $\chi_t>0$ and $t$ playing the role of an implicit parameter,
and $\alpha$, $\mu^0$, $\mu^1$, $\lambda^0$ and~$\lambda^1$ are arbitrary functions of~$y$.
For each value of the parameter function tuple $(\mu^1,\mu^0)$,
the equation~\eqref{eq:AbelOf1stKind} is an Abel equation of the first kind.

If $\mu^0=0$, then the equation~\eqref{eq:AbelOf1stKind} is a \emph{Bernoulli equation},
which can be easily integrated to
$\chi=\frac18\sqrt{\beta_y(y)/(\theta(t)-\beta(y))}$,
where $\theta_t\ne0$ and $\beta_y\ne0$ in view of $\chi_t\ne0$, and $\mu^1=\beta_{yy}/(4\beta_y)$.
Hence $\theta=t$ $(\!{}\bmod \{\mathscr D(T)\})$.
As a result, we construct a subfamily of the above family of solutions of the original system~\eqref{eq:BLPsystem}
in a completely explicit form, where
\[
\chi=\frac18\sqrt{\frac{\beta_y(y)}{t-\beta(y)}}.
\]

One can consider other differential constraints, including $G$-invariant ones,
but most of them do not lead to new explicit solutions of the system~\eqref{eq:BLPsystem}.

For instance, consider the $G$-invariant differential constraint
$v_{xt}u_y=v_xu_{yt}$, $v_{xy}u_y=v_xu_{yy}$.
We can suppose here that $u_yv_x\ne0$ since
the entire set of solutions of the system~\eqref{eq:BLPsystem}
under the $G$-invariant differential constraint $u_y=0$ (resp.\ $v_x=0$)
has been constructed above.
Then we can write the constraint $v_{xt}u_y=v_xu_{yt}$, $v_{xx}u_y=v_xu_{yx}$
in the form $(v_x/u_y)_t=(v_x/u_y)_x=0$, which obviously integrates to
the $G$-invariant family of differential constraints $\alpha(y)u_y=v_x$,
where $\alpha$ runs through the set of nonzero sufficiently smooth functions of~$y$.
Without loss of generality we can assume $\alpha\ne1$
since otherwise the differential constraint under consideration
coincides with the studied differential constraint $u_y=v_x$.
We introduce the potential $\Phi=\Phi(t,x,y)$ such that $u=\Phi_x$ and $v=\alpha\Phi_y$.
The system~\eqref{eq:BLPsystem} with substituted~$(u,v)$ implies that $\Phi_{xxxy}=0$
or, equivalently, $v_{xxx}=0$. The solutions satisfying the last constraint
have been considered above.

\medskip\par\noindent$\boldsymbol{v=u.}$
Although the constraint $v=u$ is not $G$-invariant,
yet the use of it leads to new solu\-tions of the system~\eqref{eq:BLPsystem}
that do not belong to the solution families satisfying the differential constraints considered above.
Indeed, in view of this constraint, the second equation of~\eqref{eq:BLPsystem} becomes the Burgers equation $u_t=u_{xx}+2uu_x$.
Differentiating it with respect to~$y$ and subtracting the result from the first equation of the system~\eqref{eq:BLPsystem},
in view of $v=u$ we obtain the equation $u_{xxx}=u_{xxy}$.
The general solution of the last equation is
\[
u=\chi(t,z)+u^1(t,y)z+u^0(t,y),
\]
where $z:=x+y$,
and $\chi$, $u^0$ and~$u^1$ are arbitrary sufficiently smooth functions of their arguments.
Since $\chi_{zzz}=u_{xxx}=v_{xxx}$, for finding new solutions it suffices to study the case $\chi_{zzz}\ne0$.
Similarly, $(u^0_y,u^1_y)\ne(0,0)$ since otherwise $u$ is a function of~$(t,z)$ and thus $u_y=u_x=v_x$.
Under the last inequality, the compatibility conditions of the system $u_t=u_{xx}+2uu_x$ and $u_{xxx}=u_{xxy}$,
which were computed using the package {\tt rif} \cite{reid1996a} for {\tt Maple},
imply the equation $(u^0_y/u^1_y)_y=0$
integrating to $u^0=g^1(t)u^1+g^0(t)$ with arbitrary functions $g^0$ and $g^1$ of~$t$, i.e., $u=\chi+u^1(z+g^1)+g^0$.
Hence $g^1=0$ $(\!{}\bmod \{\mathscr P(X^0)\})$, $g^0$ can be absorbed by the function~$\chi$, and thus $u=\chi+u^1z$.
Then another consequence of the compatibility conditions is the equation $z^2\chi_z=z\chi-2$
with the general solution $\chi=z^{-1}+g(t)z$, where the ``integration constant''~$g$ can be absorbed by the~$u^1$.
In other words, up to the $G$-equivalence we can set $\chi=z^{-1}$
and derive one more consequence of the compatibility conditions, $u^1_t=2(u^1)^2$.
Therefore, $u^1=(-2t+\alpha(y))^{-1}$,~i.e.,
\begin{equation}\label{eq:DiffConsrU=V}
\solution
u=v=\frac1{x+y}+\frac{x+y}{-2t+\alpha(y)},
\end{equation}
where~$\alpha$ is an arbitrary function of~$y$.

\section{Lie reductions of codimension one}\label{sec:LieReductionsOfCodim1}

$G$-inequivalent codimension-one Lie reductions of the system~\eqref{eq:BLPsystem}
are collected in Table~\ref{tab:LieReductionsOfCodim1}.
There, for each of the one-dimensional subalgebras of~$\mathfrak g$ listed in Lemma~\ref{lem:1DInequivSubalgs},
we present an ansatz constructed for $(u,v)$,
the corresponding reduced system of partial differential equations in two independent variables,
where $\varphi=\varphi(z_1,z_2)$ and $\psi=\psi(z_1,z_2)$
are the new unknown functions of the invariant independent variables~$(z_1,z_2)$.
The subscripts~1 and~2 of~$\varphi$ and~$\psi$ denote the differentiation with respect to~$z_1$ and~$z_2$, respectively.

\begin{table}[!ht]
\begin{center}
\caption{\footnotesize $G$-inequivalent Lie reductions with respect to one-dimensional subalgebras of~$\mathfrak g$
\strut}\label{tab:LieReductionsOfCodim1}
\renewcommand{\arraystretch}{1.6}
\begin{tabular}{|l|c|c|c|c|l|}
\hline
\hfil $\subset\mathfrak g$ &  $u$   &  $v$   &  $z_1$ &  $z_2$           &\hfil Reduced system      \\
\hline
$\mathfrak s_{1.1}$ & $\varphi$ & $\psi  $ & $t+y$ & $x$   & $\varphi_{11}=(\varphi^2-\varphi_2)_{12}+2\psi_{222}$, \ $\psi_1=\psi_{22}+2\varphi\psi_2$  \\
$\mathfrak s_{1.2}$ & $\varphi$ & $\psi+t$ & $x  $ & $y$   & $(\varphi^2-\varphi_1)_{12}+2\psi_{111}=0$, \  $\psi_{11}+2\varphi\psi_1=1$    \\
$\mathfrak s_{1.3}$ & $\varphi$ & $\psi  $ & $x  $ & $y$   & $(\varphi^2-\varphi_1)_{12}+2\psi_{111}=0$, \ $\psi_{11}+2\varphi\psi_1=0$ \\
$\mathfrak s_{1.4}$ & $\varphi$ & $\psi  $ & $t  $ & $x+y$ & $\varphi_{12}=(\varphi^2-\varphi_2)_{22}+2\psi_{222}$, \ $\psi_1=\psi_{22}+2\varphi\psi_2$ \\
$\mathfrak s_{1.5}$ & $\varphi$ & $\psi  $ & $t  $ & $x$   & $\psi_{222}=0$, \ $\psi_1=\psi_{22}+2\varphi\psi_2$ \\
$\mathfrak s_{1.6}$ & $\varphi$ & $\psi+x$ & $t  $ & $y$   & $\varphi_{12}=0$,\ $\psi_1=2\varphi$ \\
$\mathfrak s_{1.7}$ & $\varphi$ & $\psi  $ & $t  $ & $y$   & $\varphi_{12}=0$, \ $\psi_1=0$ \\
\hline
\end{tabular}
\end{center}
\end{table}

We index each of the listed reduced systems by the number of the corresponding one-dimensional subalgebra of~$\mathfrak g$.
The maximal Lie invariance algebras $\mathfrak a_{1.m}$ of reduced systems~$1.m$, $m=1,\dots,7$, are
\begin{gather*}
\mathfrak a_{1.1}=\big\langle\p_{z_1},\,\p_{z_2},\,2z_1\p_{z_1}+z_2\p_{z_2}-\varphi\p_\varphi-2\psi\p_\psi,\,\p_\psi\big\rangle,
\\[.5ex]
\mathfrak a_{1.2}=\big\langle\p_{z_1},\,\p_{z_2},\,z_1\p_{z_1}-2z_2\p_{z_2}-\varphi\p_\varphi+2\psi\p_\psi,\,\beta\p_\psi\big\rangle,
\\[.5ex]
\mathfrak a_{1.3}=\big\langle\p_{z_1},\,z_1\p_{z_1}-\varphi\p_\varphi,\,\alpha\p_{z_2}-\alpha_{z_2}\psi\p_\psi,\,\beta\p_\psi\big\rangle,
\\[.5ex]
\mathfrak a_{1.4}=\big\langle\p_{z_1},\,2z_1\p_{z_1}+z_2\p_{z_2}-\varphi\p_\varphi-\psi\p_\psi,\,g\p_{z_2}-\tfrac12g_{z_1}\p_\varphi,\,\p_\psi\big\rangle,
\\[.5ex]
\mathfrak a_{1.5}=\big\langle
f\p_{z_1}+\tfrac12z_2f_{z_1}\p_{z_2}-\big(\tfrac12f_{z_1}\varphi+\tfrac14f_{z_1z_1}z_2\big)\p_\varphi,\,g\p_{z_2}-\tfrac12g_{z_1}\p_\varphi,\,\psi\p_\psi,\,\p_\psi\big\rangle,
\\[.5ex]
\mathfrak a_{1.6}=\big\langle\p_{z_1},\,z_1\p_{z_1}-\varphi\p_\varphi,\,
\varphi\p_\varphi+\psi\p_\psi,\,
\gamma\p_\varphi+2\gamma z_1\p_\psi,\,\tfrac12g_{z_1}\p_\varphi+g\p_\psi,\,\beta\p_\psi,\,\alpha\p_{z_2}\big\rangle,
\\[.5ex]
\mathfrak a_{1.7}=\big\langle f\p_{z_1},\,\alpha\p_{z_2},\,g\p_\varphi,\,\gamma\p_\varphi,\,\varphi\p_\varphi,\,\Upsilon\p_\psi\big\rangle,
\end{gather*}
respectively.
Here $f$ and~$g$ run through the set of smooth functions of~$z_1$,
$\alpha$, $\beta$ and~$\gamma$ run through the set of smooth functions of~$z_2$,
and $\Upsilon$ runs through the set of smooth functions of~$(z_2,\psi)$.

Lie symmetries of a reduced system that are not induced by Lie symmetries of the original system
are called additional or hidden; see~\cite{abra2008a} and~\cite{poch2017a} for references.
The algebra of Lie symmetries of reduced system~$1.m$ that are induced by Lie symmetries of the system~\eqref{eq:BLPsystem}
is canonically isomorphic to the quotient algebra ${\rm N}_{\mathfrak g}(\mathfrak s_{1.m})/\mathfrak s_{1.m}$, where
\begin{equation*}
{\rm N}_{\mathfrak g}(\mathfrak s_{1.m})=\{Q\in\mathfrak g\mid[Q,Q']\in\mathfrak s_{1.m}\text{\ for all\ }Q'\in\mathfrak g_{1.m}\}
\end{equation*}
is the normalizer of the subalgebra $\mathfrak s_{1.m}$.
This property allows us to check whether the system~\eqref{eq:BLPsystem}
admits hidden symmetries associated with  reduced system~$1.m$.
By the definition of normalizer, we obtain
\begin{gather*}
{\rm N}_{\mathfrak g}(\mathfrak s_{1.1})=\big\langle D(t)+S(y),\,D(1),\,S(1),\,P(1),\,Z(1)     \big\rangle,\\
{\rm N}_{\mathfrak g}(\mathfrak s_{1.2})=\big\langle D(t)-S(y),\,D(1),\,S(1),\,P(1),\,Z(\beta) \big\rangle,\\
{\rm N}_{\mathfrak g}(\mathfrak s_{1.3})=\big\langle D(t),\,D(1),\,S(\alpha),\,P(1),\,Z(\beta) \big\rangle,\\
{\rm N}_{\mathfrak g}(\mathfrak s_{1.4})=\big\langle D(2t)+S(y),\,D(1),\,S(1),\,P(g),\,Z(1)    \big\rangle,\\
{\rm N}_{\mathfrak g}(\mathfrak s_{1.5})=\big\langle D(f),\,S(1),\,S(y),\,P(g),\,Z(1)          \big\rangle,\\
{\rm N}_{\mathfrak g}(\mathfrak s_{1.6})=\big\langle D(2t)-S(y),\,D(1),\,S(1),\,P(g),\,Z(\beta)\big\rangle,\\
{\rm N}_{\mathfrak g}(\mathfrak s_{1.7})=\big\langle D(t),\,D(1),\,S(\alpha),\,P(g),\,Z(\beta) \big\rangle.
\end{gather*}

If both the algebra~$\mathfrak a_{1.m}$ and the normalizer ${\rm N}_{\mathfrak g}(\mathfrak s_{1.m})$
are finite-dimensional, which is the case only for $m=1$,
the equality $\dim{\rm N}_{\mathfrak g}(\mathfrak s_{1.m})-1=\dim\mathfrak s_{1.m}$ is equivalent to the claim
that all element of~$\mathfrak a_{1.m}$ are induced by elements of~${\rm N}_{\mathfrak g}(\mathfrak s_{1.m})$.
In other words, the system~\eqref{eq:BLPsystem} admits no hidden symmetries
corresponding to the Lie reduction with respect to the subalgebra~$\mathfrak s_{1.1}$ of~$\mathfrak g$.

For all the other values of~$m$, $m=2,\dots,7$,
both the dimensions of~$\mathfrak a_{1.m}$ and ${\rm N}_{\mathfrak g}(\mathfrak s_{1.m})$ are infinite.
Therefore, to check the presence of hidden symmetries for such a value of~$m$,
we need to directly compute the quotient algebra ${\rm N}_{\mathfrak g}(\mathfrak s_{1.m})/\mathfrak s_{1.m}$.
In this way, we show that the algebra ${\rm N}_{\mathfrak g}(\mathfrak s_{1.m})/\mathfrak s_{1.m}$
is canonically isomorphic to~$\mathfrak a_{1.m}$ if $m\in\{2,3,4,5\}$,
i.e., the Lie reductions with respect to the subalgebras~$\mathfrak s_{1.m}$ with $m\in\{2,3,4,5\}$
give no hidden symmetries for the system~\eqref{eq:BLPsystem}.
For $m\in\{6,7\}$, elements of the normalizer ${\rm N}_{\mathfrak g}(\mathfrak s_{1.m})$ induces
only a proper subalgebra of~$\mathfrak a_{1.m}$.
Thus, the vector fields $D(2t)-S(y)$, $D(1)$, $S(1)$, $P(g)$ and~$Z(\beta)$
spanning ${\rm N}_{\mathfrak g}(\mathfrak s_{1.6})$ induce
$2(z_1\p_{z_1}-\varphi\p_\varphi)+(\varphi\p_\varphi+\psi\p_\psi)-z_2\p_{z_2}$,
$\p_{z_1}$, $\p_{z_2}$, $\tfrac12g_{z_1}\p_\varphi+g\p_\psi$ and $\beta\p_\psi$
from the algebra~$\mathfrak a_{1.6}$.
The vector fields~$D(t)$, $D(1)$, $S(\alpha)$, $P(g)$ and $Z(\beta)$
spanning ${\rm N}_{\mathfrak g}(\mathfrak s_{1.7})$ induce
$z_1\p_{z_1}-\frac12\varphi\p_\varphi$, $\p_{z_1}$, $\alpha\p_{z_2}-\frac12\alpha_{z_2}\psi\p_\psi$, $\frac12g\p_\varphi$ and $\beta\p_\psi$
from the algebra~$\mathfrak a_{1.7}$.
All the other vector fields from the algebra~$\mathfrak a_{1.6}$ and~$\mathfrak a_{1.7}$
are hidden symmetries of the system~\eqref{eq:BLPsystem} associated with reduced systems~1.6 and~1.7.
\looseness=-1

As a result, it is needless to carry out further Lie reductions of reduced systems~1.1--1.7 since
all the solutions of the system~\eqref{eq:BLPsystem} that can be constructed by such reductions
satisfy differential constraints considered in Section~\ref{sec:SolutionsByMethodOfDiffConstraints}
or are invariant with respect to two-dimensional subalgebras of~$\mathfrak g$ and, therefore, can be obtained
using direct codimension-two Lie reductions of the system~\eqref{eq:BLPsystem},
which is done in the next Section~\ref{sec:LieReductionsOfCodim2}.

We discuss some properties of nontrivial reduced systems~1.1--1.4 that are not related to their Lie symmetries.

\medskip\par\noindent{\bf 1.1.}\
The first equation of the reduced system can be represented in conserved form in two ways,
$\big(\varphi_1-(\varphi^2-\varphi_2)_2\big)_1=(2\psi_{22})_2$ \ and \
$(\varphi_1)_1=\big((\varphi^2-\varphi_2)_1+2\psi_{22}\big)_2$.
These conserved forms can be used for introducing potentials
although all the potential symmetries of~\eqref{eq:BLPsystem} associated with these potential systems are trivial.

\medskip\par\noindent{\bf 1.2.}\
We assume that $\psi_{111}\ne0$ and $\varphi_{11}\ne0$
(otherwise we have solutions exhaustively studied in Section~\ref{sec:SolutionsByMethodOfDiffConstraints}).
The first reduced equation integrates once with respect to~$z_1$ to $(\varphi^2-\varphi_1)_2+2\psi_{11}=\delta$,
where $\delta$ is an arbitrary function of~$z_2$.

\medskip\par\noindent{\bf 1.3.}\
Similar to the previous case, we can assume $\psi_{111}\ne0$ and $\varphi_{11}\ne0$
and integrate the first reduced equation to $(\varphi^2-\varphi_1)_2+2\psi_{11}=\delta$.
If the ``integration constant''~$\delta$, which is a function of~$z_2$, is nonzero,
then it can be set to 1 by a point symmetry transformation of reduced system~1.3,
which is induced by a transformation~$\mathscr S(Y)$.
Therefore, we can assume that $\delta\in\{0,1\}$ $(\!{}\bmod \{\mathscr S(Y)\})$.

The differential substitution $\varphi=-\tilde\varphi_1/\tilde\varphi$ allows us
to integrate the second equation with respect to~$z_1$, giving $\psi_1=\gamma(z_2)\tilde\varphi^2$.
The ``integration constant''~$\gamma$, which is a function of~$z_2$ as well, is nonzero
in view of nonvanishing~$\psi_1$ and thus can be set to~1
since the function~$\tilde\varphi$ from the differential substitution
is defined up to a nonvanishing multiplier that is a function of~$z_2$.
If $\gamma$ is negative, we need to additionally alternate the signs of $(z_1,\varphi)$,
which is the point symmetry transformation of reduced system~1.3
and is induced by the transformation~$\mathscr I(-1)$.
Thus, $\psi_1=\tilde\varphi^2$.
The integrated first reduced equation takes the form
$(\tilde\varphi_{11}/\tilde\varphi)_2+2(\tilde\varphi^2)_1=\delta$,
which can be written as
$(\tilde\varphi\tilde\varphi_{12}-\tilde\varphi_1\tilde\varphi_2)_1=\delta\tilde\varphi^2-2\tilde\varphi^2(\tilde\varphi^2)_1$
and integrated once more with respect to~$z_1$.
As a result, we obtain
$\tilde\varphi\tilde\varphi_{12}-\tilde\varphi_1\tilde\varphi_2=\delta\psi-\tilde\varphi^4+\alpha$,
where the ``integration constant''~$\alpha$ is a function of~$z_2$.
The last equation can be represented as
\begin{gather}\label{eq:RedSystem1.3IntegratedEq}
\left(\frac{\tilde\varphi_1}{\tilde\varphi}\right)_2=-\tilde\varphi^2+\frac{\delta\psi+\alpha}{\tilde\varphi^2}.
\end{gather}
Note that the equation~\eqref{eq:RedSystem1.3IntegratedEq} can be immediately written in terms of
the initial dependent variables~$\varphi$ and~$\psi$, $\psi_1(\varphi_2-\psi_1)+\delta\psi+\alpha=0$.
This shows that we have twice properly integrated reduced system 1.3.

The case $\delta=\alpha=0$ corresponds to the family of stationary solutions of~\eqref{eq:BLPsystem}
satisfying the differential constraint $u_y=v_x$,
which has been exhaustively studied in Section~\ref{sec:SolutionsByMethodOfDiffConstraints}.

Suppose that $\delta=0$ and $\alpha\ne0$.
Using a point symmetry transformation of reduced system~1.3, which is induced by a transformation~$\mathscr S(Y)$,
we set $\alpha=\ve\in\{-1,1\}$.
The substitution $\tilde\varphi=\ve'{\rm e}^{\theta/2}$ with $\ve'=\sgn\tilde\varphi$
directly reduces the equation~\eqref{eq:RedSystem1.3IntegratedEq} to
the $\sinh$-Gordon equation $\theta_{12}=-4\sinh\theta$ if $\ve=1$ or
the $\cosh$-Gordon equation $\theta_{12}=-4\cosh\theta$ if $\ve=-1$.
See~\cite[Section~7.3.1.1]{poly2012A} and references therein
for exact solutions and other properties of the $\sinh$-Gordon equation.
As a result, we construct the following family of solutions of the system~\eqref{eq:BLPsystem}:
\[
\solution
u=-\frac12\theta_x(x,y),\quad
v=\int_{x_0}^x{\rm e}^{\theta(x',y)}{\rm d}x',
\]
where $\theta=\theta(x,y)$ is an arbitrary solution of
the $\sinh$-Gordon equation $\theta_{xy}=-4\sinh\theta$ or
the $\cosh$-Gordon equation $\theta_{xy}=-4\cosh\theta$.

If $\delta=1$, the parameter function~$\alpha$ can be set to zero
by a point symmetry transformation of reduced system~1.3,
which is induced by a transformation~$\mathscr Z(V^0)$.
We solve the second equation of reduced system~1.3 with respect to~$\varphi$,
$\varphi=-\psi_{11}/(2\psi_1)$,
and substitute the obtained expression into the equation $\psi_1(\varphi_2-\psi_1)+\psi=0$,
which gives the equation
\[
\psi_1\left(\frac{\psi_{11}}{\psi_1}\right)_2+2(\psi_1)^2=2\psi
\]
with respect to~$\psi$.
Multiplying the last equation by $\psi_{11}/(\psi_1)^2$, we can represent it in the conserved form
\[
\bigg(\bigg(\frac{\psi_{11}}{\psi_1}\bigg)^{\!2\,}\bigg)_2+4\bigg(\psi_1+\frac\psi{\psi_1}-z_1\bigg)_1=0.
\]

\noindent{\bf 1.4.}\
As was remarked in~\cite{boit1987a},
reduced system~1.4 is equivalent to the dispersive long-wave equations~\cite{kupe1985b}.
We integrate the first equation and differentiate the second equation with respect to~$z_2$.
Then we re-denote $\psi_2:=\tilde\psi$ and
set the parameter function arising in the course of inte\-gration to zero
by a point symmetry transformation of reduced system~1.4,
which is induced by a point symmetry transformation of the initial system~\eqref{eq:BLPsystem}.
As a result, we derive the system $\varphi_1=(\varphi^2-\varphi_2+2\tilde\psi)_2$, $\tilde\psi_1=(\tilde\psi_2+2\varphi\tilde\psi)_2$,
which has, up to scaling of dependent variables, the form~\cite[Eq.~(1.3)]{kupe1985b}.

A big set of explicit solutions of reduced system~1.4 can be obtained
after posing the constraint $\psi=\varphi$,
under which the first equation from this system becomes a differential constraint for the second one,
and the system reduces to the single Burgers equation $\varphi_1=\varphi_{22}+2\varphi\varphi_2$.

\begin{remark}\label{rem:TrivialityOfRedSystem1.5-1.7}
The integration of reduced systems~1.5--1.7 is trivial.
Moreover, the corresponding invariant solutions of the system~\eqref{eq:BLPsystem}
satisfy the differential constraints $v_{xxx}=0$, $v_{xx}=0$ and $v_x=0$
for $m=5,6,7$, respectively.
Therefore, the Lie reductions of the system~\eqref{eq:BLPsystem}
with respect to the subalgebras~$\mathfrak s_{1.m}$, $m=5,6,7$,
do not result in new exact solutions of this system.
\end{remark}

\section{Essential Lie reductions of codimension two}\label{sec:LieReductionsOfCodim2}

Most of $G$-inequivalent codimension-two Lie reductions of the system~\eqref{eq:BLPsystem}
are not interesting in the sense that each of them can be represented as a two-step reduction,
where the first step is a codimension-one Lie reduction mentioned in Remark~\ref{rem:TrivialityOfRedSystem1.5-1.7}.
The solutions constructed using such codimension-two Lie reductions are trivial
and belong to families of exact solutions obtained in Section~\ref{sec:SolutionsByMethodOfDiffConstraints}.
The subalgebras listed in Lemma~\ref{lem:2DInequivSubalgs} that satisfy the rank condition~\cite{boyk2016a}
and contain the vector fields $S(1)$, $P(1)$ and $P(1)+Z(1)$
are associated with the above trivial codimension-two Lie reductions.
Moreover, some of the other subalgebras listed in Lemma~\ref{lem:2DInequivSubalgs}
do not satisfy the rank condition~\cite{boyk2016a} and are thus not appropriate for Lie reductions at all.
These are the subalgebras containing vector fields from $\{Z(\beta)\}$ and the subalgebras~$\mathfrak s_{2.18}^{g\beta}$.

This is why, modulo the $G$-equivalence, we need to only carry out Lie reductions
with respect to the two-dimensional non-Abelian subalgebras
$\mathfrak s_{2.1}$, $\mathfrak s_{2.2}^\de$, $\mathfrak s_{2.3}^\de$, $\mathfrak s_{2.4}^1$
and the Abelian subalgebras $\mathfrak s^{\de1}_{2.9}$, $\de\in\{0,1\}$.
Below, for each of these subalgebras
we present an ansatz for the related invariant solutions
and the corresponding reduced system,
where $\varphi$ and~$\psi$ are the new unknown functions of the single invariant variable~$\omega$.
\begin{gather*}
\hspace{-\mathindent}\makebox[\mathindent][r]{2.1.\ }
\mathfrak s_{2.1}=\langle D(1)+S(1),D(t)+S(y)\rangle\colon\\
u=\frac{\varphi}{|t-y|^{1/2}}-\frac x{4(t-y)},\ v=\frac\psi{t-y},\quad\mbox{where}\quad \omega=\frac x{|t-y|^{1/2}},
\\[-.5ex]
\omega(\varphi^2-\varphi_\omega)_{\omega\omega}+3(\varphi^2-\varphi_\omega)_\omega+4\psi_{\omega\omega\omega}=\frac12\omega,
\quad
\psi_{\omega\omega}+2\varphi\psi_\omega+\ve\psi=0.
\\[1ex]
\hspace{-\mathindent}\makebox[\mathindent][r]{2.2.\ }
\mathfrak s_{2.2}^\de=\langle D(1)+Z(\de),D(t)-S(y)\rangle,\ \de\in\{0,1\}\colon\\
u=|y|^{1/2}\varphi,\ v=y^{-1}\psi+\delta t,\quad\mbox{where}\quad \omega=|y|^{1/2}x,
\\
\omega(\varphi^2-\varphi_\omega)_{\omega\omega}+3(\varphi^2-\varphi_\omega)_\omega+4\psi_{\omega\omega\omega}=0,
\quad
\psi_{\omega\omega}+2\varphi\psi_\omega=\ve\delta.
\\[1.5ex]
\hspace{-\mathindent}\makebox[\mathindent][r]{2.3.\ }
\mathfrak s_{2.3}^\de=\langle D(1),D(t)+Z(\de)\rangle,\ \de\in\{0,1\}\colon\\
u=x^{-1}\varphi,\ v=\psi+2\delta\ln|x|,\quad\mbox{where}\quad\omega=y,\\
(\varphi^2+\varphi)_\omega+4\varepsilon\delta=0,
\quad
2\varphi=1.
\\[1.5ex]
\hspace{-\mathindent}\makebox[\mathindent][r]{2.4.\ }
\mathfrak s_{2.4}^1=\langle S(1)-P(1),D(2t)+S(y)\rangle\colon\\
u=|t|^{-1/2}\varphi,\ v=|t|^{-1/2}\psi,\quad\mbox{where}\quad\omega=|t|^{-1/2}(x+y),\\[-.5ex]
(\varphi^2-\varphi_\omega)_{\omega\omega}+2\psi_{\omega\omega\omega}+\frac12\ve(\omega\varphi)_{\omega\omega}=0,
\quad
\psi_{\omega\omega}+2\varphi\psi_\omega+\frac12\ve(\omega\psi)_\omega=0.
\\[1ex]
\hspace{-\mathindent}\makebox[\mathindent][r]{2.9.\ }
\mathfrak s_{2.9}^{\de1}=\langle D(1)+Z(\de),S(1)-P(1)\rangle,\ \de\in\{0,1\}\colon\\
u=\varphi,\ v=\psi+\delta t,\quad\mbox{where}\quad\omega=x+y,\\
(\varphi^2-\varphi_\omega)_{\omega\omega}+2\psi_{\omega\omega\omega}=0,
\quad
\psi_{\omega\omega}+2\varphi\psi_\omega=\de.
\end{gather*}

Here $\ve$ denotes the sign of the value whose modulus arise in the corresponding ansatz,
\mbox{$\ve=\sgn(t-y)$}, $\ve=\sgn y$, $\ve=\sgn t$ for reductions~2.1, 2.2 and 2.4, respectively.
We select the ansatzes in such a way that the obtained reduced systems are of the simplest form.

Let us discuss lowering the order and integrating the above reduced systems.
Reductions~2.1, 2.2, 2.3, 2.4 and 2.9 correspond, in the notation of~\cite{paqu1990a}, to
reductions of the system~\eqref{eq:BLPsystemPWform} with respect to the subalgebras
${\rm L}_{2.15}$, ${\rm L}_{2.14}$, ${\rm L}_{2.13}$, ${\rm L}_{2.16}(\epsilon)$ and ${\rm L}_{2.7}(\epsilon)$,
where $\epsilon\in\{-1,1\}$,
of the maximal Lie invariance algebra~${\rm L}_{\rm I}$ of this system, respectively.
We use~\cite[Section~5.1]{paqu1990a} as a source of hints for studying some of the above reduced systems.
However, the equivalence with respect to the corresponding point-symmetry pseudogroup
is more systematically involved in our consideration than in the consideration in~\cite{paqu1990a}.
In particular, the parameter~$\epsilon$ in the subalgebras ${\rm L}_{2.16}(\epsilon)$ and ${\rm L}_{2.7}(\epsilon)$
is in fact inessential since it can be set to a fixed value,~$1$ or~$-1$,
by discrete point symmetry transformations of the system~\eqref{eq:BLPsystemPWform},
which are discussed in Section~\ref{sec:BLPsystemEquivForms}.
Moreover, we systematically compute nontrivial integrating factors
of reduced systems of ordinary differential equations or related single ordinary differential equations
by hand and using the package {\tt Jets}~\cite{BaranMarvan},
which is based on theoretical results of~\cite{marv2009a}.
We advance integrating reduced systems further in comparison with~\cite{paqu1990a}
and completely solve all the reduced systems.

Below, $C_0$, $C_1$ and~$C_2$ are arbitrary constants arising in the course of integration.

\medskip\par\noindent{\bf 2.1.}\
${\rm N}_{\mathfrak g}(\mathfrak s_{2.1})=\mathfrak s_{2.1}$,
which agrees with the fact that the maximal Lie invariance algebra of reduced system~2.1 is zero.
The tuples $(1,0)$, $(\omega^2,0)$ and $\big(\psi,\,\omega\varphi_\omega+\varphi+\frac12\varepsilon\omega\big)$,
constitute a basis of the space of integrating factors of this system
that depend at most on $(\omega,\varphi,\psi,\varphi_\omega,\psi_\omega)$.
Therefore, we have found two more first integrals in comparison with the corresponding case in~\cite{paqu1990a}.
Admitting the first two integrating factors by reduced system~2.1 is equivalent to the fact that
its first equation alone admits two obvious integrating factors, 1 and~$\omega^2$,
which corresponds to the two once integrated equations
\begin{subequations}
\begin{gather}\label{eq:BLPsystemReducedSystem2.1FirstIntegralA}
\omega(\varphi^2-\varphi_\omega)_\omega+2(\varphi^2-\varphi_\omega)+4\psi_{\omega\omega}=\frac14\omega^2+C_1,
\\\label{eq:BLPsystemReducedSystem2.1FirstIntegralB}
\omega^3(\varphi^2-\varphi_\omega)_\omega+4\omega^2\psi_{\omega\omega}-8\omega\psi_\omega+8\psi=\frac18\omega^4-C_2.
\end{gather}
\end{subequations}
Combining these equations for canceling second-order derivatives, we derive
\[
2\omega^2(\varphi^2-\varphi_\omega)+8\omega\psi_\omega-8\psi=\frac18\omega^4+C_1\omega^2+C_2.
\]
Another way for obtaining the last equation is to integrate the equation~\eqref{eq:BLPsystemReducedSystem2.1FirstIntegralA}
once more with the integrating factor~$2\omega$,
which is a display of admitting the integrating factors~1 and~$\omega^2$ by the first equation of reduced system~2.1.
The third integral of this system, which is associated with the integrating factor
$\big(\psi,\,\omega\varphi_\omega+\varphi+\frac12\varepsilon\omega\big)$, leads to the equation
\[
\left(\omega(\varphi^2-\varphi_\omega)_\omega+2(\varphi^2-\varphi_\omega)+4\psi_{\omega\omega}
+\varepsilon\omega\varphi-\frac\varepsilon2\right)\psi
+(\omega\varphi)_\omega\psi_\omega
+\frac\varepsilon2\omega\psi_\omega-2\psi_\omega^{\,\,2}=C_0
\]
or, in view of~\eqref{eq:BLPsystemReducedSystem2.1FirstIntegralA},
$
\left(\varepsilon\omega\varphi+\frac14\omega^2+C_1-\frac12\varepsilon\right)\psi
+(\omega\varphi)_\omega\psi_\omega
+\frac12\varepsilon\omega\psi_\omega-2\psi_\omega^{\,\,2}=C_0.
$
The transformation $\tilde\varphi:=\omega\varphi+\frac14\varepsilon\omega^2$
simplifies the integrating factor and the last equation to a nicer form,
$(\psi,\,\tilde\varphi_\omega)$ and
\begin{gather}\label{eq:BLPsystemReducedSystem2.1ThirdIntegralSimplified}
\left(\varepsilon\tilde\varphi+C_1-\frac\varepsilon2\right)\psi
+\tilde\varphi_\omega\psi_\omega-2\psi_\omega^{\,\,2}=C_0.
\end{gather}

Reduced system~2.1 is equivalent to a single second-order ordinary differential equation for~$\tilde\varphi$.
To show this, we substitute the expression $\varphi=\omega^{-1}\tilde\varphi-\frac14\varepsilon\omega$
into the system consisting of
the equations~\eqref{eq:BLPsystemReducedSystem2.1FirstIntegralA} and~\eqref{eq:BLPsystemReducedSystem2.1FirstIntegralB}
and the second equation of reduced system~2.1 in terms of~$\tilde\varphi$
and solve this system
as a system of linear algebraic equations with respect to $(\psi,\psi_\omega,\psi_{\omega\omega})$.
Then we substitute the derived expressions for $\psi$ and~$\psi_\omega$
into~\eqref{eq:BLPsystemReducedSystem2.1ThirdIntegralSimplified}.
Thus, we obtain
\[
\begin{split}
\psi={}&
-\frac{\omega^2\tilde\varphi_{\omega\omega}}{2(4\tilde\varphi+\varepsilon\omega^2)}
-\omega^2\frac{3\varepsilon\omega^4-4(2C_1\varepsilon+1)\omega^2-8(\varepsilon C_2-2C_1+\varepsilon)}{32(4\tilde\varphi+\varepsilon\omega^2)}
\\[.5ex]&
+\frac14\tilde\varphi^2-\frac{\varepsilon\omega^2-1}4\tilde\varphi
+\frac3{32}\omega^4-\frac{2C_1+3\varepsilon}{16}\omega^2-\frac{C_2}8
\end{split}
\]
and
\begin{gather}\label{eq:BLPsystemReducedSystem2.1EqForPhiTilde}
\left(\tilde\varphi_{\omega\omega}-\frac12P_{\tilde\varphi}\right)^2
=\left(\frac{2\tilde\varphi}\omega+\frac\varepsilon2\omega\right)^2\big(\tilde\varphi_\omega^{\,\,2}-P\big)
\end{gather}
with
$P=P(\tilde\varphi):=
-2\varepsilon\tilde\varphi^3-(2C_1+\varepsilon)\tilde\varphi^2+(\varepsilon C_2-2C_1+\varepsilon)\tilde\varphi
+C_1C_2-\frac12\varepsilon C_2+8C_0$.
By the simple transformation  $\hat\omega=2^{-1/2}\omega$, $\hat\varphi=-\varepsilon\tilde\varphi$, $\hat P=2P$,
the equation~\eqref{eq:BLPsystemReducedSystem2.1EqForPhiTilde} is reduced to
one in the fourth family of second-degree Painlev\'e equations in Chazy's system (II),
\begin{gather}\label{eq:BLPsystemReducedSystem2.1EqForPhiHat}
\left(\hat\varphi_{\hat\omega\hat\omega}-\frac12\hat P_{\hat\varphi}\right)^2
=\left(\frac{2\hat\varphi}{\hat\omega}-\hat\omega\right)^2\big(\hat\varphi_{\hat\omega}^{\,\,2}-\hat P\big),
\end{gather}
see \cite[Eq.~(1.20)]{cosg2006b}, \cite[Eq.~(A.8)]{cosg2006b} and \cite[Eq.~(18.6)]{bure1972a}.
In the notation of~\cite{cosg2006b}, here
\begin{gather}\label{eq:BLPsystemReducedSystem2.1PHat}
\begin{split}&
\hat P(\hat\varphi)=4\hat\varphi^3+\alpha_1\hat\varphi^2+2\beta_1\hat\varphi+\gamma_1\quad\mbox{with}
\\&
\alpha_1=-2(2C_1+\varepsilon),\quad
\beta_1=-(C_2-2\varepsilon C_1+1),\quad
\gamma_1=2C_1C_2-\varepsilon C_2+16C_0.
\end{split}
\end{gather}
According to \cite[Appendix]{cosg2006b}, the equation~\eqref{eq:BLPsystemReducedSystem2.1EqForPhiHat}
is further reduced by a differential substitution
with the simultaneous change $\tau=\hat\omega^2/2=\omega^2/4$ of independent variable
to a fifth Painlev\'e equation.
Therefore, \emph{all the solutions of reduced systems~2.1 are expressed in terms of the fifth Painlev\'e transcendent}.
This can be verified via relating the general initial value problem
for each equation of the form~\eqref{eq:BLPsystemReducedSystem2.1PHat}
with an initial value problem for the corresponding Painlev\'e equation.

As a result, the constructed family of solutions of the system~\eqref{eq:BLPsystem} can be represented in the form
\[
\solution
u=-\frac\varepsilon x\hat\varphi(\tau)-\frac x{2(t-y)},\quad v=\frac{\psi(\tau)}{t-y}
\quad\mbox{with}\quad\tau=\frac{x^2}{4|t-y|},
\]
where $\ve=\sgn(t-y)$,
\begin{gather*}
\hat\varphi=\tau\frac{\varepsilon_1\chi_\tau+\chi}{\chi(\chi-1)}
+\frac{4\lambda+\alpha_1-4\varepsilon_1}{8\chi}-\frac{4\lambda+\alpha_1}8,
\\
\begin{split}
\psi={}&
-\frac\tau4\frac{2\tau\hat\varphi_{\tau\tau}+\hat\varphi_\tau}{\hat\varphi-\tau}
+\frac\tau4\frac{6\tau^2+\alpha_1\tau+\beta_1}{\hat\varphi-\tau}
+\frac{\hat\varphi^2}4+\frac{4\tau-\varepsilon}4\hat\varphi
+\frac32\tau^2-\frac{2C_1+3\varepsilon}4\tau-\frac{C_2}8,
\end{split}
\end{gather*}
$\varepsilon_1=\pm1$,
the polynomial~$\hat P$ and its coefficients~$\alpha_1$, $\beta_1$ and~$\gamma_1$
are defined in~\eqref{eq:BLPsystemReducedSystem2.1PHat},
$\lambda$ is any its root,
and $\chi=\chi(\tau)$ is an arbitrary solution of the fifth Painlev\'e equation
\[
\chi_{\tau\tau}=\left(\frac1{2\chi}+\frac1{\chi-1}\right){\chi_\tau^{\,\,2}}-\frac{\chi_\tau}\tau
+\frac{(\chi-1)^2}{\tau^2}\left(\alpha\chi+\frac\beta\chi\right)
+\gamma\frac\chi\tau+\delta\frac{\chi(\chi+1)}{\chi-1}
\]
with
\[
\alpha=-\frac{(12\lambda-\alpha_1)^2+96\beta_1-4\alpha_1^{\,\,2}}{384},\quad
\beta=-\frac{(4\lambda+\alpha_1-4\varepsilon_1)^2}{128},\quad
\gamma=\lambda-\frac{\varepsilon_1}2,\quad
\delta=-\frac12.
\]
Note that for avoiding the solution of the cubic equation $\hat P(\lambda)=0$,
we can reparameterize the constructed family solution of the system~\eqref{eq:BLPsystem},
assuming the roots of~$\hat P$ as three arbitrary constants
and expressing other constant parameters in terms of these roots.

\medskip
\par\noindent{\bf 2.2.}\
${\rm N}_{\mathfrak g}(\mathfrak s_{2.2}^0)=\langle D(1)+Z(\de),D(t),S(y),Z(y^{-1})\rangle$,
${\rm N}_{\mathfrak g}(\mathfrak s_{2.2}^1)=\langle D(1)+Z(\de),D(t)-S(y),Z(y^{-1})\rangle$.
The maximal Lie invariance algebra of reduced system~$2.2$ is
\smash{$\mathfrak a_{2.2}^0=\big\langle\p_\psi,\omega\p_\omega-\varphi\p_\varphi\big\rangle$} if $\delta=0$ and
\smash{$\mathfrak a_{2.2}^1=\big\langle\p_\psi\big\rangle$} if $\delta=1$.
In both the cases, the basis vector field~$\p_\psi$
is induced by the Lie-symmetry vector field~$Z(y^{-1})$ of the original system~\eqref{eq:BLPsystem},
and the basis vector field~$\omega\p_\omega-\varphi\p_\varphi$ of~$\mathfrak a_{2.2}^0$
is induced by the Lie-symmetry vector field~$S(y)$ of this system.

The first equation of reduced system~2.2 possesses the same two integrating factors
as in the previous case, 1~and~$\omega^2$.
The equations with the associated first integrals are
\begin{subequations}
\begin{gather}\label{eq:BLPsystemReducedSystem2.2FirstIntegralA}
\omega(\varphi^2-\varphi_\omega)_\omega+2(\varphi^2-\varphi_\omega)+4\psi_{\omega\omega}=2C_1,
\\\label{eq:BLPsystemReducedSystem2.2FirstIntegralB}
\omega^3(\varphi^2-\varphi_\omega)_\omega+4\omega^2\psi_{\omega\omega}-8\omega\psi_\omega+8\psi=-2C_2.
\end{gather}
\end{subequations}
We combine these equations, obtaining
\begin{gather*}
\omega^2(\varphi^2-\varphi_\omega)+4\omega\psi_\omega-4\psi=C_1\omega^2+C_2.
\end{gather*}
The same equation can be derived by integrating the equation~\eqref{eq:BLPsystemReducedSystem2.2FirstIntegralA}
once more with the integrating factor~$\omega$,
which is another display of the existence of the two integrating factors for the first equation of reduced system~2.2.

The integration constant~$C_2$ can be set to be equal zero
using a transformation $\omega=\hat\omega$, $\varphi=\hat\varphi$, $\psi=\hat\psi-\frac14C_2$
from the one-parameter group generated
by the induced Lie-symmetry vector field $\p_\psi$.
In other words, up to the $G$-equivalence the constant~$C_2$ is inessential, and thus below it is assumed to be equal to zero.

We solve the system consisting of the equations~\eqref{eq:BLPsystemReducedSystem2.2FirstIntegralA} and~\eqref{eq:BLPsystemReducedSystem2.2FirstIntegralB}
and the second equation of reduced system~2.2 as a system of algebraic equations with respect to~$(\psi,\psi_\omega,\psi_{\omega\omega})$.
The computed expression for~$\psi$ is
\[
\psi=-\frac\omega{8\varphi}\big(\omega\varphi_{\omega\omega}+2\varphi_\omega-2\omega\varphi^3-2\varphi^2+2C_1\omega\varphi+2C_1-4\varepsilon\delta\big).
\]
The consistency of the expressions for~$\psi$, $\psi_\omega$ and~$\psi_{\omega\omega}$ implies the single equation for~$\varphi$,
\begin{gather*}
\omega\varphi\varphi_{\omega\omega\omega}
-(\omega\varphi_\omega-3\varphi)\varphi_{\omega\omega}
-2\varphi_\omega^{\,\,2}-(4\omega\varphi^3+2C_1-4\varepsilon\delta)\varphi_\omega
-4\varphi^4+4C_1\varphi^2=0.
\end{gather*}
This equation admits a single independent integrating factor that is affine in~$\varphi_{\omega\omega}$
with coefficients depending on $(\omega,\varphi,\varphi_\omega)$,
$\varphi^{-3}(\omega\varphi_{\omega\omega}+2\varphi_\omega
+2C_1\omega\varphi+2C_1-4\varepsilon\delta)$.
The corresponding integrated equation can be represented in the form
\begin{gather}\label{eq:BLPsystemReducedSystem2.2EqForPhiIntegrated}
(\tilde\varphi_{\omega\omega}-\alpha_1\tilde\varphi-\beta_1)^2
=4\frac{\tilde\varphi^2}{\omega^2}(\tilde\varphi_\omega^{\,\,2}-\alpha_1\tilde\varphi^2-2\beta_1\tilde\varphi-\gamma_1),
\end{gather}
where $\tilde\varphi:=\omega\varphi$,
$\alpha_1=-4C_1$, $\beta_1=-2(C_1-2\varepsilon\delta)$ and $\gamma_1=-C_0$.
This is one in the third family of second-degree Painlev\'e equations in Chazy's system (II);
see the equations~(1.19) and~(A.5) in~\cite{cosg2006b}, which in fact coincide to each other up to a scaling,
as well as~\cite[Section~23]{bure1972a}.%
\noprint{
\[
\big(\breve\varphi_{\omega\omega}+4C_1\breve\varphi+4(C_1-2\varepsilon\delta)\big)^2
=\frac{\breve\varphi^2}{\omega^2}\big(\breve\varphi_\omega^{\,\,2}+4C_1\breve\varphi^2+8(C_1-2\varepsilon\delta)\breve\varphi+4C_0\big),
\]
where $\breve\varphi:=2\omega\varphi=2\tilde\varphi$.
\[
\big((\omega\varphi)_{\omega\omega}+4C_1\omega\varphi+2C_1-4\varepsilon\delta\big)^2
=4\varphi^2\big((\omega\varphi_\omega+\varphi)^2+4C_1\omega^2\varphi^2+4(C_1-2\varepsilon\delta)\omega\varphi+C_0\big).
\]
}
The equation~\eqref{eq:BLPsystemReducedSystem2.2EqForPhiIntegrated}
is further reduced by a differential substitution, whose form depends on equation parameters
to a third Painlev\'e equation.
Recall that the general form of third Painlev\'e equations~is\looseness=-1
\begin{gather}\label{eq:BLPsystemPainleveIII}
\chi_{\omega\omega}=\frac{\chi_\omega^{\,\,2}}\chi-\frac{\chi_\omega}\omega
+\frac{\alpha\chi^2+\beta}\omega+\gamma\chi^3+\frac{\tilde\delta}\chi,
\end{gather}
where $\alpha$, $\beta$, $\gamma$ and $\tilde\delta$ are constant parameters.
The differential substitutions used in~\cite[Appendix]{cosg2006b}
and the corresponding values of the parameters $\alpha$, $\beta$, $\gamma$ and $\tilde\delta$ are
\begin{subequations}
\begin{gather}\label{eq:BLPsystemReducedSystem2.2DiffSubs1a}
\tilde\varphi=\varepsilon_1\frac\omega2\left(
\frac{\chi_\omega}\chi-\frac{\alpha_1}{2\chi}+\frac\chi2\right),
\\\label{eq:BLPsystemReducedSystem2.2DiffSubs1b}
\alpha=-\frac{2\varepsilon_1\lambda+1}2,\quad
\beta=2\varepsilon_1\beta_1+\alpha_1\frac{2\varepsilon_1\lambda+1}2,\quad
\gamma=\frac14,\quad
\tilde\delta=-\frac{\alpha_1^{\,\,2}}4
\noprint{
\quad (\alpha,\beta,\gamma,\tilde\delta)=\left(
-\frac{2\varepsilon_1\lambda+1}2,\, 2\varepsilon_1\beta_1+\alpha_1\frac{2\varepsilon_1\lambda+1}2,\,
\frac14,\, -\frac{\alpha_1^{\,\,2}}4 \right)
}
\end{gather}
\end{subequations}
if $(\alpha_1,\beta_1)\ne(0,0)$,
where $\lambda$ is any root of the equation $\alpha_1\lambda^2+2\beta_1\lambda+\gamma_1=0$, or
\begin{subequations}
\begin{gather}\label{eq:BLPsystemReducedSystem2.2DiffSubs2a}
\tilde\varphi=\varepsilon_1\frac\omega2\left(\frac{\chi_\omega}\chi+\beta_1\chi\right),
\\\label{eq:BLPsystemReducedSystem2.2DiffSubs2b}
\alpha=\varepsilon_1\gamma_1-\beta_1,\quad
\beta=\varepsilon_1,\quad
\gamma=\beta_1^{\,\,2},\quad
\tilde\delta=0.
\end{gather}
\end{subequations}
if $\alpha_1=0$ and $(\beta_1,\gamma_1)\ne(0,0)$.
The other differential substitution $\tilde\varphi=\omega(\chi_\omega-\beta_1)/(2\chi)$
reduces the equation~\eqref{eq:BLPsystemReducedSystem2.2EqForPhiIntegrated}
with $\alpha_1=0$ to the third Painlev\'e equation
with $\alpha=1$, $\beta=\beta_1+\gamma_1$, $\gamma=0$ and $\tilde\delta=-\beta_1^{\,\,2}$ \cite[Section~23]{bure1972a}.

Thus, \emph{all the solutions of reduced systems~2.2 are expressed in terms of the third Painlev\'e transcendent}.
This can be verified via relating the general initial value problem
for each equation of the form~\eqref{eq:BLPsystemReducedSystem2.2EqForPhiIntegrated}
with an initial value problem for the corresponding Painlev\'e equation.
Note that the counterpart of the case $C_1\ne0$ was not studied in~\cite{paqu1990a}, cf.\ \cite[Section~5.1.3]{paqu1990a}.

As a result, we construct the following family of solutions of the system~\eqref{eq:BLPsystem}:
\[
\solution
u=x^{-1}\tilde\varphi(\omega),\quad v=y^{-1}\psi(\omega)+\delta t,
\]
where $\omega=|y|^{1/2}x$, $\ve=\sgn y$, $\de\in\{0,1\}$,
the function~$\tilde\varphi$ is defined by~\eqref{eq:BLPsystemReducedSystem2.2DiffSubs1a} or~\eqref{eq:BLPsystemReducedSystem2.2DiffSubs2a}
with an arbitrary solution~$\chi$ of the third Painlev\'e equation~\eqref{eq:BLPsystemPainleveIII}
in which the parameter tuple $(\alpha,\beta,\gamma,\tilde\delta)$
is given by~\eqref{eq:BLPsystemReducedSystem2.2DiffSubs1b} or~\eqref{eq:BLPsystemReducedSystem2.2DiffSubs2b}
if $(\alpha_1,\beta_1)\ne(0,0)$ or $\alpha_1=0$, respectively,
$\alpha_1=-4C_1$, $\beta_1=-2(C_1-2\varepsilon\delta)$, $\gamma_1=-C_0$, and
\[
\psi=-\frac{\omega^2}{8\tilde\varphi}(\tilde\varphi_{\omega\omega}-\beta_1)
+\frac14\big(\tilde\varphi^2+\tilde\varphi-C_1\omega^2\big).
\]

In the case $\alpha_1=\beta_1=\gamma_1=0$ or, equivalently, $C_1=C_0=\delta=0$,
the equation~\eqref{eq:BLPsystemReducedSystem2.2EqForPhiIntegrated} is directly integrated, after presenting in the form
$\omega\tilde\varphi_{\omega\omega}=2\varepsilon_1\tilde\varphi\tilde\varphi_\omega$ with $\varepsilon_1=\pm1$,
to a family of Riccati equations, which are easily integrated further.
The corresponding solutions of the system~\eqref{eq:BLPsystem}
are expressed in terms of elementary functions
and can in addition be simplified by scaling transformations from~$G$,
\begin{gather*}
\solution
u=-\frac{\varepsilon_1}{x\tau}-\frac{\varepsilon_1}{2x},\quad
v=\frac{1-\varepsilon_1}{4y\tau}+\frac{1-2\varepsilon_1}{16y},
\\
\solution
u=\frac{\varepsilon_1}x\kappa\tan(\kappa\tau)-\frac{\varepsilon_1}{2x},\quad
v=\frac{1-\varepsilon_1}{4y}\kappa\tan(\kappa\tau)+\frac{1-2\varepsilon_1-4\kappa^2}{16y},
\\
\solution
u=-\frac{\varepsilon_1}x\kappa
\frac{{\rm e}^{2\kappa\tau}-\nu}{{\rm e}^{2\kappa\tau}+\nu}
-\frac{\varepsilon_1}{2x},\quad
v=\frac{1-\varepsilon_1}{2y}\frac{\kappa\nu}{{\rm e}^{2\kappa\tau}+\nu}
+\frac{1-2\varepsilon_1+4\kappa(\kappa+1-\varepsilon_1)}{16y},
\end{gather*}
where $\varepsilon_1=\pm1$, $\tau=\ln|x|+\frac12\ln|y|$,
and $\kappa$ and $\nu$ are arbitrary constants.
All these solutions with $\varepsilon_1=1$ as well as the solutions,
where $\tilde\varphi=\const$, satisfy the differential constraint $v_x=0$
and can thus be considered as trivial.

The performed integrations can be interpreted in terms of first integrals of the entire reduced system~2.2.
The counterparts of the above integrating factors of single differential equations for it
are the tuples $(1,0)$, $(\omega^2,0)$ and $\big(\psi,\,\omega\varphi_\omega+\varphi\big)=\big(\psi,\,\tilde\varphi_\omega\big)$,
which constitute a basis of the space of its integrating factors
that depend at most on $(\omega,\varphi,\psi,\varphi_\omega,\psi_\omega)$.

\medskip\par\noindent{\bf 2.3.}\
This reduced system is overdetermined.
It is inconsistent if $\delta=1$.
If $\delta=0$, then it has the trivial solution
$\varphi=1/2$, $\psi=\psi(\omega)$,
where~$\psi$ is an arbitrary sufficiently smooth function of~$\omega:=y$,
which can be set to zero modulo $\{\mathscr Z(V^0)\}$.
In Section~\ref{sec:SolutionsByMethodOfDiffConstraints}, we have considered
the solutions of the system~\eqref{eq:BLPsystem} with the differential constraint $v_x=0$,
which have been expressed via solutions of the equation~\eqref{eq:BackwardHeatEqWithPot}.
The solutions invariant with respect to the algebra~$\mathfrak s_{2.3}^0$
satisfy this constraint and thus are among solutions constructed under this constraint.
Indeed, the value $\varphi=1/2$ corresponds to $u=1/(2x)$, which is related
via the Hopf--Cole transformation to the solution $\Phi=|x|^{-1/2}$
of the heat equation~\eqref{eq:BackwardHeatEqWithPot} with the potential $H=3x^{-2}/4$.

\medskip\par\noindent{\bf 2.4.}\
${\rm N}_{\mathfrak g}(\mathfrak s_{2.4}^1)=\langle D(2t)+S(y),\,S(1)-P(1),\,P(|t|^{1/2})\rangle$.
The maximal Lie invariance algebra of reduced system~$2.4$ is
$\mathfrak a_{2.4}=\langle\p_\omega-\frac14\varepsilon\p_\varphi\rangle$,
and its basis vector field $\p_\omega-\frac14\varepsilon\p_\varphi$ is induced
by the Lie-symmetry vector field~\smash{$P(|t|^{1/2})$} of the original system~\eqref{eq:BLPsystem}.

It is obvious that the first equation of reduced system~2.4 admits two integrating factors, 1 and~$\omega$,
and we can integrate this equation twice, obtaining
\begin{gather}\label{eq:BLPSystemReduced2.4Eq1Integrated}
2(\varphi^2-\varphi_\omega)+4\psi_\omega+\ve\omega\varphi=C_2\omega+C_1.
\end{gather}
The integration constant~$C_2$ can be set to be equal zero
using a transformation $\omega=\hat\omega-4C_2$, $\varphi=\hat\varphi+\varepsilon C_2$, $\psi=\hat\psi$
from the one-parameter group generated
by the induced Lie-symmetry vector field $\p_\omega-\frac14\varepsilon\p_\varphi$
and re-denoting the arbitrary constant~$C_1$.
In other words, up to the $G$-equivalence the constant~$C_2$ is inessential and thus assumed to be equal zero below.
We solve the equation~\eqref{eq:BLPSystemReduced2.4Eq1Integrated} with respect to~$\psi_\omega$,
substitute the derived expression into the second equation of reduced system~2.4
and solve the resulting equation with respect to~$\psi$,
\begin{gather}\label{eq:BLPSystemReduced2.4ExpressionForPsi}
\psi=-\varepsilon\varphi_{\omega\omega}+2\varepsilon\varphi^3+\frac32\omega\varphi^2
+\left(\frac14\varepsilon\omega^2-\varepsilon C_1+\frac12\right)\varphi-\frac14C_1\omega.
\end{gather}
Substituting the expression~\eqref{eq:BLPSystemReduced2.4ExpressionForPsi} for~$\psi$
into the equation~\eqref{eq:BLPSystemReduced2.4Eq1Integrated}, we obtain
a single third-order ordinary differential equation for~$\varphi$,
\begin{gather}\label{eq:BLPSystemReduced2.4EqForPhi}
\varphi_{\omega\omega\omega}
-\left(6\varphi^2+3\varepsilon\omega\varphi+\frac14\omega^2-C_1\right)\varphi_\omega
-2\varepsilon\varphi^2-\frac34\omega\varphi+\frac12\varepsilon C_1=0,
\end{gather}
which belongs the Chazy VIII class \cite[Eq.~(A11)]{cosg2000a}.
This equation admits the integrating factor $\varphi+\frac12\varepsilon\omega$.
The corresponding integrated equation is
\begin{gather}\label{eq:BLPSystemReduced2.4IntEqForPhi}
\begin{split}
&\left(\varphi+\frac12\varepsilon\omega\right)\varphi_{\omega\omega}
 -\frac12\varphi_\omega^{\,2}-\frac12\varepsilon\varphi_\omega
 -\frac32\varphi^4-2\varepsilon\omega\varphi^3-\left(\frac78\omega^2-\frac12C_1\right)\varphi^2
\\
&\qquad
-\frac18\varepsilon\omega(\omega^2-4C_1)\varphi+\frac18C_1\omega^2=C_0.
\end{split}
\end{gather}
The change of the unknown function
\begin{gather}\label{eq:BLPSystemReduced2.4ChangeOfVars}
\varphi(\omega)=-\frac12\varepsilon\big(\tilde\varphi(\tilde\omega)+\omega\big)
\quad\mbox{with}\quad
\tilde\omega=\frac12\omega   
\end{gather}
reduces the equation~\eqref{eq:BLPSystemReduced2.4IntEqForPhi}
to the fourth Painlev\'e equation
\begin{gather}\label{eq:BLPSystemReduced2.4EqPainleve}
\tilde\varphi\tilde\varphi_{\tilde\omega\tilde\omega}
=\frac12\tilde\varphi_{\tilde\omega}^{\,\,2}
+\frac32\tilde\varphi^4+4\tilde\omega\tilde\varphi^3+2(\tilde\omega^2-C_1)\tilde\varphi^2+\tilde C_0,
\end{gather}
where the arbitrary constant~$C_0$ is replaced by the arbitrary constant~$\tilde C_0$,
which are related by $\tilde C_0=16C_0-2$.
Thus, \emph{the general solution of reduced system~2.4 is expressed in terms of the fourth Painlev\'e transcendent
via the change of variables~\eqref{eq:BLPSystemReduced2.4ChangeOfVars} and the formula~\eqref{eq:BLPSystemReduced2.4ExpressionForPsi}},
cf.\ \cite[Section~5.1.2]{paqu1990a}.
Note that another way of reducing the equation~\eqref{eq:BLPSystemReduced2.4EqForPhi} to the equation~\eqref{eq:BLPSystemReduced2.4EqPainleve},
which is followed in \cite[Section~5.1.2]{paqu1990a},
is to first change variables according~\eqref{eq:BLPSystemReduced2.4ChangeOfVars} and
then integrate the resulting equation once using the correspondingly transformed integrating factor,
see also \cite[p.~74, Eq.~(35.11)]{bure1964b}
and \cite[Section~8]{cosg2000a}.

This result can be interpreted in terms of first integrals of the entire reduced system~2.4.
The counterparts of the above integrating factors of single differential equations for this system
are the tuples $(1,0)$, $(\omega,0)$
and $\big(\psi,\,\varphi_\omega+\frac12\varepsilon\big)$.
It turns out that they constitute a basis of
the space of integrating factors of reduced system~2.4
that depend at most on $\omega$, $\varphi$, $\psi$, $\varphi_\omega$ and~$\psi_\omega$.
Equating each the associated first integrals to an arbitrary constant and solving the obtained equations
jointly with the second equation of reduced system~2.4 with respect to $(\varphi_{\omega\omega},\psi,\psi_\omega,\psi_{\omega\omega})$,
we in particular derive equations that are equivalent to
the equations~\eqref{eq:BLPSystemReduced2.4ExpressionForPsi} and~\eqref{eq:BLPSystemReduced2.4EqPainleve}.

The constructed family of solutions of the system~\eqref{eq:BLPsystem} can be represented in the form
\[
\solution
u=-\frac12\varepsilon|t|^{-1/2}\tilde\varphi(\tilde\omega)-\frac12t^{-1}(x+y),\quad v=|t|^{-1/2}\psi(\tilde\omega),
\]
where $\tilde\omega=\frac12|t|^{-1/2}(x+y)$,
the function~$\tilde\varphi$ is an arbitrary solution of the fourth Painlev\'e equation~\eqref{eq:BLPSystemReduced2.4EqPainleve}, and
\[
\psi=\frac18\tilde\varphi_{\tilde\omega\tilde\omega}
-\frac14\tilde\varphi^3-\frac34\tilde\omega\tilde\varphi^2
-\left(\frac12\tilde\omega^2-\frac12C_1+\frac14\varepsilon\right)\tilde\varphi
+\frac12(C_1-\varepsilon)\tilde\omega.
\]

\medskip\par
\noindent{\bf 2.9.}\
${\rm N}_{\mathfrak g}(\mathfrak s_{2.9}^{01})=\langle D(1),S(1),P(1),Z(1),D(2t)+S(y)\rangle$,
${\rm N}_{\mathfrak g}(\mathfrak s_{2.9}^{11})=\langle D(1),S(1),P(1),Z(1)\rangle$.
The maximal Lie invariance algebra of reduced system~$2.9$ is
\smash{$\mathfrak a_{2.9}^0=\big\langle\p_\omega,\p_\psi,\omega\p_\omega-\varphi\p_\varphi-\psi\p_\psi\big\rangle$} if $\delta=0$ and
\smash{$\mathfrak a_{2.9}^1=\big\langle\p_\omega,\p_\psi\big\rangle$} if $\delta=1$.
In both the cases, the basis vector fields~$\p_\omega$ and~$\p_\psi$
are induced by the Lie-symmetry vector fields~$P(1)$ and~$Z(1)$ of the original system~\eqref{eq:BLPsystem}, respectively,
and the basis vector field~$\omega\p_\omega-\varphi\p_\varphi-\psi\p_\psi$ of~$\mathfrak a_{2.4}^0$
is induced by the Lie-symmetry vector field~$D(2t)+S(y)$ of this system.

Reduced system~2.9 admits three independent integrating factors depending at most on
$(\omega,\varphi,\psi,\varphi_\omega,\psi_\omega)$,
which are $(1,0)$, $(\omega,0)$ and $(\psi,\,\varphi_\omega)$.
This fact underlies all the possible integrations related to this system.
Integrating the first equation of the reduced system twice with respect to the variable~$\omega$
and thus using the first two integrating factors,
we obtain the expression for~$\psi_\omega$ in terms of~$\varphi$, $2\psi_\omega=\varphi_\omega-\varphi^2-C_1\omega-C_0$.
Substituting this expression into the second equation of the system leads to the equation
\begin{gather}\label{eq:BLPsystemReducedSystem2.9SecondIntegral}
\varphi_{\omega\omega}=2\varphi^3+2(C_1\omega+C_0)\varphi+C_1+2\delta.
\end{gather}
We equate the first integral associated with the third integrating factor to an arbitrary constant~$C_2$
and successively exclude the derivatives~$\psi_{\omega\omega}$, $\psi_\omega$ and $\varphi_{\omega\omega}$
from the obtained equation using the above expression for~$\psi_\omega$
and the equation~\eqref{eq:BLPsystemReducedSystem2.9SecondIntegral}.
This leads to the equation
\begin{equation}\label{eq:BLPsystemReducedSystem2.9ThirdIntegral}
\varphi^2_\omega=(\varphi^2+C_1\omega+C_0)^2+4\delta\varphi+4C_1\psi+C_2.
\end{equation}

Further analysis depends on whether or not the constant~$C_1$ is zero.

\medskip\par\noindent$C_1=0.$
Then the equation~\eqref{eq:BLPsystemReducedSystem2.9ThirdIntegral}
is the first-order ordinary differential equation with respect to~$\varphi$,
which takes, after re-denoting $C_0^{\,\,2}+C_2$ by $C_2$, the form
\begin{equation}\label{eq:BLPsystemReducedSystem2.9SecondIntegralC1=0}
\varphi^2_\omega=F(\varphi):=\varphi^4+2C_0\varphi^2+4\delta\varphi+C_2.
\end{equation}
It can also be derived via multiplying the equation~\eqref{eq:BLPsystemReducedSystem2.9SecondIntegral}
by $2\varphi_\omega$ and integrating it once with respect to~$\omega$.
For each fixed value of the parameter tuple $(C_0,\delta,C_2)$,
the general solution of the equation~\eqref{eq:BLPsystemReducedSystem2.9SecondIntegralC1=0}
is expressed via its particular nonconstant solution, where the argument~$\omega$ is shifted
with one more arbitrary constant~$C_3$.
Since the shifts of~$\omega$ are induced by shifts of~$x$ (or $y$) and
we present solutions of the system~\eqref{eq:BLPsystem} up to the $G$-equivalence,
this constant is neglected below.
\looseness=-1

Consider a general value of $(C_0,\delta,C_2)$, for which $F$ has no multiple roots.
The corresponding equation~\eqref{eq:BLPsystemReducedSystem2.9SecondIntegralC1=0} is solvable
in terms of Jacobi elliptic functions or, even more conveniently, in terms of Weierstrass elliptic functions,
see Section~\ref{sec:IntegrationOfODEsRelatedToEllipticFunctions}.
Let $\wp(\omega)=\wp(\omega;g_2,g_3)$ denote the Weierstrass elliptic function,
where
\begin{gather}\label{eq:BLPsystemReducedSystem2.9g2g3}
g_2=C_2+\frac13C_0{}^{\!2},\quad
g_3=\frac13C_0C_2-\frac1{27}C_0{}^{\!3}-\delta^2
\end{gather}
are the invariants of the quartic polynomial $F$ with the indeterminate~$\varphi$.
If at least one root~$\lambda$ of~$F$ is real and is known,
then the equation~\eqref{eq:BLPsystemReducedSystem2.9SecondIntegralC1=0} has
the particular (nonconstant) solution
\[
\varphi=\frac14F_\varphi(\lambda)\left(\wp(\omega)-\frac1{24}F_{\varphi\varphi}(\lambda)\right)^{-1}+\lambda.
\]
We can avoid working with roots of the polynomial $F$
using the Weierstrass representation~\cite[Section~20.6, Example~2]{whit1927A} 
for a particular (nonconstant) solution of~\eqref{eq:BLPsystemReducedSystem2.9SecondIntegralC1=0},
\begin{gather}\label{eq:SolutionOfeqForEllipticFunctions}
\varphi=6\frac
{48\sqrt{F(a)}\,\wp_\omega(\omega)+F_\varphi(a)\big(24\wp(\omega)-F_{\varphi\varphi}(a)\big)+48a F(a)}
{\big(24\wp(\omega)-F_{\varphi\varphi}(a)\big)^2-144F(a)}+a,
\end{gather}
where $a$ is a constant with $F(a)\geqslant0$, which can be selected in order to better express this solution.
It is obvious that such a constant always exists.
Thus, if $C_2\geqslant0$, then we can set~$a=0$, and the expression for~$\varphi$ takes the simpler form
\begin{gather}\label{eq:SimplifiedSolutionOfeqForEllipticFunctions}
\varphi=6\frac{3C_2^{1/2}\wp_\omega(\omega)+\delta(6\wp(\omega)-C_0)}{(6\wp(\omega)-C_0)^2-9C_2}.
\end{gather}
Otherwise, the polynomial~$F$ has a positive root, and hence we can choose a value of~$a$
that exceeds the maximum positive root of~$F$, i.e., it suffices to take $a\geqslant2|C_0|+|C_2|$.
The particular case of~\eqref{eq:BLPsystemReducedSystem2.9SecondIntegralC1=0} with $C_0=\delta=0$ was considered in~\cite{kudr2011b}.

Up to the $G$-equivalence, the corresponding family of solutions of the system~\eqref{eq:BLPsystem} is
\[
\solution
u=\varphi(\omega),\quad v=-\frac12\int_{\omega_0}^\omega\big(\varphi(\omega')\big)^2{\rm d}\omega'+\frac12\varphi(\omega)-\frac12C_0\omega+\delta t,
\]
where $\omega=x+y$, $\delta\in\{0,1\}$,
the function~$\varphi$ satisfies the equation~\eqref{eq:BLPsystemReducedSystem2.9SecondIntegralC1=0} with
an arbitrary value of $(C_0,\delta,C_2)$, for which $F$ has no multiple roots.
Thus, the function~$\varphi$ is, e.g., of the form~\eqref{eq:SolutionOfeqForEllipticFunctions},
and then the corresponding expression for~$v$ takes, up to shifts of~$v$, the form
\[
v:=3\frac{4\wp_\omega(\omega)+4a\sqrt{F(a)}+F_\varphi(a)}{24\wp(\omega)-F_{\varphi\varphi}(a)-12\sqrt{F(a)}}
+\zeta(\omega)-\frac{C_0}3\omega+\delta t.
\]
Here $\zeta(\omega):=\zeta(\omega;g_2,g_3)$ denotes the Weierstrass $\zeta$-function
with the same values~\eqref{eq:BLPsystemReducedSystem2.9g2g3} of the invariants~$g_2$ and~$g_3$
as those in the Weierstrass elliptic function~$\wp(\omega)=\wp(\omega;g_2,g_3)$,
$\zeta(\omega)_\omega=\wp(\omega)$.
If $C_2\geqslant0$ and the simplified expression~\eqref{eq:SimplifiedSolutionOfeqForEllipticFunctions} for~$\varphi$
is used, then the expression for simplifies to
\[
v:=3\frac{\wp_\omega(\omega)+\delta}{6\wp(\omega)-3\sqrt{C_2}-C_0}
+\zeta(\omega)-\frac{C_0}3\omega+\delta t.
\]

If $F(\varphi)$ has multiple roots, then
all solutions of the differential equation~\eqref{eq:BLPsystemReducedSystem2.9SecondIntegralC1=0}
are expressed in terms of elementary functions.
For finding such solutions, we use, as ansatzes,
the solutions of equations of the more general form~\eqref{eq:EllipticODEs},
which are presented in Section~\ref{sec:IntegrationOfODEsRelatedToEllipticFunctions}.
Up to the $G$-equivalence, the corresponding nontrivial solutions of the system~\eqref{eq:BLPsystem} are exhausted by
\begin{gather*}
\solution
u=\frac1{\omega-1}-\frac1{\omega+1}+\frac12,\quad
v=\frac1{\omega-1}+\frac{\omega+t}4,
\\
\solution
u=\frac{4{\rm e}^\omega}{4({\rm e}^\omega+\kappa)^2-1}-\kappa,\quad
v=-\frac{2\kappa-1}{2{\rm e}^\omega+2\kappa-1}+\frac{4\kappa^2-1}4(\omega-2\kappa t),
\\
\solution
u=\frac{\sin\nu}{\sin\omega+\cos\nu}+\frac12\cot\nu,\quad
v=\frac{1-\sin(\omega-\nu)}{2\cos(\omega-\nu)}+\frac{\omega+t\cot\nu}{4\sin^2\nu},
\end{gather*}
where $\omega=x+y$, and $\kappa$ and $\nu$ are arbitrary constants with $\sin\nu\ne0$.
We have not included the solutions
$(u,v)=(-\omega^{-1},0)$, $u=v=\omega^{-1}$,
$(u,v)=(\tan\omega,0)$ and $u=v=-\tan\omega$
in the above list
since each of them satisfies one of the differential constraints $v_x=0$ and $u_y=v_x$,
which are comprehensively studied in Section~\ref{sec:SolutionsByMethodOfDiffConstraints}.

\medskip\par\noindent$C_1\ne0.$
The point transformation
$\tilde\omega=(2C_1)^{1/3}(\omega+C_0/C_1)$, $\tilde\varphi=(2C_1)^{-1/3}\varphi$
reduces the equation~\eqref{eq:BLPsystemReducedSystem2.9SecondIntegral} to the second Painlev\'e equation
\begin{gather}\label{eq:SecondPainleveEq}
\tilde\varphi_{\tilde\omega\tilde\omega}=2\tilde\varphi^3+\tilde\omega\tilde\varphi+\tilde\nu,
\mbox{\quad where\quad}\tilde\nu:=\frac12+\frac\delta{C_1}.
\end{gather}

As a result, we construct the following family of solutions of the system~\eqref{eq:BLPsystem}:
\[
\solution
u=\varphi(\omega),\quad
v=\frac1{4C_1}\Big(\varphi^2_\omega(\omega)-\big(\varphi^2(\omega)+C_1\omega+C_0\big)^2-4\delta\varphi(\omega)-C_2\Big)+\delta t,
\]
where $\omega=x+y$, $\delta\in\{0,1\}$,
$\varphi(\omega)=(2C_1)^{1/3}\tilde\varphi(\tilde\omega)$ with $\tilde\omega=(2C_1)^{1/3}(\omega+C_0/C_1)$
and $\tilde\varphi$ satisfying the second Painlev\'e equation~\eqref{eq:SecondPainleveEq},
and up to the $G$-equivalence $C_0=C_2=0$ and, if $\delta=0$, $C_1=1$.

\medskip

Summing up both the cases $C_1=0$ and $C_1\ne0$,
we conclude that
\emph{all the solutions of reduced systems~2.9 are expressed in terms of
Weierstrass elliptic functions (resp.\ elementary functions in degenerate cases) or the second Painlev\'e transcendent}.

\begin{remark}\label{rem:BLPsystemSolutionsOfPainleveEqsInElementaryFunctions}
There are many known particular Painlev\'e equations admitting at least particular solutions
that are expressed in terms of elementary or classical special functions,
see, e.g.,~\cite{clar2006a} and references therein.
In this way, reductions~2.2, 2.4 and 2.9 gives families of solutions of the system~\eqref{eq:BLPsystem}
in such terms.
\end{remark}

\section{On equivalent forms of the Boiti--Leon--Pempinelli system}\label{sec:BLPsystemEquivForms}

There are other equivalent forms of the Boiti--Leon--Pempinelli system in the literature,
which are nonlocally related to the forms~\eqref{eq:BLPsystem} and~\eqref{eq:BLPsystemOriginal}.
One of such forms is
\begin{equation}\label{eq:BLPsystemYurovForm}
\begin{split}
&u_t=(u^2-u_x)_x+2q_{xx},\\
&q_{ty}=(q_{xy}+2uq_y)_x,
\end{split}
\end{equation}
which coincides with the system~(2) in~\cite{yuro1999a}
up to notation and alternating the sign of~$t$.
(The system~\eqref{eq:BLPsystemOriginal} analogously coincides with the system~(1) in~\cite{yuro1999a}.)
The system~\eqref{eq:BLPsystemOriginal} is related to the systems~\eqref{eq:BLPsystem} and \eqref{eq:BLPsystemYurovForm}
via the differential substitutions $w=v_x$ and $w=q_y$, respectively.
To derive the system~\eqref{eq:BLPsystem} from the system~\eqref{eq:BLPsystemOriginal},
one should substitute $v_x$ for~$w$ and then integrate the second equation with respect to~$x$.
The ``integration constant'', which is an arbitrary sufficiently smooth function of~$(t,y)$,
can be neglected due to the fact that $v$ is defined by the equation $v_x=w$
up to adding an arbitrary sufficiently smooth function of the same arguments.
The inverse transition can be performed via differentiating the second equation of~\eqref{eq:BLPsystem} with respect to~$x$
and then substituting~$w$ for~$v_x$.
Similarly, substituting $q_y$ for~$w$ into the system~\eqref{eq:BLPsystemOriginal}
and integrating the first obtained equation with respect to~$y$ result to the system~\eqref{eq:BLPsystemYurovForm}
since the indeterminateness of~$q$ up to adding an arbitrary sufficiently smooth function of~$(t,x)$
allows for neglecting an arbitrary function of the same arguments that arises in the course of the integration.
Hence for the transition from the system~\eqref{eq:BLPsystemYurovForm} back to the system~\eqref{eq:BLPsystemOriginal}
it suffices to differentiate the first equation of~\eqref{eq:BLPsystemYurovForm} with respect to~$y$
and then substitute~$w$ for~$q_y$.
Given a solution $(u,q)$ the system~\eqref{eq:BLPsystemYurovForm},
the $v$-component of the corresponding solution of the systems~\eqref{eq:BLPsystem} is found, $\bmod\{\mathscr Z(V^0)\}$,~by\looseness=1
\[
v=\int_{x_0}^xq_y(t,x',y)\,{\rm d}x'+\int_{t_0}^t(q_{xy}+2uq_y)\big|_{(t',x_0,y)}{\rm d}t'.
\]

The maximal Lie invariance algebra~$\hat{\mathfrak g}$ of the system~\eqref{eq:BLPsystemOriginal}
can be obtained from the maximal Lie invariance algebra~$\mathfrak g$ of the system~\eqref{eq:BLPsystem}
by prolonging to~$w$ according to the equality $w=v_x$
and then projecting from the space with the coordinates~$(t,x,u,v,w)$
onto the space with the coordinates~$(t,x,u,w)$.
Thus, $\hat{\mathfrak g}=\langle\hat D(f),\hat S(\alpha),\hat P(g)\rangle$, where
\[
\begin{split}
&\hat D(f)=f\p_t+\tfrac12f_tx\p_x-\left(\tfrac12f_tu+\tfrac14f_{tt}x\right)\p_u-\tfrac12f_tw\p_w,\quad
 \hat S(\alpha)=\alpha\p_y-\alpha_y w\p_w,\\
&\hat P(g)=g\p_x-\tfrac12g_t\p_u,\quad
\end{split}
\]
and the parameter functions $f=f(t)$, $g=g(t)$ and $\alpha=\alpha(y)$ run through the sets of smooth functions of their arguments.
Then the maximal Lie invariance algebra~$\breve{\mathfrak g}$ of the system~\eqref{eq:BLPsystemPWform}
can be derived from the algebra~$\hat{\mathfrak g}$ via computing the $(\breve u,\breve w)$-components
in view of the invertible differential substitution $\breve u=-2u$, $\breve w=4w-1-2u_y$.
As a result, we have $\breve{\mathfrak g}=\langle\breve D(f),\breve S(\alpha),\breve P(g)\rangle$
with the same notation of parameters, and
\[
\begin{split}
&\breve D(f)=f\p_t+\tfrac12f_tx\p_x-\tfrac12\left(f_t\breve u-f_{tt}x\right)\p_{\breve u}-\tfrac12f_t(\breve w+1)\p_{\breve w},\quad
 \breve S(\alpha)=\alpha\p_y-\alpha_y(\breve w+1)\p_{\breve w},\\
&\breve P(g)=g\p_x+g_t\p_{\breve u}.
\end{split}
\]
The algebra~$\hat{\mathfrak g}$ also agrees with the maximal Lie invariance algebra
$\check{\mathfrak g}$ of the system~\eqref{eq:BLPsystemYurovForm},
$\check{\mathfrak g}=\langle\check D(f),\check S(\alpha),\check P(g),\check Z^1(h^1),\check Z^2(h^2)\rangle$,
where
\[
\begin{split}
&\check D(f)=f\p_t+\tfrac12f_tx\p_x-\left(\tfrac12f_tu+\tfrac14f_{tt}x\right)\p_u
 -(\tfrac12f_tq+\tfrac1{48}f_{ttt}x^3)\p_q,\quad
 \check S(\alpha)=\alpha\p_y,\\
&\check P(g)=g\p_x-\tfrac12g_t\p_u-\tfrac18g_{tt}\p_q,\quad
 \check Z^1(h^1)=h^1\p_q,\quad
 \check Z^2(h^2)=h^2x\p_q,
\end{split}
\]
and the parameter functions $f=f(t)$, $g=g(t)$, $h^1=h^1(t)$, $h^2=h^2(t)$ and $\alpha=\alpha(y)$
run through the sets of smooth functions of their arguments.
More specifically, the algebra~$\hat{\mathfrak g}$ is obtained from the algebra~$\check{\mathfrak g}$
by prolonging to~$w$ according to the equality $w=q_y$
and then projecting from the space with the coordinates~$(t,x,u,w,q)$
onto the space with the coordinates~$(t,x,u,w)$.
In other words, all Lie symmetries of the systems~\eqref{eq:BLPsystemOriginal} and~\eqref{eq:BLPsystemYurovForm}
are quite expectable when knowing Lie symmetries of the system~\eqref{eq:BLPsystem}.

It is obvious that
the discrete point symmetry transformations~$\mathscr S(-y)$ and~$\mathscr I(-1)$ of the system~\eqref{eq:BLPsystem}
(see Corollary~\ref{cor:BLPSystemDiscretePointSymTrans}) are prolonged to the dependent variables $(w,\breve u,\breve w,q)$
in a pointwise~way,
\begin{align*}
\mathscr I(-1)\colon\ &(t,x,y,u,v,w,\breve u,\breve w,q)\mapsto(t,-x,y,-u,v,-w,-\breve u,-\breve w-2,-q),\\
\mathscr S(-y)\colon\ &(t,x,y,u,v,w,\breve u,\breve w,q)\mapsto(t,x,-y,u,-v,-w, \breve u,-\breve w-2, q).
\end{align*}
Each of the equivalent forms~\eqref{eq:BLPsystem}, \eqref{eq:BLPsystemOriginal}, \eqref{eq:BLPsystemPWform} and~\eqref{eq:BLPsystemYurovForm}
of the Boiti--Leon--Pempinelli system admits the projections of these prolonged transformations
to the space coordinatized by the corresponding independent and dependent variables
as its discrete point symmetry transformations.
It can be expected that a complete list of discrete point symmetry transformations of the system~\eqref{eq:BLPsystemPWform}
that are independent up to combining with each other and with continuous point symmetry transformations of this system
is exhausted by three commuting involutions, which are
the two described projections of the prolonged transformations~$\mathscr S(-y)$ and~$\mathscr I(-1)$
and the transformation~$\mathscr J$ simultaneously alternating the signs of~$t$ and~$\breve u$,
In~\cite{paqu1990a}, only two independent discrete point symmetry transformations of~\eqref{eq:BLPsystemPWform}
were indicated,~$\mathscr J$ and the composition $(t,x,y,\breve u,\breve w)\mapsto(t,-x,-y,-\breve u,\breve w)$ of the above projections.
The counterparts of the transformation~$\mathscr J$
for the system~\eqref{eq:BLPsystem}, \eqref{eq:BLPsystemOriginal} and~\eqref{eq:BLPsystemYurovForm} are not too simple;
in particular, the transformation~$\mathscr J$ cannot be prolonged to~$v$ and~$w$ in a point way
and, moreover, the prolongation to~$v$ is nonlocal and needs the construction of a covering of the system~\eqref{eq:BLPsystem},
\[
\mathscr J\colon\ (t,x,y,u,v,w,\breve u,\breve w,q,\varphi)\mapsto(-t,x,y,-u,v-\varphi_y,w-u_y,-\breve u,\breve w,q-u,-\varphi),
\]
where $\varphi_x=u$ and $\varphi_{ty}=(u^2-u_x)_y+2x_{xx}$.

\section{Laplace and Darboux transformations}\label{sec:LaplaceAndDarbouxTrans}

It is convenient to present
the known Lax pair of the Boiti--Leon--Pempinelli system as well as the corresponding Laplace and Darboux transformations
for the form~\eqref{eq:BLPsystemYurovForm} of this system \cite[Section~2]{yuro1999a}.
Thus, the Lax representation~\eqref{eq:BLPCovering} takes, in terms of~$(u,q)$, the form
\begin{gather}\label{eq:BLPCoveringYurov}
\psi_{xy}+u\psi_y+q_y\psi=0,\quad \psi_t+\psi_{xx}+2q_x\psi=0.
\end{gather}

\subsection{Generating new solutions with Laplace transformations}\label{sec:LaplaceTrans}

In terms of~$(u,q)$, one has the following forward and inverse Laplace transformations for the Boiti--Leon--Pempinelli system:
\begin{gather}\label{eq:BLPLaplaceTrans}
\tilde \psi=\frac{\psi_y}{q_y},\quad
\tilde u=u+\frac{q_{xy}}{q_y},\quad
\tilde q=q+u+\frac{q_{xy}}{q_y},
\\[1ex] \label{eq:BLPLaplaceTransInverse}
\tilde\psi=\psi_x+u\psi,\quad
\tilde u=u-\frac{q_{xy}-u_{xy}}{q_y-u_y},\quad
\tilde q=q-u.
\end{gather}
Explicit expressions for the iterated Laplace transformations were presented in \cite[Section~3]{yuro1999a}.

The Laplace transformations~\eqref{eq:BLPLaplaceTrans} and~\eqref{eq:BLPLaplaceTransInverse}
lead to two ways of generating new solutions from known ones
for the Boiti--Leon--Pempinelli system in the form~\eqref{eq:BLPsystem},
\begin{gather}\label{eq:BLPLaplaceTransUV}
\tilde u=u+\frac{v_{xx}}{v_x},\quad
\tilde v=v+\frac{v_{xy}}{v_x}+\int_{x_0}^xu_y\big|_{(t,x',y)}{\rm d}x'+\int_{t_0}^t\big((u^2-u_x)_y+2v_{xx}\big)\big|_{(t',x_0,y)}{\rm d}t',
\\[1.5ex] \label{eq:BLPLaplaceTransInverseUV}
\tilde u=u-\frac{u_{xy}-v_{xx}}{u_y-v_x},\quad
\tilde v=v-\int_{x_0}^xu_y\big|_{(t,x',y)}{\rm d}x'-\int_{t_0}^t\big((u^2-u_x)_y+2v_{xx}\big)\big|_{(t',x_0,y)}{\rm d}t'.
\end{gather}

Acting on the exact solutions found in Sections~\ref{sec:SolutionsByMethodOfDiffConstraints}, \ref{sec:LieReductionsOfCodim1} and~\ref{sec:LieReductionsOfCodim2}
by the Laplace transformations, it is easy to construct new solutions of the system~\eqref{eq:BLPsystem}.
Indeed, let us apply the forward Laplace transformation~\eqref{eq:BLPLaplaceTransUV}
to the solutions~\eqref{eq:BLPsystemVxxxSol1}--\eqref{eq:BLPsystemVxxxSol5},
which are rational at least with respect to~$x$.
Thus, from~\eqref{eq:BLPsystemVxxxSol1} with $\delta=1$ and from~\eqref{eq:BLPsystemVxxxSol4},
we respectively obtain the families of solutions
\begin{gather*}
\begin{split}
\solution&u=-\frac12\frac{x+\alpha}{t+\beta}+\frac2{2(x+\alpha)+\gamma(t+\beta)},
\\[1ex]
&v=\frac{\beta_y+4}4\bigg(\frac{x+\alpha}{t+\beta}\bigg)^2\!+\frac{2\gamma-\alpha_y}2\frac{x+\alpha}{t+\beta}-\frac32\frac{4t+\beta_y}{t+\beta}
-\frac{2\gamma_y(x+\alpha)-2\alpha_y\gamma-\beta_y\gamma^2}{2\gamma(x+\alpha)+\gamma^2(t+\gamma)}\ \mbox{if}\ \gamma\ne0,
\\[1ex]
&v=\frac{\beta_y+4}4\bigg(\frac{x+\alpha}{t+\beta}\bigg)^2\!+\frac{2\gamma-\alpha_y}2\frac{x+\alpha}{t+\beta}-\frac32\frac{\beta_y+4}{t+\beta}
+\frac{\alpha_y}{x+\alpha}\ \ \mbox{if}\ \ \gamma=0,
\end{split}
\\[1.5ex]
\begin{split}
\solution&u=\alpha+\frac{2\alpha}{2\alpha(x+2\alpha t)+\gamma-1},
\\[.5ex]
&v=\alpha(x+2\alpha t)^2+\gamma(x+2\alpha t)+(\alpha_y-1)(x+t)+3t-
\frac{2\alpha_y(\alpha x+\gamma-1)-\alpha\gamma_y}{2\alpha^2(x+2\alpha t)+\alpha(\gamma-1)}
\end{split}
\end{gather*}
of the system~\eqref{eq:BLPsystem},
where $\alpha$, $\beta$ and $\gamma$ are sufficiently smooth functions of~$y$, and $\alpha\ne0$ in the second family.
Solutions from these families satisfy
none of the differential constraints considered in Section~\ref{sec:SolutionsByMethodOfDiffConstraints}
if, e.g., $\alpha_y\ne0$ in the first family and $\beta_y\ne0$, $\beta_y\ne-4$ and $\gamma_y\ne0$ under $\alpha\ne0$ in the second one.
For $\alpha=0$, the counterpart of the second family is a trivial subfamily of~\eqref{eq:BLPsystemVxxxSol4}.

In fact, not all transformed solutions are new.
In particular, the forward Laplace transformation maps
the solution~\eqref{eq:BLPsystemVxxxSol1}$_{\delta=0}$ to~\eqref{eq:BLPsystemVxxxSol1}$_{\delta=1}$
with $\tilde\gamma=2(2-\alpha_y)/\beta_y$ $(\!{}\bmod \{\mathscr S(Y)\})$,
the solution~\eqref{eq:BLPsystemVxxxSol3} to~\eqref{eq:BLPsystemVxxxSol1}$_{\delta=1}$
with $\tilde\gamma=-2(\alpha_y-4\gamma)/(\beta_y-4)$ $(\!{}\bmod \{\mathscr S(Y)\})$,
the solution~\eqref{eq:BLPsystemVxxxSol5} to~\eqref{eq:BLPsystemVxxxSol4} with $\tilde\gamma=\alpha(\alpha_y+2\beta)$,
and the solution family~\eqref{eq:BLPsystemVxxxSol2} to itself with $\theta$ shifted by~$1$.

Analyzing the above generation of solutions by the forward Laplace transformation,
it becomes obvious that~\eqref{eq:BLPsystemVxxxSol1},~\eqref{eq:BLPsystemVxxxSol2} and~\eqref{eq:BLPsystemVxxxSol4}
do not leave, under the action of the inverse Laplace transformation,
the family of solutions satisfying the differential constraint $v_{xxx}=0$.
Indeed, the solution~\eqref{eq:BLPsystemVxxxSol1}$_{\delta=0}$ with $\beta_y\ne4$ and $\beta_y=4$
is mapped to~\eqref{eq:BLPsystemVxxxSol3} with $\tilde\gamma=-(\alpha_y+2\gamma)/(\beta_y-4)$
and to~\eqref{eq:BLPsystemVxxxSol1}$_{\delta=0}$ $(\!{}\bmod \{\mathscr S(Y)\})$, respectively,
whereas~\eqref{eq:BLPsystemVxxxSol3}$_{\delta=0}$ is mapped to~\eqref{eq:BLPsystemVxxxSol1} with $\tilde\gamma=-(2+\alpha_y)/\beta_y$.
The parameter function~$\theta$ in the solution~\eqref{eq:BLPsystemVxxxSol2} is shifted by~$-1$
and the solution~\eqref{eq:BLPsystemVxxxSol4} is transformed to~\eqref{eq:BLPsystemVxxxSol4} with $\tilde\beta=-(\alpha_y+1-\gamma)/2\alpha$.
The use of the inverse Laplace transformation is justified only for the solutions~\eqref{eq:BLPsystemVxxxSol3} and~\eqref{eq:BLPsystemVxxxSol5}
from the set~\eqref{eq:BLPsystemVxxxSol1}--\eqref{eq:BLPsystemVxxxSol5}
since only with them we go beyond the set of solutions satisfying the differential constraint $v_{xxx}=0$,
\begin{gather*}
\begin{split}
\solution&u=-\frac12\frac{x+\alpha}{t+\beta}-\frac1{\omega}-
\frac{(\beta_y-4)\omega^3-4\omega_y(t+\beta)^2}{\omega\big((\beta_y-4)\omega^3-(\alpha_y+\beta_y\gamma)(t+\beta)\omega^2+2\omega_y(t+\beta)^2\big)},
\\[1ex]
&v=-\frac{\beta_y-4}4\bigg(\frac{x+\alpha}{t+\beta}\bigg)^2+\frac{\alpha_y+4\gamma}{2}\frac{x+\alpha}{t+\beta}-
\frac32\frac{\beta_y-4}{t+\beta}+\frac{\alpha_y+\beta_y\gamma}\omega-\frac{\gamma_y}\gamma\frac{x+\alpha}\omega
\quad\mbox{if}\quad\gamma\ne0,
\\[1ex]
&v=-\frac{\beta_y-4}{4}\bigg(\frac{x+\alpha}{t+\beta}\bigg)^2+\frac{\alpha_y}2\frac{x+\alpha}{t+\beta}-\frac32\frac{\beta_y-4}{t+\beta}+\frac{\alpha_y}{x+\alpha}
\quad\mbox{if}\quad\gamma=0,
\\[1ex]
&\mbox{here}\ \ \omega:=x+\alpha+\gamma(t+\beta)\ \ \mbox{and}\ \ (\gamma,\alpha_y,\beta_y)\ne(0,0,4),\\
\end{split}
\\[1ex]
\solution u=\frac1x-\frac{2x}{x^2-2(t+4y)}-\frac12\frac x{t+4y}, \quad v=\frac12\frac x{t+4y}\quad\mbox{if}\quad(\gamma,\alpha_y,\beta_y)=(0,0,4),
\\[1ex]
\solution
u=\alpha+\frac{4\omega^3-\alpha_y\omega^2+2\alpha_yt+\beta_y}{\omega(\alpha_y\omega^2-2\omega^3+2\alpha_yt+\beta_y)},\quad
v=\omega^2-6t+\frac{2\alpha_yt+\beta_y}\omega\quad
\text{with}\ \omega:=x+2\alpha t+\beta.
\end{gather*}
Here $\alpha$, $\beta$ and $\gamma$ are again sufficiently smooth functions of~$y$
satisfying the indicated constraints.

The forward and inverse Laplace transformations~\eqref{eq:BLPLaplaceTransUV} and~\eqref{eq:BLPLaplaceTransInverseUV}
can be used iteratively for finding new solutions for the system~\eqref{eq:BLPsystem}.
In particular, the double forward Laplace transformation maps
the simple solution~\eqref{eq:DiffConsrU=V}, that satisfies the constraint $u=v$,
contains the arbitrary parameter function $\alpha$ of~$y$ and is rational with respect to $(t,x)$,
to a much more complicated solution of similar structure,
\begin{gather*}
\begin{split}
\solution u={}&\frac{2z}{z^2+\tau}-\frac z{\tau}-\frac1z,
\quad
v=\frac{2z-\alpha_y}{z^2+\tau}-\frac{\alpha_y}2\frac{z^2}{\tau^2}-2\frac z\tau+\frac32\frac{\alpha_y}{\tau},
\end{split}
\\[1ex]
\begin{split}
\solution u={}&-\frac{z}{\tau}-\frac{2z}{z^2+\tau}+\frac{2(z^2+\tau)(2z\alpha_y+3\tau)}{(z^4+2\tau z^2-\tau^2)\alpha_y+2z\tau(z^2+3\tau)},
\\
v={}&\frac{\alpha_{yy}-6}{\alpha_y-6z}+\frac{5\alpha_y-3z}{\tau}-\alpha_y\frac{z^2}{\tau^2}\\
&{}+\frac{2z^3(3z^2+7\tau)\alpha_{yy}+(2z^2-2\tau)\alpha_y^3-2z(7z^2-10\tau)\alpha_y^2+6z^2(z^2-11\tau)\alpha_y+24z^3\tau}
{\big((z^4+2\tau z^2-\tau^2)\alpha_y+2z\tau(z^2+3\tau)\big)(6z-\alpha_y)},
\end{split}
\\
\mbox{where}\ \ z:=x+y\ \ \mbox{and}\ \ \tau:=2t-\alpha.
\end{gather*}

In fact, it is more convenient to construct new solution families for the equivalent system~\eqref{eq:BLPsystemYurovForm}
using~\eqref{eq:BLPLaplaceTrans} or~\eqref{eq:BLPLaplaceTransInverse} instead of~\eqref{eq:BLPLaplaceTransUV} or~\eqref{eq:BLPLaplaceTransInverseUV},
thus avoiding the integration operation.

Consider the solution~\eqref{eq:BLPSolutionFamilyViaHeatEq2}, which satisfies the differential constraint $u_y=v_x$ and
is parameterized by an arbitrary nonzero solution $\Phi=\Phi(t,x,y)$ of the heat equation $\Phi_t-\Phi_{xx}+H\Phi=0$ with potential $H=H(t,x)$.
In view of the relation $v_x=q_y$, the differential constraint $u_y=v_x$ in terms of $(u,q)$ takes the form $u_y=q_y$.
The solution~\eqref{eq:BLPSolutionFamilyViaHeatEq2} of~\eqref{eq:BLPsystem} corresponds to
the solution $u=\Phi_x/\Phi$ and $q=\Phi_x/\Phi+L$ of~\eqref{eq:BLPsystemYurovForm},
where $L=L(t,x)$ is an antiderivative of $-H/2$ with respect to~$x$, $L_x=-H/2$.
The inverse Laplace transformation~\eqref{eq:BLPLaplaceTransInverse} is not defined for such solutions,
whereas the multiple application of the forward Laplace transformation~\eqref{eq:BLPLaplaceTrans} gives
the following recursive formulas for a chain of solutions:
\begin{gather*}
\solution u=\frac{\Phi^n_x}{\Phi^n}-\frac{\Phi^{n-1}_x}{\Phi^{n-1}},\quad q=\frac{\Phi^n_x}{\Phi^n}+L, \quad n\geqslant0,
\\
\hspace*{-\mathindent}\mbox{where}
\\[-1ex]
\Phi^{-1}:=1,\quad
\Phi^0:=\Phi,\quad
\Phi^n:=\frac{\Phi^{n-1}\Phi^{n-1}_{xy}-\Phi^{n-1}_x\Phi^{n-1}_y}{\Phi^{n-2}},\quad n\geqslant1,
\end{gather*}
$\Phi=\Phi(t,x,y)$ is an arbitrary nonzero solution of the equation $\Phi_t-\Phi_{xx}-2L_x\Phi=0$,
and $L$ is an arbitrary sufficiently smooth function of~$(t,x)$.
The $n$-fold forward Laplace transformation is defined for a function~$\Phi$ if and only if
all the recursively computed functions $\Phi^1$, \dots, $\Phi^n$ do not vanish.
In other words, for each~$i\in\{0,\dots,n\}$, the variables~$x$ and~$y$ are not separated in~$\Phi^i$,
i.e., $(\ln|\Phi^i|)_{xy}\ne0$.

The inverse Laplace transformation~\eqref{eq:BLPLaplaceTransInverse} can also be applied
for finding new solutions of the Boiti--Leon--Pempinelli system
that are parameterized by the general solutions of the heat equation with potential.
Indeed, the family of solutions~\eqref{eq:InhomPotentialBurgersEq},
which satisfy the differential constraint $v_x=0$,
takes the form $u=-\Phi_x/\Phi$, $q=L$ in terms of~$(u,q)$,
and the corresponding differential constraint is $q_y=0$.
Here the function $\Phi=\Phi(t,x,y)$ is an arbitrary nonzero solution of the equation $\Phi_t+\Phi_{xx}-H\Phi=0$
with potential $H=-2L_x$, and $L=L(t,x)$ an arbitrary sufficiently smooth function.
Then the $n$-fold inverse Laplace transformation for the above family is given~by
\begin{gather*}
\solution u=-\frac{\Phi^n_x}{\Phi^n},\quad q=S^n+L, \quad n\geqslant0,
\\
\hspace*{-\mathindent}\mbox{where}
\\[-1ex]
S^0:=0,\quad \Phi^0:=\Phi,\quad
S^n:=\sum_{i=0}^{n-1}\frac{\Phi^i_x}{\Phi^i},\quad
\Phi^n:=\Phi^{n-1}S^n_y,\quad n\geqslant1,
\end{gather*}
$\Phi=\Phi(t,x,y)$ is an arbitrary nonzero solution of the equation $\Phi_t+\Phi_{xx}+2L_x\Phi=0$,
$L$ is an arbitrary sufficiently smooth function of~$(t,x)$, and
all the recursively computed auxiliary functions~$\Phi^1$, \dots, $\Phi^n$ do not vanish.
Note that the forward Laplace transformation~\eqref{eq:BLPLaplaceTrans}
is not defined for the solutions with $q_y=0$.

\subsection{Generating new solutions with Darboux transformations}\label{sec:DarbouxTrans}

For each solution $(u,q)$ of the system~\eqref{eq:BLPsystemYurovForm} and
for each fixed nonzero solution~$\varphi$ of the corresponding system of the form~\eqref{eq:BLPCoveringYurov},
there are two types of Darboux transformations \cite[Section~3]{yuro1999a},
\begin{gather}\label{eq:BLPDT1}
{\rm DT}_1[\varphi]\colon\quad
\tilde\psi=\psi-\frac\varphi{\varphi_y}\psi_y,\quad
\tilde u=u+\frac{\varphi_{xy}}{\varphi_y}-\frac{\varphi_x}\varphi,\quad
\tilde q=q-\frac\varphi{\varphi_y}q_y,
\\[.5ex]\label{eq:BLPDT2}
{\rm DT}_2[\varphi]\colon\quad
\tilde\psi=\psi_x-\frac{\varphi_x}\varphi\psi,\quad
\tilde u=u-\frac{(u+\varphi_x/\varphi)_x}{u+\varphi_x/\varphi},\quad
\tilde q=q+\frac{\varphi_x}\varphi.
\end{gather}
It is easy to check by direct computation that such Darboux transformations ``commute''.
More specifically, given linearly independent solutions~$\varphi^1$ and~$\varphi^2$
of the system~\eqref{eq:BLPCoveringYurov} for a solution $(u,q)$ of the system~\eqref{eq:BLPsystemYurovForm},
one has
\begin{gather}\label{eq:BLPDTsCommute}
\begin{split}
&{\rm DT}_1\big[{\rm DT}_1[\varphi^1]\varphi^2\big]\circ{\rm DT}_1[\varphi^1]=
 {\rm DT}_1\big[{\rm DT}_1[\varphi^2]\varphi^1\big]\circ{\rm DT}_1[\varphi^2],
\\[.5ex]
&{\rm DT}_2\big[{\rm DT}_1[\varphi^1]\varphi^2\big]\circ{\rm DT}_2[\varphi^1]=
 {\rm DT}_2\big[{\rm DT}_1[\varphi^2]\varphi^1\big]\circ{\rm DT}_2[\varphi^2],
\\[.5ex]
&{\rm DT}_2\big[{\rm DT}_1[\varphi^1]\varphi^2\big]\circ{\rm DT}_1[\varphi^1]=
 {\rm DT}_1\big[{\rm DT}_1[\varphi^2]\varphi^1\big]\circ{\rm DT}_2[\varphi^2]
\end{split}
\end{gather}
whenever the involved expressions are defined.

Explicit expressions for the iterated Darboux transformations were presented in \cite[Section~3]{yuro1999a} as well.
Thus, given linearly independent solutions~$\varphi^1$, \dots, $\varphi^n$
of the system~\eqref{eq:BLPCoveringYurov} for a solution $(u,q)$ of the system~\eqref{eq:BLPsystemYurovForm},
the associated $n$-times iterated Darboux transformations of the first and the second kinds
respectively take the form
\begin{gather}
\label{eq:BLPnTimesIteratedDT1}
\begin{split}
{\rm DT}_1[\varphi^1,\dots,\varphi^n]\colon\quad
&\tilde\psi=(-1)^n\frac{\bar{\rm W}(\varphi^1,\dots,\varphi^n,\psi)}{\bar{\rm W}(\varphi^1_y,\dots,\varphi^n_y)},\quad
 \tilde u=u-\p_x\ln\left|\frac{\bar{\rm W}(\varphi^1,\dots,\varphi^n)}{\bar{\rm W}(\varphi^1_y,\dots,\varphi^n_y)}\right|,\\
&\tilde q=(-1)^n\frac{\bar{\rm W}(\varphi^1,\dots,\varphi^n,q)}{\bar{\rm W}(\varphi^1_y,\dots,\varphi^n_y)},
\end{split}
\\[1.5ex]
\label{eq:BLPnTimesIteratedDT2}
\begin{split}
{\rm DT}_2[\varphi^1,\dots,\varphi^n]\colon\quad
&\tilde\psi=\frac{{\rm W}(\varphi^1,\dots,\varphi^n,\psi)}{{\rm W}(\varphi^1,\dots,\varphi^n)},\quad
 \tilde u=\frac{q_yu-nq_{xy}+A^{n1}_{xy}-A^{n1}A^{n1}_y+A^{n2}_y}{q_y-A^{n1}_y},\!\!\!\\[.5ex]
&\tilde q=q-A^{n1},\\
&A^{n1}:=-\p_x\ln\left|{\rm W}(\varphi^1,\dots,\varphi^n)\right|,\quad
 A^{n2}:=\frac{{\rm W}(\varphi^1,\dots,\varphi^n,\psi)^{n-1}_{n+1}}{{\rm W}(\varphi^1,\dots,\varphi^n)},
\end{split}
\end{gather}
where $\bar{\rm W}$ and $\rm W$ denote the Wronskians of the indicated function tuples with respect to~$y$ and~$x$, respectively,
$\bar{\rm W}(f^1,\dots,f^k):=\det(\p_y^{i-1}f^j)_{i,j=1,\dots,k}$ and
$    {\rm W}(f^1,\dots,f^k):=\det(\p_x^{i-1}f^j)_{i,j=1,\dots,k}$,
${\rm W}(\varphi^1,\dots,\varphi^n,\psi)^{n-1}_{n+1}$ is the minor of the Wronskian ${\rm W}(\varphi^1,\dots,\varphi^n,\psi)$
obtained by deleting the $(n-1)$th row and the $(n+1)$th column,
$\bar{\rm W}(\varphi^1_y,\dots,\varphi^n_y)\ne0$ for the first transformation and
${\rm W}(\varphi^1,\dots,\varphi^n)\ne0$ for the second transformation.
(We have checked and enhanced the derivation of the iterated Darboux transformations that was carried out in \cite[Section~3]{yuro1999a}.
We have also modified the representation of these transformations.)

Given a solution $(u,q)$ of the system~\eqref{eq:BLPsystemYurovForm} and
linearly independent solutions~$\varphi^1$,~\dots, $\varphi^n$, $\chi^1$,~\dots, $\chi^m$
of the system~\eqref{eq:BLPCoveringYurov} with these~$u$ and~$q$,
it follows from the ``commutation'' property~\eqref{eq:BLPDTsCommute} that
the various compositions, in any order,
of the Darboux transformations of the first kind
that are constructed with the functions~$\varphi^1$, \dots, $\varphi^n$ and their relevant images and
of the Darboux transformations of the second kind
that are constructed with the functions~$\chi^1$,~\dots,~$\chi^m$ and their relevant images
lead to the same result.
There are two convenient representation for the joint transformations,
\begin{gather*}
{\rm DT}_2[\tilde\chi^1,\dots,\tilde\chi^m]\circ{\rm DT}_1[\varphi^1,\dots,\varphi^n],
\quad\mbox{where}\quad
\tilde\chi^j={\rm DT}_1[\varphi^1,\dots,\varphi^n]\chi^j,\quad j=1,\dots,m,
\\
{\rm DT}_1[\tilde\varphi^1,\dots,\tilde\varphi^n]\circ{\rm DT}_2[\chi^1,\dots,\chi^m],
\quad\mbox{where}\quad
\tilde\varphi^i={\rm DT}_2[\chi^1,\dots,\chi^m]\varphi^i,\quad i=1,\dots,n,
\end{gather*}
and explicit expressions can be derived via composing~\eqref{eq:BLPnTimesIteratedDT1} and~\eqref{eq:BLPnTimesIteratedDT2}.

The application of the Darboux transformations~\eqref{eq:BLPDT1} and~\eqref{eq:BLPDT2}
is not so easy as that of the Laplace transformations~\eqref{eq:BLPLaplaceTrans} or~\eqref{eq:BLPLaplaceTransInverse}
since the former relies on finding solutions~$\psi$
of the overdetermined linear system of partial differential equations~\eqref{eq:BLPCoveringYurov}.

An obvious simple choice for a seed solution of the Boiti--Leon--Pempinelli system is $u=q=0$
since the general solution of the corresponding system $\psi_{xy}=0$, $\psi_t+\psi_{xx}=0$
is presented in the form $\psi=\chi(t,x)+\alpha(y)$,
where $\chi=\chi(t,x)$ is the general solution of the backward linear heat equation, $\chi_t+\chi_{xx}=0$,
and $\alpha=\alpha(y)$ is an arbitrary sufficiently smooth function of~$y$.
Nevertheless, using the seed solution $u=q=0$ in
is not efficient since even its multifold version only leads to solutions that satisfy the constraint $q=0$
or, in terms of $(u,v)$, $v_x=0$.
Similarly, the multifold version of the Darboux transformation~\eqref{eq:BLPDT2} maps
the seed solution $u=q=0$ only to solutions that satisfy the differential constraint $u_y=q_y$
or, in terms of $(u,v)$, $u_y=v_x$.
In particular, the solutions of the form~(31) in \cite{yuro1999a} are exactly of this kind,
where $H=0$ in the corresponding linear heat equation~\eqref{eq:ForwardHeatEqWithPot}.
To construct solutions out of the sets of solutions satisfying the constraints $q=0$ or $u_y=q_y$,
one should combine the Darboux transformations~\eqref{eq:BLPDT1} and~\eqref{eq:BLPDT2}
or, more generally, their multifold versions.
The above claims on the seed solution $u=q=0$ follows from the presented below consideration
of arbitrary solutions in these sets as seed solutions for the Darboux transformations
since the zero solution belongs to the intersection of these sets.\looseness=-1

Let us discuss the usage of the solutions satisfying the differential constraint $q_y=0$ or $u_y=q_y$
as seed solutions for the Darboux transformations.

As shown above, every solution of the system~\eqref{eq:BLPsystemYurovForm} with the differential constraint $q_y=0$
takes the form $u=-\Phi_x/\Phi$, $q=L$, where $L=L(t,x)$ is an arbitrary sufficiently smooth function of~$(t,x)$,
and $\Phi=\Phi(t,x,y)$ is an arbitrary nonzero solution of the equation $\Phi_t+\Phi_{xx}+2L_x\Phi=0$.
After substituting this pair $(u,q)$ into the system~\eqref{eq:BLPCoveringYurov},
we integrate the first equation and use the second equation for deriving the equations
that should be satisfied by the functions $\zeta=\zeta(t,y)$ and $\theta=\theta(t,x)$
arising in the course of the integration, $\zeta_t=0$ and $\theta_t+\theta_{xx}+2L_x\theta=0$.
As a result, we derive the following representation for every solution of the system~\eqref{eq:BLPCoveringYurov}
with these values of~$u$ and~$q$:
\[
\psi=\frac1\Phi\smash{\int_{y_0}^y\zeta(y')\Phi(t,x,y')\,{\rm d}y'}+\theta(t,x),
\]
where $\zeta=\zeta(y)$ is an arbitrary sufficiently smooth function of~$y$,
and $\theta=\theta(t,x)$ is an arbitrary solution of the equation $\theta_t+\theta_{xx}+2L_x\theta=0$.

Analogously, every solution of the system~\eqref{eq:BLPsystemYurovForm} with the differential constraint $u_y=q_y$
admits the representation $u=\Phi_x/\Phi$ and $q=\Phi_x/\Phi+L$,
where $L=L(t,x)$ is an arbitrary sufficiently smooth function of~$(t,x)$,
and $\Phi=\Phi(t,x,y)$ is an arbitrary nonzero solution of the equation $\Phi_t-\Phi_{xx}-2L_x\Phi=0$.
The system~\eqref{eq:BLPCoveringYurov} with these values of~$u$ and~$q$ reduces,
via integrating the first equation with respect to~$y$ and recombining the equations,
to the system
\begin{gather}\label{eq:BLPCoveringYurovWithUy=QyReduced}
\psi_x=-u\psi+\theta,\quad
\psi_t=(u^2-u_x)\psi+\theta_x-\theta u,
\end{gather}
where $\theta=\theta(t,x)$ is the parameter function arising in the course of the integration.
The consistency of the last system with respect to~$\psi$ implies
that the function~$\theta$ satisfies the equation $\theta_t+\theta_{xx}+2L_x\theta=0$.
The further integration of the system~\eqref{eq:BLPCoveringYurovWithUy=QyReduced} leads to
the following representation for its solutions:
\begin{gather*}
\psi=\frac1\Phi\smash{\int_{x_0}^x\theta(t,x')\Phi(t,x',y)\,{\rm d}x'}+\frac1\Phi\smash{\int_{t_0}^t(\theta\Phi_x-\theta_x\Phi)\big|_{(t',x_0,y)}{\rm d}t'}+\frac\zeta\Phi,
\end{gather*}
where the description of the functions~$\zeta$ and~$\theta$ is the same as above.

The constructed representations for the solutions of systems of the form~\eqref{eq:BLPCoveringYurov} imply that
the solution set of the system~\eqref{eq:BLPsystemYurovForm} with the differential constraint $q_y=0$ (resp.\ $u_y=q_y$)
is preserved by the Darboux transformations of the first (resp.\ second) kind,
and the application of the iterated Darboux transformations of the second (resp.\ first) kind
gives two series of solution families of the system~\eqref{eq:BLPsystemYurovForm},
\begin{gather*}
\solution (\tilde u,\tilde q)\mbox{ of the form~\eqref{eq:BLPnTimesIteratedDT2},\quad where}\quad
u=-\smash{\frac{\Phi_x}\Phi},\quad q=L,\\
\varphi^i=\frac1\Phi\int_{y_0}^y\zeta^i(y')\Phi(t,x,y')\,{\rm d}y'+\theta^i(t,x),\quad i=1,\dots,n,\\
\mbox{$\Phi=\Phi(t,x,y)$ is an arbitrary nonzero solution of the equation $\Phi_t+\Phi_{xx}+2L_x\Phi=0$};
\\[1.5ex]
\solution (\tilde u,\tilde q)\mbox{ of the form~\eqref{eq:BLPnTimesIteratedDT1},\quad where}\quad
u=\frac{\Phi_x}\Phi,\quad q=\frac{\Phi_x}\Phi+L,\\[.5ex]
\varphi^i=\frac1\Phi\int_{x_0}^x\theta^i(t,x')\Phi(t,x',y)\,{\rm d}x'
+\frac1\Phi\int_{t_0}^t(\theta^i\Phi_x-\theta^i_x\Phi)\big|_{(t',x_0,y)}{\rm d}t'+\frac{\zeta^i}\Phi,\quad
i=1,\dots,n,\\
\mbox{$\Phi=\Phi(t,x,y)$ is an arbitrary nonzero solution of the equation $\Phi_t-\Phi_{xx}-2L_x\Phi=0$}.
\end{gather*}
In both the series, $\zeta^i=\zeta^i(y)$ are arbitrary sufficiently smooth functions of~$y$,
and $\theta^i=\theta^i(t,x)$ are arbitrary solutions of the equation $\theta_t+\theta_{xx}+2L_x\theta=0$
such that ${\rm W}(\varphi^1,\dots,\varphi^n)\ne0$ and ${\rm W}(\varphi^1_y,\dots,\varphi^n_y)\ne0$
for the first and second series, respectively,
and for the first series in addition $q_y\ne A^{n1}_y=-\p_x\p_y\ln\left|{\rm W}(\varphi^1,\dots,\varphi^n)\right|$.

\section{Conclusion}\label{sec:Conclusion}

In the course of proving Theorem~\ref{thm:PointSymPseudogroupOfBLPSystem},
we have computed the point-symmetry pseudogroup~$G$ of the system~\eqref{eq:BLPsystem}
using an original version of the algebraic method, which is based on pushing forward,
instead of the maximal Lie invariance algebra~$\mathfrak g$ of~\eqref{eq:BLPsystem},
its proper subalgebra~$\mathfrak s$.
In contrast to the algebra~$\mathfrak g$, which is infinite-dimensional,
the subalgebra~$\mathfrak s$ is finite-dimensional
and, moreover, its dimension $\dim\mathfrak s=9$ is not too great.
Since the method proved very efficient,
we have additionally analyzed certain aspects related to it in
Remark~\ref{rem:FurtherUseOfAlgMethod} and Proposition~\ref{pro:RespectingMegaideal}.
In particular, it was unexpected  that the conditions $\Phi_*\mathfrak g\subseteq\mathfrak g$  and $\Phi_*\mathfrak s\subseteq\mathfrak g$
for a point transformation~$\Phi$ are equivalent.

In the context of Remark~\ref{rem:FurtherUseOfAlgMethod} and Proposition~\ref{pro:RespectingMegaideal},
the consideration of Section~\ref{PointSymPseudogroup} has led to
several interesting questions
within the theory underlying the algebraic method of finding point-symmetry (pseudo)groups of systems of differential equations.
Let $\mathfrak g$ be the maximal Lie invariance algebra of a system~$\mathcal L$ of differential equations.
For a proper subalgebra~$\mathfrak s$ of $\mathfrak g$, consider the following property:
\begin{description}
\item[(P)]
The conditions $\Phi_*\mathfrak g\subseteq\mathfrak g$ and $\Phi_*\mathfrak s\subseteq\mathfrak g$
for point transformations~$\Phi$ in the space coordinatized by the independent and dependent variables of the system~$\mathcal L$ are equivalent.
\end{description}
What are (necessary or sufficient) conditions on~$\mathcal L$
for the existence of a proper subalgebra of $\mathfrak g$ with the property~(P)?
This question is especially relevant for systems with infinite-dimensional maximal Lie invariance algebras
and can even be strengthened for this particular case.
If~\mbox{$\dim\mathfrak g=\infty$}, does there necessarily exist a finite-dimensional subalgebra of $\mathfrak g$
with the property~(P)?
How can one estimate the dimensions of such subalgebras?
Can such subalgebras be found algorithmically?
Is there an easy method for selecting, among such subalgebras, a subalgebra of minimal dimension?

One more question -- about respecting megaideals -- can be formulated in a more general framework.
Given a Lie algebra~$\mathfrak g$ of vector fields on a (finite-dimensional) manifold~$M$,
a family $\{\mathfrak i_\gamma\mid\gamma\in\Gamma\}$ of megaideals of~$\mathfrak g$ with some index set~$\Gamma$
and $\mathfrak i_{\gamma_0}=\mathfrak g$ for some $\gamma_0\in\Gamma$,
and a proper subalgebra~$\mathfrak s$ of~$\mathfrak g$ with the property
\[
\{\Phi\in{\rm Diff}(M)\mid\forall\gamma\in\Gamma\colon\Phi_*(\mathfrak s\cap\mathfrak i_\gamma)\subseteq\mathfrak i_\gamma\}
=\{\Phi\in{\rm Diff}(M)\mid\Phi_*\mathfrak g\subseteq\mathfrak g\}.
\]
The question is whether all $\Phi\in{\rm Diff}(M)$ with $\Phi_*\mathfrak s\subseteq\mathfrak g$
respect the megaideals~$\mathfrak i_\gamma$, $\gamma\in\Gamma$, i.e.,
$\Phi_*(\mathfrak s\cap\mathfrak i_\gamma)\subseteq\mathfrak i_\gamma$ for all $\gamma\in\Gamma$,
and thus $\Phi_*\mathfrak i_\gamma\subseteq\mathfrak i_\gamma$ for all $\gamma\in\Gamma$.

Another interesting result of the present paper is the identification of
a number of famous partial and ordinary differential equations
among reductions of the Boiti--Leon--Pempinelli system~\eqref{eq:BLPsystem}
using differential constraints or Lie symmetries.
In Calogero's terminology~\cite{calo2017a},
some of these reductions are S-integrable, which means the solvability via the scattering transform,
but most of them are C-integrable, i.e., they are linearized by a change of variables,
although the original system~\eqref{eq:BLPsystem} is not C-integrable.
The Riccati (Hopf--Cole) transformation is especially relevant for reductions of this system.
Using various differential constraints, we have reduced the system~\eqref{eq:BLPsystem} to
\begin{itemize}\itemsep=0ex
\item[$\circ$]
the classical Burgers equation and, more generally,
general inhomogeneous Burgers equations and inhomogeneous potential Burgers equations,
which are respectively linearized by the Hopf--Cole transformation and by a point transformation
to (1+1)-dimensional linear heat equations with potentials,
\item[$\circ$]
the (1+1)-dimensional Liouville equation, which is linearized by a differential substitution to
the (1+1)-dimensional linear wave equation,
\item[$\circ$]
quite general families of Riccati equations, of Bernoulli equations and of Abel equations of the first kind,
where one of the independent variables plays the role of an implicit parameter,
\item[$\circ$]
ordinary differential equation for Weierstrass elliptic functions and
the second, third, fourth and fifth Painlev\'e equations,
\item[$\circ$]
the $\sinh$-Gordon and $\cosh$-Gordon equations, both of which are S-integrable,
\item[$\circ$]
the system of dispersive long-wave equations, which is S-integrable as well.
\end{itemize}

In the literature, there are a number of papers devoted to finding exact solutions of the system~\eqref{eq:BLPsystem}
by various methods,
where all the constructed solutions satisfy the differential constraint $u_y=v_x$,
and thus they can be easily derived from solutions of (1+1)-dimensional linear heat equations with potentials
via the two-dimensional version of the Hopf--Cole transformation.
In many papers, attempts to construct exact invariant solutions were made using weakened versions of the method of Lie reductions,
and again all found solutions turn out $G$-equivalent to simple members of families of solutions
presented in Section~\ref{sec:SolutionsByMethodOfDiffConstraints}.
In other words, these solutions satisfy one of the simplest differential constraints
that are consistent with the system~\eqref{eq:BLPsystem}.
Both the above two sets of papers include \cite{zhao2017a},
and we would not like to cite more such papers.
Moreover all or at least substantial part of solutions presented in them are not solutions at all, 
and other solutions are just different representations of known ones. 
This resulted from making common errors in finding exact solutions of nonlinear differential equations,
which were analyzed in~\cite{kudr2009a,popo2010c}. 
Multiple flaws of one of such papers were discussed in~\cite{kudr2011b}.
\looseness=1

As we have shown in Section~\ref{sec:LaplaceAndDarbouxTrans},
new large families of exact solutions of the Boiti--Leon--Pempinelli system
can efficiently be generated from known solutions of this system via
iteratively applying the forward and inverse Laplace transformations~\eqref{eq:BLPLaplaceTrans} and~\eqref{eq:BLPLaplaceTransInverse} or
the Darboux transformations~\eqref{eq:BLPDT1} and~\eqref{eq:BLPDT2}.
In particular, in Section~\ref{sec:DarbouxTrans}
we have constructed huge solution families of the Boiti--Leon--Pempinelli system.
Each of these families is parameterized by 
\begin{itemize}\itemsep=1ex
\item[$\circ$]
an arbitrary solution $\Phi=\Phi(t,x,y)$
of the (1+1)-dimensional linear heat equation with an arbitrary potential depending on $(t,x)$,
where~$t$ and~$x$ are the only genuine independent variables and $y$ plays the role of an implicit parameter,
\item[$\circ$]
an arbitrary (finite) number of functions of $(t,x)$ that are
arbitrary solutions of the same equation (resp.\ its adjoint), and
\item[$\circ$]
an arbitrary (finite) number of arbitrary sufficiently smooth functions of~$y$.
\end{itemize}

In group analysis of differential equations,
it is often required to classify low-dimensional subalgebras of infinite-dimensional Lie algebras of vector fields.
Nevertheless, there are not too many papers in which such classifications have correctly been carried out.
The main kinds of problems, where such classifications are necessary solution steps,
are given by classifications of Lie reductions of systems of differential equations with infinite-dimensional Lie invariance algebras
\cite{andr1998A,bihl2009a,cham1988a,davi1986a,fush1994a,fush1994b,mart1989a,opan2020a,paqu1990a}
and by group classifications of classes of such systems whose equivalence algebras are infinite-dimensional,
see \cite{bihl2012b,vane2020b} and references therein.
In the present paper, we have completely solved the problem of the former kind
for the Boiti--Leon--Pempinelli system~\eqref{eq:BLPsystem}.

We have also exhaustively integrated all inequivalent essential reduced systems of ordinary differential equations 
for the system~\eqref{eq:BLPsystem} 
in terms of the second, third, fourth and fifth Painlev\'e transcendents, 
Weierstrass $\wp$- and $\zeta$-functions and elementary functions, 
and the way of integration is nontrivial. 
This result well agrees with the known general observation~\cite{ablo1977b} on 
a connection between nonlinear partial differential equations 
that can be solved by the inverse scattering transform 
and nonlinear ordinary differential equations without movable critical points. 
So deep and complete analysis of all systems of ordinary differential equations 
related to a particular multidimensional system of partial differential equations 
with the framework of Lie reduction is rather exceptional in the literature. 
Similar studies were carried out 
for the (much simpler) three-wave resonant interaction system in~\cite{mart1989a} 
and, partly, for the system~\eqref{eq:BLPsystemPWform} in~\cite{paqu1990a}.

\appendix

\section[Integration of ordinary differential equations related to elliptic functions]
{Integration of ordinary differential equations\\ related to elliptic functions}\label{sec:IntegrationOfODEsRelatedToEllipticFunctions}

The general solution of the autonomous ordinary differential equation
\begin{gather}\label{eq:WeierstrassODEs}
\psi_z{}^{\!2}=P(\psi):=4\psi^3-\textsl{g}_2\psi-\textsl{g}_3
\end{gather}
for the unknown function $\psi=\psi(z)$ (over~$\mathbb C$)
with constant (complex) coefficients~$\textsl{g}_2$ and~$\textsl{g}_3$ is expressed
in terms of the (doubly periodic) Weierstrass elliptic function~$\wp(z;\textsl{g}_2,\textsl{g}_3)$
if and only if all the roots of the polynomial $P$ are simple,
$\psi=\wp(z+C;\textsl{g}_2,\textsl{g}_3)$, where $C$ is an arbitrary constant.
If the polynomial $P$ has a multiple root, i.e., its discriminant~$\mathop{\rm Disc}_\psi(P)$ is zero,
\[\Delta:=\mathop{\rm Disc}\nolimits_\psi(P)/16=\textsl{g}_2{}^{\!3}-27\textsl{g}_3{}^{\!2}=0,\]
then at least one of the two solution periods degenerates,
and the general solution of~\eqref{eq:WeierstrassODEs} is expressed in terms of elementary functions.
(This degeneration can be considered as the degeneration of the Weierstrass elliptic functions.)
In particular, the general solution of~\eqref{eq:WeierstrassODEs} with $\textsl{g}_2=\textsl{g}_3=0$
is $\psi=(z+C)^{-2}$, and thus, roughly speaking, both the periods are infinite.
See, e.g.,~\cite[Section~4.16]{zhur1941A} for more details. 

Consider the ordinary differential equations of the form
\begin{equation}\label{eq:EllipticODEs}
\varphi_z{}^{\!2}=F(\varphi):=a_0\varphi^4+4a_1\varphi^3+6a_2\varphi^2+4a_3\varphi+a_4,
\end{equation}
where $a_0$, \dots, $a_4$ are (complex) constants with $a_0\ne0$.

If the polynomial~$F$ has a simple root~$\lambda$,
then the equation~\eqref{eq:EllipticODEs} is reduced using the fractional linear change
\[
\varphi=\frac l{4\psi-m}+\lambda,\quad\mbox{where}\quad
l:=F_\varphi(\lambda),\quad
m:=\frac16F_{\varphi\varphi}(\lambda),
\]
of the dependent variable~$\varphi$ by $\psi$ \cite[Section~20.6]{whit1927A}
to the equation of the form~\eqref{eq:WeierstrassODEs},
where $\textsl{g}_2$ and $\textsl{g}_3$ are the invariants of the quartic polynomial~$F$ that are given by
\begin{gather}\label{eq:InvsOfQuarticPolynomial}
\textsl{g}_2=a_0a_4-4a_1a_3+3a_2{}^{\!2},\quad
\textsl{g}_3
=\begin{vmatrix}a_0&a_1&a_2\\a_1&a_2&a_3\\a_2&a_3&a_4\end{vmatrix}.
\end{gather}
Since $\mathop{\rm Disc}_\varphi(F)=16\mathop{\rm Disc}_\psi(P)=256\Delta$,
the polynomial $F$ has multiple roots if and only if the corresponding polynomial $P$ does.
In other words, the degenerate cases of equations of the form~\eqref{eq:EllipticODEs}
with~$F$ having simple roots
correspond to those of equations of the form~\eqref{eq:WeierstrassODEs}.

Therefore, if all the roots of the polynomial~$F(\varphi)$ are simple,
the general solution of the equation~\eqref{eq:EllipticODEs}
can be represented via the Weierstrass elliptic functions \cite[Section~20.6]{whit1927A},
\[
\varphi=\frac14F_\varphi(\lambda)\left(\wp(z+C;\textsl{g}_2,\textsl{g}_3)-\frac1{24}F_{\varphi\varphi}(\lambda)\right)^{-1}+\lambda,
\]
where~$\lambda$ is a (simple) root of~$F$,
$C$ is an arbitrary constant,
and the constant parameters~$\textsl{g}_2$ and~$\textsl{g}_3$ are the invariants~\eqref{eq:InvsOfQuarticPolynomial} of~$F$.

For our purpose, we integrate the degenerate equation~\eqref{eq:EllipticODEs} over~$\mathbb R$
for all the real values of the coefficients~$a_0$,~\dots,~$a_4$, neglecting the condition $a_0\ne0$.

If the polynomial~$F$ has a multiple real root~$\lambda$, then the equation~\eqref{eq:EllipticODEs} can be represented~as
\begin{gather}\label{eq:EllipticODEsWithMultipleRealRoot}
\varphi_z{}^{\!2}=(\varphi-\lambda)^2(\hat a\varphi^2+\hat b\varphi+\hat c)
\end{gather}
with $(\hat a,\hat b,\hat c):=(a_0,\,2a_0\lambda+4a_1,\,3a_0\lambda^2+8a_1\lambda+6a_2)$,
and the equalities 
\[\hat c\lambda^2=a_4,\quad 2\hat c\lambda=2a_0\lambda^3+4a_1\lambda^2-4a_3\]
follow from the multiplicity of the root~$\lambda$.
To conveniently integrate the last ordinary differential equation,
we make the change $\tilde\varphi=(\varphi-\lambda)^{-1}$, which leads to the equation
\[
\tilde\varphi_z{}^{\!2}=c\tilde\varphi^2+b\tilde\varphi+a,
\]
where \
$a:=\hat a=a_0$, \
$b:=\hat b+2\hat a\lambda=4(a_0\lambda+a_1)$, \
$c:=\hat c+\hat b\lambda+\hat a\lambda^2=6(a_0\lambda^2+2a_1\lambda+a_2)$.
After integrating the last equation, 
we return to the old dependent variable~$\varphi$.
As a result,
in addition to the obvious constant solutions,
$\varphi=\nu$, where $\nu$ is an arbitrary root of~$F$ if $(c,b,a)\ne(0,0,0)$,
or $\varphi$ is an arbitrary constant if $c=b=a=0$,
we obtain the following cases for the general solutions
of the degenerate ordinary differential equation~\eqref{eq:EllipticODEsWithMultipleRealRoot}:
\begin{gather*}
\varphi=\frac{\ve a^{-1/2}}{z+C}+\lambda \quad\mbox{if}\quad c=b=0,\ a\ne0,\ \mbox{and thus}\ a>0,
\\
\varphi=\frac{4b}{b^2(z+C)^2-4a}+\lambda \quad\mbox{if}\quad c=0,\ b\ne0,
\\
\varphi=\frac{4c}{4C{\rm e}^{\varepsilon\sqrt{c}\,z}+DC^{-1}{\rm e}^{-\varepsilon\sqrt{c}\,z}-2b}+\lambda \quad\mbox{if}\quad c>0
\quad (\mbox{here}\ C\ne0),
\\
\varphi=\frac{2c}{\varepsilon\sqrt{D}\sin(\sqrt{-c}\,z+C)-b}+\lambda \quad\mbox{if}\quad c<0,\ D>0,
\end{gather*}
where $D:=b^2-4ac$, $C$ is an arbitrary constant, and $\varepsilon=\pm1$.

One more case of degenerating equations of the form~\eqref{eq:EllipticODEs} over~$\mathbb R$
is associated with $F$ having a multiple root from~$\mathbb C\setminus\mathbb R$.
Then the multiplicity of this root is equal to two, and its conjugate is a root of~$F$ as well.
Then the polynomial~$F$ can be represented in the form $F=(a\varphi^2+b\varphi+c)^2$
for some real constants $a$, $b$ and~$c$ with $D:=b^2-4ac<0$.
The corresponding general solution~is
\[
\varphi=\frac{\ve\sqrt{|D|}}{2a}\tan\bigg(\frac{\sqrt{|D|}}2z+C\bigg)-\frac b{2a}.
\]

\section{Conservation laws}\label{sec:BLPsystemCLs}

When studying cosymmetries%
\footnote{%
Given a system~$\mathcal L$ of differential equation, 
a cosymmetry of~$\mathcal L$ is a differential function \cite[p.~288]{olve1993A}
that satisfies the adjoint \cite[Eq.~(5.83)]{olve1993A}
to the condition for generalized symmetries of~$\mathcal L$ \cite[Definition~5.2]{olve1993A}.
Cosymmetries are also called \emph{adjoint symmetries} \cite{sarl1987a},
see \cite[p.~107]{blum2010A} for further references
A characteristic of a conservation law of~$\mathcal L$ is necessarily 
a cosymmetry of~$\mathcal L$ but the converse is not true.
}
of the system~\eqref{eq:BLPsystem},
in view of their equivalence on the solution set of~\eqref{eq:BLPsystem}
it suffices to consider only cosymmetries that do not depend
on the derivatives of~$u$ simultaneously involving differentiations with respect to~$t$ and~$y$ and
on the derivatives of~$v$ involving differentiations with respect to~$t$.
We call such cosymmetries \emph{reduced}.
Computing with the package {\tt Jets}~\cite{BaranMarvan},
which is based on theoretical results of~\cite{marv2009a},
we checked that the reduced cosymmetries at least up to order four for the system~\eqref{eq:BLPsystem}
are exhausted by the pairs of differential functions of $(u,v)$ of the form
$\Lambda^0(h^0)+\Lambda^1(h^1)+\Lambda^2(h^2)+\Lambda^4(\mu)+\Lambda^5(\nu)$,
where
\begin{gather*}
\Lambda^0(h^0)=\big(h^0,0\big),\quad
\Lambda^1(h^1)=\big(h^1x,0\big),\quad
\Lambda^2(h^2)=\big(h^2x^2,0\big),\\
\Lambda^4(\mu)=\big(\mu,0\big),\quad
\Lambda^5(\nu)=\big(\nu v,\nu u_y\big),
\end{gather*}
and the parameter functions $h^i=h^i(t)$, $i=0,1,2$, $\mu=\mu(y)$ and $\nu=\nu(y)$
run through the sets of sufficiently smooth functions of their arguments.
Each of these cosymmetries is a conservation-law characteristic of the system~\eqref{eq:BLPsystem}.

We conjecture that 
\emph{the above cosymmetries
constitute the entire space of reduced cosymmetries of this system
and thus the entire space of its reduced conservation-law characteristics.}
The associated conservation laws of~\eqref{eq:BLPsystem} contain the following conserved currents:
\begin{gather*}
\mathcal F^0(h^0)=\big(0,\,-2h^0v_{xx},\,h^0(u_t-2uu_x+u_xx)\big),\\
\mathcal F^1(h^1)=\big(0,\,-2h^1(xv_{xx}-v_x),\,h^1x(u_t-2uu_x+u_xx)\big),\\
\mathcal F^2(h^2)=\big(0,\,-2h^2(x^2v_{xx}-2xv_x+2v),\,h^2x^2(u_t-2uu_x+u_xx)\big),\\
\mathcal F^4(\mu)=\big(\mu u_y,\,\mu(u_{xy}-2uu_y-2v_{xx}),\,0\big),\\
\mathcal F^5(\nu)=\big(\nu vu_y,\,\nu(vu_{xy}-u_yv_x-2vuu_y-2vv_{xx}+v_x^{\,\,2}),\,0\big).
\end{gather*}
As representatives of the conservation laws as equivalence classes of conserved currents,
we choose the ``short'' conserved currents, each of which have only two nonvanishing components.
Moreover, selected conserved currents do not involve derivatives of the parameter functions (of order one and higher).

\section*{Acknowledgments}

The authors are grateful to Vyacheslav Boyko, Michael Kunzinger, Dmytro Popovych and Galyna Popovych for helpful discussions and suggestions.
The research of ROP was supported by the Austrian Science Fund (FWF), projects P28770 and~P30233.
This research was also supported in part by the Ministry of Education, Youth and Sports of the Czech Republic (M\v SMT \v CR)
under RVO funding for I\v C47813059.
The authors express deepest thanks to the Armed Forces of Ukraine and the civil Ukrainian people
for their bravery and courage in defense of peace and freedom in Europe and in the entire world from russism.


\footnotesize
\begin{thebibliography}{10}

\bibitem{ablo1977b}
Ablowitz M.J. and Segur H.,
Exact linearization of a Painlev\'e transcendent,
{\it Phys. Rev. Lett.} {\bf 38} (1977), 1103--1106.

\bibitem{abra2008a}
Abraham-Shrauner B. and Govinder K.S.,
Master partial differential equations for a type II hidden symmetry,
{\it J.~Math. Anal. Appl.} {\bf 343} (2008), 525--530.

\bibitem{andr1998A}
Andreev V.K., Kaptsov O.V., Pukhnachov V.V. and Rodionov A.A.,
\emph{Applications of group-theoretical methods in hydrodynamics},
Kluwer,  Dordrecht, 1998.

\bibitem{BaranMarvan}
Baran H. and Marvan M.,
Jets. A software for differential calculus on jet spaces and diffieties,
\newline \texttt{http://jets.math.slu.cz}

\bibitem{blum2010A}
Bluman G.W., Cheviakov A.F. and Anco S.C.,
{\it Applications of symmetry methods to partial differential equations},
Springer, New York, 2010. 

\bibitem{blum1969a}
Bluman~G.W. and Cole~J.D.,
The general similarity solution of the heat equation,
{\it J.~Math. Mech.} {\bf 18} (1969), 1025--1042.

\bibitem{bihl2015a}
Bihlo A., Dos Santos Cardoso-Bihlo E.M. and Popovych R.O.,
Algebraic method for finding equivalence groups,
{\it J. Phys.: Conf. Ser.} {\bf 621} (2015) 012001, arXiv:1503.06487. 

\bibitem{bihl2012b}
Bihlo A., Dos Santos Cardoso-Bihlo E. and Popovych R.O.,
Complete group classification of a class of nonlinear wave equations,
{\it J.~Math. Phys.} {\bf 53} (2012), 123515, arXiv:1106.4801.

\bibitem{bihl2011c}
Bihlo A. and Popovych R.O.,
Point symmetry group of the barotropic vorticity equation,
in \emph{Proceedings of 5th Workshop ``Group Analysis of Differential Equations \& Integrable Systems'' (June 6--10, 2010, Protaras, Cyprus)},
University of Cyprus, Nicosia, 2011, pp. 15--27, arXiv:1009.1523.

\bibitem{bihl2009a}
Bihlo A. and Popovych R.O.,
Lie symmetries and exact solutions of the barotropic vorticity equation,
{\it J.~Math. Phys.} {\bf 50} (2009), 123102, arXiv:0902.4099.

\bibitem{boit1987a}
Boiti M., Leon J.J.P. and Peminelli F.,
Integrable two-dimensional generalisation of the sine- and sinh-Gordon equations,
{\it Inverse Problems} {\bf 3} (1987), 37--49.

\bibitem{boyk2016a}
Boyko V.M., Kunzinger M. and Popovych R.O.,
Singular reduction modules of differential equations,
{\it J.~Math. Phys.} {\bf 57} (2016) 101503, arXiv:1201.3223.

\bibitem{bure1964b}
Bureau F.J.,
Differential equations with fixed critical point. II,
{\it Ann. Mat. Pura Appl.\ (4)} {\bf 66} (1964), 1--116.

\bibitem{bure1972a}
Bureau F.J.,
\'{E}quations diff\'{e}rentielles du second ordre en $Y$ et du second degr\'{e} en $\ddot Y$ dont l'int\'{e}grale g\'{e}n\'{e}rale est \`a points critiques fixes,
{\it Ann. Mat. Pura Appl.\ (4)} {\bf 91} (1972), 163--281.

\bibitem{calo2017a}
Calogero F.,
New C-integrable and S-integrable systems of nonlinear partial differential equations,
{\it J.~Nonlinear Math. Phys.} {\bf 24} (2017), 142--148.

\bibitem{cham1988a}
Champagne B. and Winternitz P.,
On the infinite-dimensional symmetry group of the Davey--Stewartson equations,
{\it J.~Math. Phys.} {\bf 29} (1988), 1--8.

\bibitem{clar2006a}
Clarkson P.A.,
Painlev\'e equations -- nonlinear special functions,
in W. Van Assche, F. Marcellan (eds.), {\it Orthogonal polynomials and special functions},
Springer, Berlin, 2006, pp.~331--411.

\bibitem{cosg2000a}
Cosgrove C.M.,
Chazy classes IX--XI of third-order differential equations,
{\it Stud. Appl. Math.} {\bf 104} (2000), 171--228.

\bibitem{cosg2006b}
Cosgrove C.M.,
Chazy's second-degree Painlev\'e equations,
{\it J.~Phys.~A} {\bf 39} (2006), 11955--11971.

\bibitem{davi1986a}
David D., Kamran N., Levi D. and Winternitz P.,
Symmetry reduction for the Kadomtsev--Petviashvili equation using a loop algebra,
{\it J. Math. Phys.} {\bf 27} (1986), 1225--1237.

\bibitem{card2011a}
Dos Santos Cardoso-Bihlo E., Bihlo A. and Popovych R.O.,
Enhanced preliminary group classification of a class of generalized diffusion equations,
{\it Commun. Nonlinear Sci. Numer. Simulat.} {\bf 16} (2011), 3622--3638, arXiv:1012.0297.

\bibitem{card2013a}
Dos Santos Cardoso-Bihlo E. and Popovych R.O.,
Complete point symmetry group of the barotropic vorticity equation on a rotating sphere,
{\it J.~Engrg.\ Math.} {\bf 82} (2013), 31--38, arXiv:1206.6919.

\bibitem{fush1994a}
Fushchych W. and Popowych R.,
Symmetry reduction and exact solutions of the Navier--Stokes equations.~I,
{\it J.~Nonlinear Math. Phys.} {\bf 1} (1994), 75--113, arXiv:math-ph/0207016.

\bibitem{fush1994b}
Fushchych W. and Popowych R.,
Symmetry reduction and exact solutions of the Navier--Stokes equations.~II,
{\it J.~Nonlinear Math. Phys.} {\bf 1} (1994), 158--188, arXiv:math-ph/0207016.

\bibitem{fush1993A}
Fushchych W.I., Shtelen W.M. and Serov N.I.,
{\it Symmetry analysis and exact solutions of equations of nonlinear mathematical physics},
Dordrecht, Kluwer Academic Publishers, 1993.

\bibitem{gara1994a}
Garagash T.I.,
Modification of the Painleve test for systems of nonlinear partial differential equations,
{\it Theoret. and Math. Phys.} {\bf 100} (1994), 1075--1081.

\bibitem{gard1967a}
Gardner C.S., Greene J.M., Kruskal M.D. and  Miura R.M.,
Method for solving the Korteweg--de Vries equation,
{\it Phys. Rev. Lett.} {\bf 19} (1967), 1095--1097.

\bibitem{hilg2011A}
Hilgert J. and Neeb K.H.,
{\it Structure and geometry of Lie groups},
Springer, New York, 2012.

\bibitem{hydo1998a}
Hydon P.E.,
Discrete point symmetries of ordinary differential equations,
{\it Proc. R. Soc. Lond. Ser. A Math. Phys. Eng. Sci.} {\bf 454} (1998), 1961--1972.

\bibitem{hydo2000b}
Hydon P.E.,
How to construct the discrete symmetries of partial differential equations,
{\it Eur. J. Appl. Math.} {\bf 11} (2000), 515--527.

\bibitem{ivan2008b}
Ivanova N.M.,
Exact solutions of diffusion--convection equations,
{\it Dyn. Partial Differ. Equ.} {\bf 5} (2008), 139--171, arXiv:0710.4000.

\bibitem{kudr2009a}
Kudryashov N.A.,
Seven common errors in finding exact solutions of nonlinear differential equations,
{\it Commun. Nonlinear Sci. Numer. Simulat.} {\bf 14} (2009), 3507--3529.

\bibitem{kudr2011b}
Kudryashov N.A.,
Redundant exact solutions of nonlinear differential equations,
{\it Commun. Nonlinear Sci. Numer. Simul.} {\bf 16} (2011), 3451--3456.

\bibitem{kunz2008b}
Kunzinger M. and Popovych R.O.,
Singular reduction operators in two dimensions,
{\it J. Phys. A: Math. Theor.} {\bf 41} (2008), 505201, arXiv:0808.3577.

\bibitem{kunz2009a}
Kunzinger M. and Popovych R.O.,
Is a nonclassical symmetry a symmetry?,
in {\it Proceedings of Fourth Workshop ``Group Analysis of Differential Equations and Integrable Systems'' (October 26--30, 2008, Protaras, Cyprus)},
University of Cyprus, Nicosia, 2009, pp. 107--120, arXiv:0903.0821.

\bibitem{kunz2011a}
Kunzinger M. and Popovych R.O.,
Generalized conditional symmetries of evolution equations,
{\it J. Math. Anal. Appl.} {\bf 379} (2011), 444--460, arXiv:1011.0277.

\bibitem{kupe1985b}
Kupershmidt B.A.,
Mathematics of dispersive water waves,
{\it Comm. Math. Phys.} {\bf 99} (1985), 51--73.

\bibitem{mart1989a}
Martina L. and Winternitz P.,
Analysis and applications of the symmetry group of the multidimensional three-wave resonant interaction problem, 
{\it Ann. Physics} {\bf 196} (1989), 231--277.

\bibitem{marv2009a}
Marvan M.,
Sufficient set of integrability conditions of an orthonomic system,
{\it Found. Comput. Math.} {\bf 9} (2009), 651--674, arXiv:nlin/0605009.

\bibitem{matv1991A}
Matveev V.B. and Salle M.A.,
{\it Darboux transformations and solitons},
Springer-Verlag, Berlin, 1991.

\bibitem{olve1993A}
Olver P.J.,
{\it Application of Lie groups to differential equations},
Springer, New York, 2000.

\bibitem{olve1986b}
Olver P.J. and Rosenau P.,
The construction of special solutions to partial differential equations,
{\it Phys. Lett.~A} {\bf 114} (1986), 107--112.

\bibitem{opan2020a}
Opanasenko S., Bihlo A., Popovych R.O. and Sergyeyev A.,
Extended symmetry analysis of isothermal no-slip drift flux model,
{\it Phys. D} {\bf 402} (2020), 132188, arXiv:1705.09277. 

\bibitem{ovsi1982A}
Ovsiannikov L.V.,
{\it Group analysis of differential equations},
Acad. Press, New York, 1982.

\bibitem{paqu1990a}
Paquin G. and Winternitz P.,
Group theoretical analysis of dispersive long wave equations in two space dimensions,
{\it Phys.~D} {\bf 46} (1990), 122--138.

\bibitem{pate1975a}
Patera J., Winternitz P. and Zassenhaus H.,
Continuous subgroups of the fundamental groups of physics. I. General method and the Poincar\'e group,
{\it J.~Math. Phys.} {\bf 16} (1975), 1597--1614.

\bibitem{poch2013d}
Pocheketa O.A.,
Normalized classes of generalized Burgers equations,
in {\it Proceedings of the Sixth International Workshop ``Group Analysis of Differential Equations and Integrable Systems'' (Protaras, Cyprus, June 17--21, 2012)},
University of Cyprus, Nicosia, 2013, pp. 170--178, arXiv:1605.04077.

\bibitem{poch2017a}
Pocheketa O.A. and Popovych R.O.,
Extended symmetry analysis of generalized Burgers equations,
{\it J.~Math. Phys.} {\bf 58} (2017), 101501, arXiv:1603.09377. 

\bibitem{poly2012A}
Polyanin A.D. and Zaitsev V.F.,
{\it Handbook of nonlinear partial differential equations}, second edition,
Chapman \& Hall/CRC, Boca Raton, FL, 2012.

\bibitem{popop1995a}
Popovych R.O.,
On the symmetry and exact solutions of a transport equation,
{\it Ukrainian Math.~J.} {\bf 47} (1995), 142--148.

\bibitem{popo2003a}
Popovych R.O., Boyko V.M., Nesterenko M.O. and Lutfullin M.W.,
Realizations of real low-dimensional Lie algebras,
{\it J.~Phys.~A} {\bf 36} (2003), 7337--7360, arXiv:math-ph/0301029. 

\bibitem{popo2008a}
Popovych R.O., Kunzinger M. and Ivanova N.M.,
Conservation laws and potential symmetries of linear parabolic equations,
{\it Acta Appl. Math.} {\bf 100} (2008), 113--185, arXiv:0706.0443.

\bibitem{popo2010c}
Popovych R.O. and Vaneeva O.O.,
More common errors in finding exact solutions of nonlinear differential equations. I,
{\it Commun. Nonlinear Sci. Numer. Simul.} {\bf 15} (2010), 3887--3899, arXiv:0911.1848.

\bibitem{pucc1992b}
Pucci E. and Saccomandi G.,
On the weak symmetry groups of partial differential equations,
{\it J. Math. Anal. Appl.} {\bf 163} (1992), 588--598.

\bibitem{pogr1996a}
Pogrebkov A.K. and Garagash T.I.,
On a solution of the Cauchy problem for the Boiti--Leon--Pempinelli equation,
{Theoret. and Math. Phys.} {\bf 109} (1996), 1369--1378.

\bibitem{reid1996a}
Reid G.J., Wittkopf A.D. and Boulton A.,
Reduction of systems of nonlinear partial differential equations to simplified involutive forms,
{\it Eur. J. Appl. Math.} {\bf 7} (1996), 604--635.

\bibitem{sarl1987a}
Sarlet W., Cantrijn F. and Crampin M.,
Pseudosymmetries, Noether's theorem and the adjoint equation,
{\it J.~Phys. A} {\bf 20} (1987), 1365--1376.

\bibitem{sido1984A}
Sidorov A.F., Shapeev V.P. and Yanenko N.N.,
{\it The method of differential constraints and its applications in gas dynamics},
``Nauka'' Sibirsk. Otdel., Novosibirsk, 1984 (in Russian).

\bibitem{vane2020b}
Vaneeva O.O., Bihlo A. and Popovych R.O.,
Generalization of the algebraic method of group classification with application to nonlinear wave and elliptic equations,
{\it Commun. Nonlinear Sci. Numer. Simul.} {\bf 91} (2020), 105419, arXiv:2002.08939.

\bibitem{vane2021a}
Vaneeva O.O., Popovych R.O. and Sophocleous C.,
Enhanced symmetry analysis of two-dimensional degenerate Burgers equation,
{\it J.~Geom. Phys.} {\bf 169} (2021), 104336, arXiv:1908.01877.

\bibitem{whit1927A}
Whittaker E.T. and Watson G.N,
{\it A course of modern analysis},
Cambridge University Press, Cambridge, 1996.

\bibitem{widd1975A}
Widder D.V.,
{\it The heat equation},
Pure and Applied Mathematics, vol.~67, Academic Press, New York--London, 1975.

\bibitem{wint2004a}
Winternitz P.,
Subalgebras of Lie algebras. Example of ${\rm sl}(3,\mathbb R)$,
in {\it Symmetry in physics}, {\it CRM Proc. Lecture Notes}, vol. 34, Amer. Math. Soc., Providence, RI, 2004,
pp. 215--227.

\bibitem{yane1964a}
Yanenko N.N.,
Compatibility theory and methods of integrating systems of nonlinear partial differential equations,
in {\it Proc. Fourth All-Union Mathematics Congress}, Nauka, Leningrad, 1964, pp. 247--259 (in Russian).

\bibitem{yuro1999a}
Yurov A.V.,
BLP dissipative structures in plane,
{\it Phys. Lett.~A} {\bf 262} (1999), 445--452.

\bibitem{zhao2017a}
Zhao Z., Han B.,
Lie symmetry analysis, B\"acklund transformations, and exact solutions of a (2+1)-dimensional Boiti--Leon--Pempinelli system,
{\it J.~Math. Phys.} {\bf 58} (2017), 101514.

\bibitem{zhda1995e}
Zhdanov R.Z.,
Conditional Lie--B\"acklund symmetry and reduction of evolution equations,
{\it J. Phys. A}, 1995, {\bf 28}, 3841--3850.

\bibitem{zhur1941A}
Zhuravsky A.M.,
{\it Handbook of Elliptic Functions},
 USSR Academy of Sciences Publishing House, Moscow--Leningrad, 1941 (in Russian).

\end{thebibliography}
\end{document}